\documentclass[10pt]{article}
\usepackage{amsfonts,amsmath,amssymb,amsthm}
\usepackage[pdftex]{graphicx}
\usepackage[hmargin=1in,vmargin=1in]{geometry}
\usepackage{natbib}
\usepackage[colorlinks,citecolor=blue]{hyperref}
\usepackage{setspace}
\usepackage{enumerate}
\usepackage{mathtools}
\usepackage[toc,title,titletoc,header]{appendix}
\usepackage{etoolbox} 
\usepackage{tabularx}
\usepackage{delarray}
\usepackage{dcolumn}
\usepackage{float}
\usepackage{booktabs}
\usepackage{adjustbox}
\usepackage{ragged2e}
\allowdisplaybreaks

\def\qed{\rule{2mm}{2mm}}
\def\independent{\perp \!\!\! \perp}
\usepackage{accents}
\newcommand{\ubar}[1]{\underaccent{\bar}{#1}}

\newtheorem{theorem}{Theorem}[section]
\newtheorem{lemma}{Lemma}[section]

\theoremstyle{definition}

\newtheorem{remark}{Remark}[section]
\newtheorem{assumption}{Assumption}[section]
\newtheorem{algorithm}{Algorithm}[section]

\AtEndEnvironment{remark}{~\qed}
\AtEndEnvironment{example}{~\qed}

\DeclareMathOperator*{\var}{Var}
\DeclareMathOperator*{\cov}{Cov}
\DeclareMathOperator*{\diag}{diag}
\DeclareMathOperator*{\argmin}{argmin}

\begin{document}
\hypersetup{pageanchor=false}
\author{
Yuehao Bai \\
Department of Economics\\
University of Southern California \\
\url{yuehao.bai@usc.edu}
\and
Liang Jiang \\
Fanhai International School of Finance\\
Fudan University \\
\url{jiangliang@fudan.edu.cn}
\and
Joseph P.\ Romano \\
Departments of Economics \& Statistics\\
Stanford University\\
\url{romano@stanford.edu}
\and
Azeem M.\ Shaikh\\
Department of Economics\\
University of Chicago \\
\url{amshaikh@uchicago.edu}
\and
Yichong Zhang \\
School of Economics \\
Singapore Management University \\
\url{yczhang@smu.edu.sg}
}

\bigskip

\title{Covariate Adjustment in Experiments with Matched Pairs\thanks{Yichong Zhang acknowledges the financial support from the NSFC under the grant No. 72133002 and a Lee Kong Chian fellowship. Any and all errors are our own.}}

\begin{spacing}{1.2}
\maketitle
\end{spacing}

\vspace{-0.3in}

\begin{spacing}{1.05}
\begin{abstract}
This paper studies inference on the average treatment effect (ATE) in experiments in which treatment status is determined according to ``matched pairs'' and it is additionally desired to adjust for observed, baseline covariates to gain further precision.   By a ``matched pairs'' design, we mean that units are sampled i.i.d.\ from the population of interest, paired according to observed, baseline covariates and finally, within each pair, one unit is selected at random for treatment.  Importantly, we presume that not all observed, baseline covariates are used in determining treatment assignment.  We study a broad class of estimators based on a ``doubly robust'' moment condition that permits us to study estimators with both finite-dimensional and high-dimensional forms of covariate adjustment.  We find that estimators with finite-dimensional, linear adjustments need not lead to improvements in precision relative to the unadjusted difference-in-means estimator.  This phenomenon persists even if the adjustments are interacted with treatment; in fact, doing so leads to no changes in precision.  However, gains in precision can be ensured by including fixed effects for each of the pairs.  Indeed, we show that this adjustment leads to the minimum asymptotic variance of the corresponding ATE estimator among all finite-dimensional, linear adjustments.  We additionally study an estimator with a regularized adjustment, which can accommodate high-dimensional covariates.  We show that this estimator leads to improvements in precision relative to the unadjusted difference-in-means estimator and also provide conditions under which it leads to the ``optimal'' nonparametric, covariate adjustment.  A simulation study confirms the practical relevance of our theoretical analysis, and the methods are employed to reanalyze data from an experiment using a ``matched pairs'' design to study the effect of macroinsurance on microenterprise.

\end{abstract}
\end{spacing}



\thispagestyle{empty}
\newpage
\setcounter{page}{1}
\hypersetup{pageanchor=true}

\section{Introduction}

This paper studies inference on the average treatment effect in experiments in which treatment status is determined according to ``matched pairs.''  By a ``matched pairs'' design, we mean that units are sampled i.i.d.\ from the population of interest, paired according to observed, baseline covariates and finally, within each pair, one unit is selected at random for treatment.  This method is used routinely in all parts of the sciences. Indeed, commands to facilitate its implementation are included in popular software packages, such as \texttt{sampsi} in Stata. References to a variety of specific examples can be found, for instance, in the following surveys of various field experiments: \cite{donner2000design}, \cite{glennerster2013running}, and \cite{rosenberger2015randomization}.  See also \cite{bruhn2009pursuit}, who, based on a survey of selected development economists, report that 56\% of researchers have used such a design at some point.  \cite{bai2021inference} develop methods for inference on the average treatment effect in such experiments based on the difference-in-means estimator.  In this paper, we pursue the goal of improving upon the precision of this estimator by exploiting observed, baseline covariates that are not used in determining treatment status.  

To this end, we study a broad class of estimators for the average treatment effect based on a ``doubly robust'' moment condition.  The estimators in this framework are distinguished via different ``working models'' for the conditional expectations of potential outcomes under treatment and control given the observed, baseline covariates.  Importantly, because of the double-robustness, these ``working models'' need not be correctly specified in order for the resulting estimator to be consistent.  In this way, the framework permits us to study both finite-dimensional and high-dimensional forms of covariate adjustment without imposing unreasonable restrictions on the conditional expectations themselves.  Under high-level conditions on the ``working models'' and their corresponding estimators and a requirement that pairs are formed so that units within pairs are suitably ``close'' in terms of the baseline covariates, we derive the limiting distribution of the covariate-adjusted estimator of the average treatment effect.  We further construct an estimator for the variance of the limiting distribution and provide conditions under which it is consistent for this quantity. 


Using our general framework, we first consider finite-dimensional, linear adjustments.  For this class of estimators, our main findings are summarized as follows.  First, we find that estimators with such adjustments are not guaranteed to be weakly more efficient than the unadjusted difference-in-means estimator. This finding echoes similar findings by \cite{yang2001efficiency} and \cite{tsiatis2008covariate} in settings in which treatment is determined by i.i.d.\ coin flips, and \cite{freedman2008regression} in a finite population setting in which treatment is determined according to complete randomization. See \cite{negi2021revisiting} for a succinct treatment of that literature. Moreover, we find that this phenomenon persists even if the adjustments are interacted with treatment. In fact, doing so leads to no changes in precision. In this sense, our results diverge from those in settings with complete randomization and treated fraction one half, where adjustments based on the uninteracted and interacted linear adjustments both guarantee gains in precision. Last, we show that estimators with both uninteracted and interacted linear adjustments with pair fixed effects are guaranteed to be weakly more efficient than the unadjusted difference-in-means estimator.   

We then use our framework to consider high-dimensional adjustments based on $\ell_1$ penalization.  Specifically, we first obtain an intermediate estimator by using the LASSO to estimate the ``working model'' for the relevant conditional expectations. When the treatment is determined according to ``matched pairs,'' however, this estimator need not be more precise than the unadjusted difference-in-means estimator.  Therefore, following \cite{cohen2020no-harm}, we consider, in an additional step, an estimator based on the finite-dimensional, linear adjustment described above that uses the predicted values for the ``working model'' as the covariates and includes fixed effects for each of the pairs.  We show that the resulting estimator improves upon both the intermediate estimator and the unadjusted difference-in-means estimator in terms of precision. Moreover, we provide conditions under which the refitted adjustments attain the relevant efficiency bound derived by \cite{armstrong2022asymptotic}.


Concurrent with our paper,
\cite{C23} considers covariate adjustment in experiments in which units are grouped into tuples with possibly more than two units, rather than pairs.  Both our paper and \cite{C23} find that finite-dimensional, linear regression adjustments with pair fixed effects are guaranteed to improve precision relative to the unadjusted difference-in-means estimator, and show that such adjustments are indeed optimal among all linear adjustments. However, \cite{C23} does not pursue more general forms of covariate adjustments, including the regularized adjustments described above.  Such results permit us to study nonparametric adjustments as well as high-dimensional adjustments using covariates whose dimension diverges rapidly with the sample size.

The remainder of our paper is organized as follows.  In Section \ref{sec:setup}, we describe our setup and notation.  In particular, there we describe the precise sense in which we require that units in each pair are ``close'' in terms of their baseline covariates.  In Section \ref{sec:main}, we introduce our general class of estimators based on a ``doubly robust'' moment condition.  Under certain high-level conditions on the ``working models'' and their corresponding estimators, we derive the limiting behavior of the covariate-adjusted estimator.  In Section \ref{sec:linear}, we use our general framework to study a variety of estimators with finite-dimensional, linear covariate adjustment.  In Section \ref{sec:highdim}, we use our general framework to study  covariate adjustment based on the regularized regression.  In Section \ref{sec:simulations}, we examine the finite-sample behavior of tests based on these different estimators via a small simulation study.  We find that covariate adjustment can lead to considerable gains in precision. Finally, in Section \ref{sec:empirical}, we apply our methods to reanalyze data from an experiment using a ``matched pairs'' design to study the effect of macroinsurance on microenterprise. Proofs of all results and some details for simulations are given in the Online Supplement.

\section{Setup and Notation} \label{sec:setup}

Let $Y_i \in \mathbf R$ denote the (observed) outcome of interest for the $i$th unit, $D_i \in \{0,1\}$ be an indicator for whether the $i$th unit is treated, and $X_i \in \mathbf R^{k_x}$ and $W_i \in \mathbf R^{k_w}$ denote observed, baseline covariates for the $i$th unit; $X_i$ and $W_i$ will be distinguished below through the feature that only the former will be used in determining treatment assignment.  Further denote by $Y_i(1)$ the potential outcome of the $i$th unit if treated and by $Y_i(0)$ the potential outcome of the $i$th unit if not treated. The (observed) outcome and potential outcomes are related to treatment status by the relationship
\begin{equation} \label{eq:obsy}
Y_i = Y_i(1)D_i + Y_i(0)(1 - D_i)~.
\end{equation}
For a random variable indexed by $i$, $A_i$, it will be useful to denote by $A^{(n)}$ the random vector $(A_1, \ldots, A_{2n})$.  Denote by $P_n$ the distribution of the observed data $Z^{(n)}$, where $Z_i = (Y_i, D_i,X_i,W_i)$, and by $Q_n$ the distribution of $U^{(n)}$, where $U_i = (Y_i(1),Y_i(0),X_i,W_i)$.  Note that $P_n$ is determined by \eqref{eq:obsy}, $Q_n$, and the mechanism for determining treatment assignment.  We assume throughout that $U^{(n)}$ consists of $2n$ i.i.d.\ observations, i.e., $Q_n = Q^{2n}$, where $Q$ is the marginal distribution of $U_i$.  We therefore state our assumptions below in terms of assumptions on $Q$ and the mechanism for determining treatment assignment. Indeed, we will not make reference to $P_n$ in the sequel, and all operations are understood to be under $Q$ and the mechanism for determining the treatment assignment.  Our object of interest is the average effect of the treatment on the outcome of interest, which may be expressed in terms of this notation as 
\begin{equation} \label{eq:ate}
\Delta(Q) = E[Y_i(1) - Y_i(0)]~.
\end{equation}

We now describe our assumptions on $Q$.  We restrict $Q$ to satisfy the following mild requirement:
\begin{assumption} \label{ass:Q}
The distribution $Q$ is such that
\vspace{-.3cm}
\begin{enumerate}[(a)]
\item $0 < E[\mathrm{Var}[Y_i(d) | X_i]]$ for $d \in \{0, 1\}$.
\item $E[Y_i^2(d)] < \infty$ for $d \in \{0, 1\}$.
\item $E[Y_i(d) | X_i = x]$ and $E[Y_i^2(d) | X_i = x]$ are Lipschitz for $d \in \{0, 1\}$.
\end{enumerate}
\end{assumption}

Next, we describe our assumptions on the mechanism determining treatment assignment. In order to describe these assumptions more formally, we require some further notation to define the relevant pairs of units. The $n$ pairs may be represented by the sets $$\{\pi(2j-1), \pi(2j)\} \text{ for } j = 1, \ldots, n~,$$ where $\pi = \pi_n(X^{(n)})$ is a permutation of $2n$ elements.  Because of its possible dependence on $X^{(n)}$, $\pi$ encompasses a broad variety of different ways of pairing the $2n$ units according to the observed, baseline covariates $X^{(n)}$. Given such a $\pi$, we assume that treatment status is assigned as described in the following assumption:
\begin{assumption} \label{ass:treatment}
Treatment status is assigned so that $(Y^{(n)}(1),Y^{(n)}(0),W^{(n)}) \independent D^{(n)} | X^{(n)}$ and, conditional on $X^{(n)}$, $(D_{\pi(2j-1)},D_{\pi(2j)}), j = 1, \ldots, n$ are i.i.d. and each uniformly distributed over the values in $\{(0,1),(1,0)\}$.
\end{assumption}

Following \cite{bai2021inference}, our analysis will additionally require some discipline on the way in which pairs are formed. Let $\|\cdot\|_2$ denote the Euclidean norm. We will require that units in each pair are ``close'' in the sense described by the following assumption:
\begin{assumption} \label{ass:close}
The pairs used in determining treatment status satisfy $$\frac{1}{n} \sum_{1 \leq j \leq n} \|X_{\pi(2j)} - X_{\pi(2j-1)}\|_2^r \stackrel{P}{\rightarrow} 0$$
for $r \in \{1, 2\}$.
\end{assumption}
\noindent It will at times be convenient to require further that units in consecutive pairs are also ``close'' in terms of their baseline covariates. One may view this requirement, which is formalized in the following assumption, as ``pairing the pairs`` so that they are ``close'' in terms of their baseline covariates.
\begin{assumption} \label{ass:pairsclose}
The pairs used in determining treatment status satisfy $$\frac{1}{n} \sum_{1 \leq j \leq \lfloor \frac{n}{2} \rfloor} \|X_{\pi(4j -k)} - X_{\pi(4j-\ell)}\|_2^2 \stackrel{P}{\rightarrow} 0$$
for any $k \in \{2,3\}$ and $\ell \in \{0,1\}$.
\end{assumption}
\noindent \cite{bai2021inference} provide results to facilitate constructing pairs satisfying Assumptions \ref{ass:close}--\ref{ass:pairsclose} under weak assumptions on $Q$.  In particular, given pairs satisfying Assumption \ref{ass:close}, it is frequently possible to ``re-order'' them so that Assumption \ref{ass:pairsclose} is satisfied.  See Theorem 4.3 in \cite{bai2021inference} for further details.  As in \cite{bai2021inference}, we highlight the fact that Assumption \ref{ass:pairsclose} will only be used to enable consistent estimation of relevant variances.

\begin{remark}
Under this setup, \cite{bai2021inference} consider the unadjusted difference-in-means estimator 
\begin{align}\label{eq:na}
 \hat \Delta_n^{\rm unadj} = \frac{1}{n}\sum_{1\leq i\leq 2n}D_i Y_i - \frac{1}{n}\sum_{1\leq i\leq 2n}(1-D_i) Y_i
\end{align}
 and show that it is consistent and asymptotically normal with limiting variance 
 \begin{align*}
 \sigma_{\mathrm{unadj}}^2(Q) = \frac{1}{2}\var[E[Y_i(1) - Y_i(0) | X_i]] + E[ \var[Y_i(1)|X_i]] + E[\var[Y_i(0)|X_i]]~.
 \end{align*}
We note that $ \hat \Delta_n^{\rm unadj}$ is the unadjusted estimator because it does not use information in $W_i$ in either the design or analysis stage. If both $X_i$ and $W_i$ are used to form pairs in the ``matched pairs'' design, then the difference-in-means estimator, which we refer to as $\hat \Delta_n^{\rm ideal}$, has limiting variance  
\begin{align*}
 \sigma^2_{\mathrm{ideal}}(Q) = \frac{1}{2}\var[E[Y_i(1) - Y_i(0)  | X_i,W_i]] + E[\var[Y_i(1)|X_i,W_i]] + E[\var[Y_i(0)|X_i,W_i]]~.   
\end{align*}
In this case, $\hat \Delta_n^{\rm ideal}$ achieves the efficiency bound derived by \cite{armstrong2022asymptotic}, and we can see that 
\begin{align*}
\sigma_{\mathrm{unadj}}^2(Q) - \sigma^2_{\mathrm{ideal}}(Q) & = \frac{1}{2} E [\var[ E[Y_i(1) + Y_i(0) | X_i,W_i ] | X_i] ] \geq 0~.
\end{align*}
For related results for parameters other than the average treatment effect, see \cite{bai2023efficiency}. We note, however, that it is not always practical to form pairs using both $X_i$ and $W_i$ for two reasons. First, the covariate $W_i$ may only be collected along with the outcome variable and therefore may not be available at the design stage.  Second, the quality of pairing decreases with the dimension of matching variables. Indeed, it is common in practice to match on some but not all baseline covariates. Such considerations motivate our analysis below.
\end{remark}



\section{Main Results} \label{sec:main}

To accommodate various forms of covariate-adjusted estimators of $\Delta(Q)$ in a single framework, it is useful to note that it follows from Assumption \ref{ass:treatment} that for any $d \in \{0,1\}$ and any function $m_{d, n}: \mathbf R^{k_x} \times \mathbf R^{k_w} \rightarrow \mathbf R$ such that $E[|m_{d, n}(X_i,W_i)|] < \infty$,
\begin{equation} \label{eq:doublemoment}
E \left [ 2 I\{D_i = d\}(Y_i - m_{d, n}(X_i,W_i)) + m_{d, n}(X_i,W_i) \right] = E[Y_i(d)]~.
\end{equation}

We note that \eqref{eq:doublemoment} is just the augmented inverse propensity score weighted moment for $E[Y_i(d)]$ in which the propensity score is $1/2$ and the conditional mean model is $m_{d,n}(X_i,W_i)$. Such a moment is also ``doubly robust.'' As the propensity score for the ``matched pairs'' design is exactly one half, we do not require the conditional mean model to be correctly specified,  i.e., $m_{d, n}(X_i,W_i) = E[Y_i(d)|X_i,W_i]$. See, for instance, \cite{robins1995analysis}. Intuitively, $m_{d, n}$ is the ``working model'' which researchers use to estimate $E[Y_i(d) | X_i, W_i]$, and can be arbitrarily misspecified because of \eqref{eq:doublemoment}. Although $m_{d, n}$ will be identical across $n \geq 1$ for the examples in Section \ref{sec:linear}, the notation permits $m_{d, n}$ to depend on the sample size $n$ in anticipation of the high-dimensional results in Section \ref{sec:highdim}. Based on the moment condition in \eqref{eq:doublemoment}, our proposed estimator of $\Delta(Q)$ is given by
\begin{equation} \label{eq:hatdelta}
\hat \Delta_n = \hat \mu_n(1) - \hat \mu_n(0)~,
\end{equation}
where, for $d \in \{0,1\}$,
\begin{equation} \label{eq:hatmu}
\hat \mu_n(d) = \frac{1}{2n} \sum_{1 \leq i \leq 2n} ( 2I\{D_i = d\}(Y_i - \hat m_{d, n}(X_i,W_i)) + \hat m_{d, n}(X_i,W_i) )
\end{equation}
and $\hat m_{d, n}$ is a suitable estimator of the ``working model'' $m_{d, n}$ in \eqref{eq:doublemoment}.

By some simple algebra, we have\footnote{We thank the referee for this excellent point.} 
\begin{align}\label{eq:Deltahat_alt}
\hat \Delta_n = \frac{1}{n}\sum_{1\leq i \leq 2n} D_i \tilde Y_i - \frac{1}{n}\sum_{1\leq i \leq 2n} (1-D_i) \tilde Y_i~,
\end{align}
where 
\begin{align}\label{eq:Ytilde}
\tilde Y_i & = Y_i - \frac{1}{2} (\hat m_{1, n}(X_i, W_i) + \hat m_{0, n}(X_i, W_i))~. 
\end{align}
It means our regression adjusted estimator can be viewed as a difference-in-means estimator, but with the ``adjusted'' outcome $\tilde Y_i$.

We require some new discipline on the behavior of $m_{d, n}$ for $d \in \{0, 1\}$ and $n \geq 1$:
\begin{assumption} \label{ass:md}
The functions $m_{d, n}$ for $d \in \{0, 1\}$ and $n \geq 1$ satisfy
\vspace{-.3cm}
\begin{enumerate}[(a)]
\item For $d \in \{0, 1\}$,
\[\liminf_{n \to \infty} E \left [ \var \left [ Y_i(d) - \frac{1}{2}(m_{1, n}(X_i, W_i) + m_{0, n}(X_i, W_i)) \Bigg | X_i \right ] \right ] > 0~. \]
\item For $d \in \{0, 1\}$,
\[\lim_{\lambda \to \infty} \limsup_{n \to \infty} E[m_{d, n}^2(X_i, W_i) I \{|m_{d, n}(X_i, W_i)| > \lambda\}] = 0~. \]
\item $E[m_{d, n}(X_i, W_i) | X_i = x]$, $E[m_{d, n}^2(X_i, W_i) | X_i = x]$, $E[m_{d, n}(X_i, W_i) Y_i(d) | X_i = x]$ for $d \in \{0, 1\}$, and $E[m_{1, n}(X_i, W_i) m_{0, n}(X_i, W_i) | X_i = x]$ are Lipschitz uniformly over $n \geq 1$.
\end{enumerate}
\end{assumption}

Assumption \ref{ass:md}(a) is an assumption to rule out degenerate situations. Assumption \ref{ass:md}(b) is a mild uniform integrability assumption on the ``working models.'' If $m_{d, n}(\cdot) \equiv m_d(\cdot)$ for $d \in \{0, 1\}$, then it is satisfied as long as $E[m_d^2(X_i, W_i)] < \infty$. Assumption \ref{ass:md}(c) ensures that units that are ``close'' in terms of the observed covariates are also ``close'' in terms of potential outcomes, uniformly across $n \geq 1$.

Theorem \ref{thm:main} below establishes the limit in distribution of $\hat \Delta_n$.  We note that the theorem depends on high-level conditions on $m_{d, n}(\cdot)$ and $\hat m_{d, n}(\cdot)$.  In the sequel, these conditions will be verified in several examples.

\begin{theorem} \label{thm:main}
Suppose $Q$ satisfies Assumption \ref{ass:Q}, the treatment assignment mechanism satisfies Assumptions \ref{ass:treatment}--\ref{ass:close}, and $m_{d, n}(\cdot)$ for $d \in \{0, 1\}$ and $n \geq 1$ satisfy Assumption \ref{ass:md}. Further suppose $\hat m_{d, n}(\cdot)$ satisfies
\begin{equation} \label{eq:rate}
\frac{1}{\sqrt {2n}} \sum_{1 \leq i \leq 2n} (2D_i - 1)(\hat m_{d, n}(X_i,W_i) - m_{d, n}(X_i,W_i)) \stackrel{P}{\rightarrow} 0~.
\end{equation}
Then, $\hat \Delta_n$ defined in \eqref{eq:hatdelta} satisfies 
\begin{equation} \label{eq:normal}
\frac{\sqrt n (\hat \Delta_n - \Delta(Q))}{\sigma_n(Q)} \stackrel{d}{\rightarrow} N(0,1)~,
\end{equation}
where $\sigma_n^2(Q) = \sigma_{1, n}^2(Q) + \sigma_{2,n}^2(Q) + \sigma_{3,n}^2(Q)$ with 
\begin{align*}
\sigma_{1, n}^2(Q) & = \frac{1}{2}E[ \var[E[Y_i(1) + Y_i(0)|X_i,W_i] - (m_{1, n}(X_i,W_i) + m_{0, n}(X_i,W_i))| X_i  ]]\\
\sigma_{2,n}^2(Q) & = \frac{1}{2}\var[E[Y_i(1) - Y_i(0) | X_i,W_i ]] \\
\sigma_{3,n}^2(Q) & = E[\var[Y_i(1)|X_i,W_i]] + E[\var[Y_i(0)|X_i,W_i]]~.
\end{align*}
\end{theorem}

In order to facilitate the use of Theorem \ref{thm:main} for inference about $\Delta(Q)$, we next provide a consistent estimator of $\sigma_n(Q)$. Define
\begin{align*}
\hat \tau_n^2 & = \frac{1}{n} \sum_{1 \leq j \leq n} (\tilde Y_{\pi(2j - 1)} - \tilde Y_{\pi(2j)})^2 \\
\hat \lambda_n & = \frac{2}{n} \sum_{1 \leq j \leq \lfloor \frac{n}{2} \rfloor} (\tilde Y_{\pi(4j - 3)} - \tilde Y_{\pi(4j - 2)}) (\tilde Y_{\pi(4j - 1)} - \tilde Y_{\pi(4j)}) (D_{\pi(4j - 3)} - D_{\pi(4j - 2)}) (D_{\pi(4j - 1)} - D_{\pi(4j)})~,
\end{align*}
where $\tilde Y_i$ is defined in \eqref{eq:Ytilde}.
The variance estimator is given by
\begin{equation} \label{eq:var}
\hat \sigma_n^2 = \hat \tau_n^2 - \frac{1}{2}(\hat \lambda_n + \hat \Delta_n^2)~.
\end{equation}
The variance estimator in \eqref{eq:var}, in particular its component $\hat \lambda_n$, is analogous to the ``pairs of pairs'' variance estimator in \cite{bai2021inference}. Such a variance estimator has also been used in \cite{abadie2008estimation} in a related setting. Note that it can be shown similarly as in Remark 3.9 of \cite{bai2021inference} that $\hat \sigma_n^2$ in \eqref{eq:var} is nonnegative.

Theorem \ref{thm:var} below establishes the consistency of this estimator and its implications for inference about $\Delta(Q)$.  In the statement of the theorem, we make use of the following notation: for any scalars $a$ and $b$, $[a \pm b]$ is understood to be $[a - b, a + b]$.
\begin{theorem} \label{thm:var}
Suppose $Q$ satisfies Assumption \ref{ass:Q}, the treatment assignment mechanism satisfies Assumptions \ref{ass:treatment}--\ref{ass:pairsclose}, and $m_{d, n}(\cdot)$ for $d \in \{0, 1\}$ and $n \geq 1$ satisfy Assumption \ref{ass:md}. Further suppose $\hat m_{d, n}(\cdot)$ satisfies \eqref{eq:rate} and
\begin{equation} \label{eq:L2}
\frac{1}{2n} \sum_{1 \leq i \leq 2n} (\hat m_{d, n}(X_i,W_i) - m_{d, n}(X_i,W_i))^2 \stackrel{P}{\rightarrow} 0~.
\end{equation}
Then,
\[ \frac{\hat \sigma_n}{\sigma_n(Q)} \stackrel{P}{\rightarrow} 1~. \] 
Hence, \eqref{eq:normal} holds with $\hat \sigma_n$ in place of $\sigma_n(Q)$.  In particular, for any $\alpha \in (0,1)$,
\[ P \left \{ \Delta(Q) \in \left [\hat \Delta_n \pm \hat \sigma_n \Phi^{-1} \left ( 1 - \frac{\alpha}{2} \right )\right ] \right\} \rightarrow 1-\alpha~, \]
where $\Phi$ is the standard normal c.d.f.\
\end{theorem}

\begin{remark}\label{rem:conservative_1}
    Based on \eqref{eq:Deltahat_alt}, it is natural to estimate $\sigma_n^2(Q)$ using the usual estimator of the limiting variance of the difference-in-means estimator, i.e.,
   \begin{align*}
\hat \sigma_{\mathrm{diff},n}^{2} & = \frac{1}{n}\sum_{1\leq i\leq 2n} D_i \left(\tilde Y_i- \left(\frac{1}{n}\sum_{1 \leq i \leq 2n}D_i \tilde Y_i\right) \right)^2  + \frac{1}{n}\sum_{1\leq i\leq 2n} (1-D_i) \left(\tilde Y_i- \left(\frac{1}{n}\sum_{1 \leq i \leq 2n}(1-D_i) \tilde Y_i\right) \right)^2.
   \end{align*} 
   However, it can be shown that $\hat \sigma_{\mathrm{diff},n}^{2} = \sigma_{\mathrm{diff},n}^{2}(Q) + o_P(1)$, where
\begin{equation*}
\sigma_{\mathrm{diff},n}^{2}(Q) =    \var \left[Y_i(1) - \frac{1}{2}(m_{1,n}(X_i,W_i) + m_{0,n}(X_i,W_i))\right] + \var \left[Y_i(0) - \frac{1}{2}(m_{1,n}(X_i,W_i) + m_{0,n}(X_i,W_i) )\right]~.
\end{equation*}
Furthermore, 
\begin{align*}
\sigma_{\mathrm{diff},n}^{2}(Q) - \sigma_{n}^{2}(Q) & = \frac{1}{2} \var \left[ E [Y_i(1)+Y_i(0)- (m_{1,n}(X_i,W_i) + m_{0,n}(X_i,W_i)) |X_i]\right] \geq 0~,
\end{align*}
where the inequality is strict unless
\begin{align*}
 E [Y_i(1)+Y_i(0)- (m_{1,n}(X_i,W_i) + m_{0,n}(X_i,W_i)) |X_i] = E [ Y_i(1)+Y_i(0)- (m_{1,n}(X_i,W_i) + m_{0,n}(X_i,W_i))]
\end{align*}
with probability one.  In this sense, the usual estimator of the limiting variance of the difference-in-means estimator is conservative.
\end{remark}

\begin{remark} \label{remark:equivalent}
An important and immediate implication of Theorem \ref{thm:main} is that $\sigma_n^2(Q)$ is minimized when 
\begin{eqnarray*}
&& E[Y_i(0) + Y_i(1)|X_i,W_i] - E[Y_i(0) + Y_i(1)|X_i] = \\
&& \hspace{2cm} m_{0, n}(X_i,W_i) + m_{1, n}(X_i,W_i) - E[m_{0, n}(X_i,W_i) + m_{1, n}(X_i,W_i)| X_i] 
\end{eqnarray*}
with probability one.  In other words, the ``working model'' for $E[Y_i(0) + Y_i(1)|X_i,W_i]$ given by $m_{0, n}(X_i,W_i) + m_{1, n}(X_i,W_i)$, need only be correct ``on average'' over the variables that are not used in determining the pairs.  For such a choice of $m_{0, n}(X_i,W_i)$ and $m_{1, n}(X_i,W_i)$, $\sigma_n^2(Q)$ in Theorem \ref{thm:main} becomes simply
\[ \frac{1}{2}\var[E[Y_i(1) - Y_i(0) \Big | X_i,W_i]] + E[\var[Y_i(1)|X_i,W_i]] + E[\var[Y_i(0)|X_i,W_i]]~, \]
which agrees with the variance obtained in \cite{bai2021inference} when both $X_i$ and $W_i$ are used in determining the pairs. Such a variance also achieves the efficiency bound derived by \cite{armstrong2022asymptotic}.
\end{remark}

\begin{remark}
    Following \cite{bai2023inference}, it is straightforward to extend the analysis in this paper to the case with multiple treatment arms and where treatment status is determined using a ``matched tuples'' design, but we do not pursue this further in this paper.
\end{remark}

\begin{remark}
    Following \cite{bai2021inference}, we conjecture it it possible to establish the validity of a randomization test based on the test statistic studentized by a randomized version of \eqref{eq:var}. We emphasize that the validity of the randomization test depends crucially on the choice of studentization in the test statistic. See, for instance, Remark 3.16 in \cite{bai2021inference}. Such tests have been studied in finite-population settings with covariate adjustments by \cite{zhao2021covariate-adjusted}. We leave a detailed analysis of randomization tests for future work.
\end{remark}

\section{Linear Adjustments} \label{sec:linear}
In this section, we consider linearly covariate-adjusted estimators of $\Delta(Q)$ based on a set of regressors generated by $X_i \in  \mathbf R^{k_x}$ and $W_i \in \mathbf R^{k_w}$.  To this end, define $\psi_i = \psi(X_i, W_i)$, where $\psi: \mathbf R^{k_x} \times \mathbf R^{k_w} \to \mathbf R^p$. We impose the following assumptions on the function $\psi$:

\begin{assumption} \label{ass:psi}
The function $\psi$ is such that
\vspace{-.3cm}
\begin{enumerate}[(a)]
\item no component of $\psi$ is constant and $E[\var[\psi_i | X_i]]$ is non-singular.
\item $\var[\psi_i] < \infty$.
\item $E[\psi_i | X_i = x]$, $E[\psi_i \psi_i' | X_i = x]$, and $E[\psi_i Y_i(d) | X_i = x]$ for $d \in \{0, 1\}$ are Lipschitz.
\end{enumerate}
\end{assumption}
\noindent Assumption \ref{ass:psi} is analogous to Assumption \ref{ass:Q}. Note, in particular, that Assumption \ref{ass:psi}(a) rules out situations where $\psi_i$ is a function of $X_i$ only. See Remark \ref{remark:fogarty} for a discussion of the behavior of the covariate-adjusted estimators in such situations.

\subsection{Linear Adjustments without Pair Fixed Effects}
Consider the following linear regression model:
\begin{equation} \label{eq:naive}
Y_i = \alpha + \Delta D_i + \psi_i' \beta + \epsilon_i~.
\end{equation}
Let $\hat \alpha_n^{\rm naive}$, $\hat \Delta_n^{\rm naive}$, and $\hat \beta_n^{\rm naive}$ denote the OLS estimators of $\alpha$, $\Delta$, and $\beta$ in \eqref{eq:naive}. We call these estimators na\"ive because the corresponding regression adjustment is subject to Freedman's critique and can lead to an adjusted estimator that is less efficient than the simple difference-in-means estimator $\hat \Delta_n^{\rm unadj}$. 

It follows from direct calculation that
\[ \hat \Delta_n^{\rm naive} = \frac{1}{n} \sum_{1 \leq i \leq 2n} (Y_i - \psi_i' \hat \beta_n^{\rm naive}) (2 D_i - 1)~. \]
Therefore, $\hat \Delta_n^{\rm naive}$ satisfies \eqref{eq:hatdelta}--\eqref{eq:hatmu} with
\[ \hat m_{d, n}(X_i, W_i) = \psi_i' \hat \beta_n^{\rm naive}~. \]

Theorem \ref{thm:naive} establishes \eqref{eq:rate} and \eqref{eq:L2} for a suitable choice of $m_{d, n}(X_i, W_i)$ for $d \in \{0, 1\}$ and, as a result, the limiting distribution of $\hat \Delta_n^{\rm naive}$ and the validity of the variance estimator.

\begin{theorem} \label{thm:naive}
Suppose $Q$ satisfies Assumption \ref{ass:Q} and the treatment assignment mechanism satisfies Assumptions \ref{ass:treatment}--\ref{ass:close}. Further suppose $\psi$ satisfies Assumption \ref{ass:psi}. Then, as $n \to \infty$,
\[ \hat \beta_n^{\rm naive} \stackrel{P}{\to} \beta^{\rm naive} = \var[\psi_i]^{-1} \cov[\psi_i, Y_i(1) + Y_i(0)]~. \]
Moreover, \eqref{eq:rate}, \eqref{eq:L2}, and Assumption \ref{ass:md} are satisfied with
\begin{equation} \nonumber
m_{d, n}(X_i, W_i) = \psi_i' \beta^{\rm naive}   
\end{equation}
for $d \in \{0, 1\}$ and $n \geq 1$.
\end{theorem}

\begin{remark} \label{remark:lin}
\cite{freedman2008regression} studies regression adjustment based on \eqref{eq:naive} when treatment is assigned by complete randomization instead of a ``matched pairs'' design. In such settings, \cite{lin2013agnostic} proposes adjustment based on the following linear regression model:
\begin{equation} \label{eq:lin}
Y_i = \alpha + \Delta D_i + (\psi_i - \bar \psi_n)' \gamma + D_i (\psi_i - \bar \psi_n)' \eta + \epsilon_i~,
\end{equation}
where
\[ \bar \psi_n = \frac{1}{2n} \sum_{1 \leq i \leq 2n} \psi_i~. \]
Let $\hat \alpha_n^{\rm int}, \hat \Delta_n^{\rm int}, \hat \gamma_n^{\rm int}, \hat \eta_n^{\rm int}$ denote the OLS estimators for $\alpha, \Delta, \gamma, \eta$ in \eqref{eq:lin}. It is straightforward to show $\hat \Delta_n^{\rm int}$ satisfies \eqref{eq:hatdelta}--\eqref{eq:hatmu} with
\begin{align*}
\hat m_{1, n}(X_i, W_i) & = (\psi_i - \hat \mu_{\psi, n}(1))' (\hat \gamma_n^{\rm int} + \hat \eta_n^{\rm int}) \\
\hat m_{0, n}(X_i, W_i) & = (\psi_i - \hat \mu_{\psi, n}(0))' \hat \gamma_n^{\rm int}~,
\end{align*}
where
\[ \hat \mu_{\psi, n}(d) = \frac{1}{n} \sum_{1 \leq i \leq 2n} I \{D_i = d\} \psi_i~. \]
It can be shown using similar arguments to those used to establish Theorem \ref{thm:naive} that \eqref{eq:rate} and Assumption \ref{ass:md} are satisfied with
\[ m_{d, n}(X_i, W_i) = (\psi_i - E[\psi_i])' \var[\psi_i]^{-1} \cov[\psi_i, Y_i(d)] \]
for $d \in \{0, 1\}$ and $n \geq 1$. It thus follows by inspecting the expression for $\sigma^2_n(Q)$ in Theorem \ref{thm:main} that the limiting variance of $\hat \Delta_n^{\rm int}$ is the same as that of $\hat \Delta_n^{\rm naive}$ based on \eqref{eq:naive}.
\end{remark}

\begin{remark}\label{rem:sigma_naive1}
Note that $\hat \Delta_n^{\rm naive}$ is the ordinary least squares estimator for $\Delta$ in the linear regression 
\begin{align*}
    Y_i - \psi_i' \hat \beta_n^{\rm naive} =  \alpha + \Delta D_i + \epsilon_{i}~.
\end{align*}
Furthermore, Theorem \ref{thm:naive} implies that its limiting variance is $\sigma_{\mathrm{naive}}^2(Q)$, given by $\sigma_n^2(Q)$ in Theorem \ref{thm:main} with $m_d(X_i,W_i) = \psi_i'\beta^{\rm naive}$.  The usual heteroskedasticity-robust estimator of the limiting variance of $\hat \Delta_n^{\rm naive}$ is, however, simply $\hat \sigma_{\mathrm{diff}, n}^2$ defined in Remark \ref{rem:conservative_1} with $\hat m_{d,n}(X_i,W_i) = \psi_i' \hat \beta_n^{\rm naive}$.  It thus follows that $\hat \sigma_{\mathrm{diff}, n}^2$ is conservative for $\sigma_{\mathrm{naive}}^2(Q)$ in the sense described therein. 
It is, of course, possible to estimate $\sigma_{\mathrm{naive}}^2(Q)$ consistently using $\hat \sigma_n^2$ proposed in Theorem \ref{thm:var} with $\hat m_{d,n}(W_i,X_i) = \psi_i'\hat \beta_n^{\rm naive}$, but $\sigma_{\mathrm{naive}}^2(Q)$ is not guaranteed to be smaller than the limiting variance of the unadjusted estimator, i.e., $\sigma_{\mathrm{unadj}}^2(Q)$, so the linear adjustment without pair fixed effects can harm the precision of the estimator.  Evidence of this phenomenon is provided in our simulations in Section \ref{sec:simulations}.
\end{remark}


\subsection{Linear Adjustments with Pair Fixed Effects}
\label{sec:pfe}
Remark \ref{remark:lin} implies that in ``matched pairs'' designs, including interaction terms in the linear regression does not lead to an estimator with lower limiting variance than the one based on the linear regression without interaction terms. It is therefore interesting to study whether there exists a linearly covariate-adjusted estimator with lower limiting variance than the ones based on \eqref{eq:naive} and \eqref{eq:lin} as well as the difference-in-means estimator. To that end, consider instead the following linear regression model:
\begin{equation} \label{eq:pfe}
Y_i = \Delta D_i + \psi_i' \beta + \sum_{1 \leq j \leq n} \theta_j I \{i \in \{\pi(2j - 1), \pi(2j)\}\} + \epsilon_i~.
\end{equation}
Let $\hat \Delta_n^{\rm pfe}$, $\hat \beta_n^{\rm pfe}$, and $\hat \gamma_{j, n}$, $1 \leq j \leq n$ denote the OLS estimators of $\Delta$, $\beta$, $\theta_j$, $1 \leq j \leq n$ in \eqref{eq:pfe}, where ``pfe'' stands for pair fixed effect. It follows from the Frisch-Waugh-Lovell theorem that
\[ \hat \Delta_n^{\rm pfe} = \frac{1}{n} \sum_{1 \leq i \leq 2n} (Y_i - \psi_i' \hat \beta_n^{\rm pfe}) (2 D_i - 1)~. \]
Therefore, $\hat \Delta_n^{\rm pfe}$ satisfies \eqref{eq:hatdelta}--\eqref{eq:hatmu} with
\[ \hat m_{d, n}(X_i, W_i) = \psi_i' \hat \beta_n^{\rm pfe}~. \]

Theorem \ref{thm:pfe} establishes \eqref{eq:rate} and \eqref{eq:L2} for a suitable choice of $m_{d, n}(X_i, W_i), d \in \{0, 1\}$ and, as a result, the limiting distribution of $\hat \Delta_n^{\rm pfe}$ and the validity of the variance estimator.

\begin{theorem} \label{thm:pfe}
Suppose $Q$ satisfies Assumption \ref{ass:Q} and the treatment assignment mechanism satisfies Assumptions \ref{ass:treatment}--\ref{ass:close}. Then, as $n \to \infty$,
\[ \hat \beta_n^{\rm pfe} \stackrel{P}{\to} \beta^{\rm pfe} = (2 E[\var[\psi_i | X_i]])^{-1} E[\cov[\psi_i, Y_i(1) + Y_i(0) | X_i]]~. \]
Moreover, \eqref{eq:rate}, \eqref{eq:L2}, and Assumption \ref{ass:md} are satisfied with
\[ m_{d, n}(X_i, W_i) = \psi_i' \beta^{\rm pfe} \]
for $d \in \{0, 1\}$ and $n \geq 1$.
\end{theorem}

\begin{remark} \label{remark:fogarty}
When $\psi$ is restricted to be a function of $X_i$ only, $\hat \Delta_n^{\rm pfe}$ coincides to first order with the unadjusted difference-in-means estimator $\hat \Delta_n^{\rm unadj}$ defined in \eqref{eq:na}. To see this, suppose further that $\psi$ is Lipschitz and that $\text{Var}[Y_i(d)|X_i = x], d \in \{0,1\}$ are bounded.  The proof of Theorem \ref{thm:pfe} reveals that $\hat \Delta_n^{\rm pfe}$ and $\hat \beta_n^{\rm pfe}$ coincide with the OLS estimators of the intercept and slope parameters in a linear regression of $(Y_{\pi(2j)} - Y_{\pi(2j-1)})(D_{\pi(2j)} - D_{\pi(2j-1)})$ on a constant and $(\psi_{\pi(2j)} - \psi_{\pi(2j-1)})(D_{\pi(2j)} - D_{\pi(2j-1)})$.  Using this observation, it follows by arguing as in Section S.1.1 of \cite{bai2021inference} that \[ \sqrt n (\hat \Delta_n^{\rm pfe} - \Delta(Q)) = \sqrt n (\hat \Delta_n^{\rm unadj} - \Delta(Q)) + o_P(1) ~.\] See also Remark 3.8 of \cite{bai2021inference}.
\end{remark}

\begin{remark} \label{remark:optimal}
Note in the expression of $\sigma_n^2(Q)$ in Theorem \ref{thm:main} only depends on $m_{d, n}(X_i, W_i), d \in \{0,1\}$ through $\sigma_{1, n}^2(Q)$. With this in mind, consider the class of all linearly covariate-adjusted estimators based on $\psi_i$, i.e., $m_{d, n}(X_i, W_i) = \psi_i' \beta(d)$. For this specification of $m_{d, n}(X_i, W_i), d \in \{0,1\}$,
\[ \sigma_{1, n}^2(Q) = E[(E[Y_i(1) + Y_i(0) | X_i, W_i] - E[Y_i(1) + Y_i(0) | X_i] - (\psi_i - E[\psi_i | X_i])' (\beta(1) + \beta(0)))^2]~. \]
It follows that among all such linear adjustments, $\sigma_n^2(Q)$ in \eqref{eq:normal} is minimized when
\[ \beta(1) + \beta(0) = 2 \beta^{\rm pfe}~. \]
This observation implies that the linear adjustment with  pair fixed effects, i.e., $\hat \Delta_n^{\rm pfe}$, yields the optimal linear adjustment in the sense of minimizing $\sigma_n^2(Q)$.  Its limiting variance is, in particular, weakly smaller than the limiting variance of the unadjusted difference-in-means estimator defined in \eqref{eq:na}. On the other hand, the covariate-adjusted estimators based on \eqref{eq:naive} or \eqref{eq:lin}, i.e., $\hat \Delta_n^{\rm naive}$ and $\hat \Delta_n^{\rm int}$, are in general not optimal among all linearly covariate-adjusted estimators based on $\psi_i$. In fact, the limiting variances of these two estimators may even be larger than that of the unadjusted difference-in-means estimator. 
\end{remark}

\begin{remark}
 ``Matched pairs'' design is essentially a non-parametric way to  adjust for $X_i$. Projecting $\psi_i$ on the pair dummies in \eqref{eq:pfe} is equivalent to pair-wise demeaning, which effectively removes $E (\psi_i|X_i)$ from $\psi_i$. This is key to the optimality of $\hat \Delta_n^{\rm pfe}$ over all linearly adjusted estimators. Following the same logic, we expect that by replacing the pair dummies with sieve bases of $X_i$ in \eqref{eq:pfe}, the linear regression can still effectively remove $E (\psi_i|X_i)$ from $\psi_i$ so that the new adjusted estimator is asymptotically equivalent to $\hat \Delta_n^{\rm pfe}$, and thus, linearly optimal.      
\end{remark}

\begin{remark}\label{rem:sigma_naive3}
  Remark \ref{rem:sigma_naive1} also applies here with $\beta^{\rm naive}$ replaced by $\beta^{\rm pfe}$. Even though $\hat \Delta_n^{\rm pfe}$ can be computed via OLS estimation of \eqref{eq:pfe}, we emphasize that the usual heteroskedascity-robust standard errors that na\"ively treats the data (including treatment status) as if it were i.i.d.\ need not be consistent for the limiting variance derived in our analysis. 
\end{remark}

\begin{remark} \label{remark:int-pfe}
One can also consider the estimator based on the following linear regression model:
\begin{equation} \label{eq:int-pfe}
Y_i = \Delta D_i + (\psi_i - \bar \psi_n)' \gamma + D_i (\psi_i - \hat \mu_{\psi, n}(1))' \eta + \sum_{1 \leq j \leq n} \theta_j I \{i \in \{\pi(2j - 1), \pi(2j)\}\} + \epsilon_i~.
\end{equation}
Let $\hat \Delta_n^{\rm int-pfe}, \hat \gamma_n^{\rm int-pfe}, \hat \eta_n^{\rm int-pfe}$ denote the OLS estimators for $\Delta, \gamma, \eta$ in \eqref{eq:int-pfe}. It is straightforward to show $\hat \Delta_n^{\rm int - pfe}$ satisfies \eqref{eq:hatdelta}--\eqref{eq:hatmu} with
\begin{align*}
\hat m_{1, n}(X_i, W_i) & = (\psi_i - \hat \mu_{\psi, n}(1))' \hat \eta_n^{\rm int-pfe} \\
\hat m_{0, n}(X_i, W_i) & = (\psi_i - \hat \mu_{\psi, n}(0))' (\hat \eta_n^{\rm int-pfe} - \hat \gamma_n^{\rm int-pfe})~.
\end{align*}
Following similar arguments to those used in the proof of Theorem \ref{thm:naive}, we can establish that \eqref{eq:rate} and Assumption \ref{ass:md} are satisfied with
\begin{align*}
m_{1, n}(X_i, W_i) & = (\psi_i - E[\psi_i])' \eta^{\rm int-pfe} \\
m_{0, n}(X_i, W_i) & = (\psi_i - E[\psi_i])' (\eta^{\rm int-pfe} - \gamma^{\rm int-pfe})~,
\end{align*}
where
\begin{align*}
\gamma^{\rm int-pfe} & = (E[\var[\psi_i | X_i]])^{-1} E[\mathrm{Cov}[\psi_i, Y_i(1) - Y_i(0) | X_i]]~, \\
\eta^{\rm int-pfe} & = (E[\var[\psi_i | X_i]])^{-1} E[\mathrm{Cov}[\psi_i, Y_i(1) | X_i]]~.
\end{align*}
Because $2\eta^{\rm int-pfe} - \gamma^{\rm int-pfe} = 2 \beta^{\rm pfe}$, it follows from Remark \ref{remark:optimal} that the limiting variance of $\hat \Delta_n^{\rm int-pfe}$ is identical to the limiting variance of $\hat \Delta_n^{\rm pfe}$.
\end{remark}

\begin{remark}
\cite{WGB21} consider the covariate adjustment for paired experiments under the design-based framework, where the covariates are treated as deterministic, and thus, the cross-sectional dependence between units in the same pair due to the closeness of their covariates is not counted in their analysis. We differ from them by considering the sampling-based framework in which the covariates are treated as random and the pairs are formed by matching, and thus, have an impact on statistical inference. Under their framework, \cite{WGB21} point out that covariate adjustments may have a positive or negative effect on the estimation accuracy depending on how they are estimated. This is consistent with our findings in this section. Specifically, we show that when the regression adjustments are estimated by a linear regression with pair fixed effects, the resulting ATE estimator is guaranteed to weakly improve upon the difference-in-means estimator in terms of efficiency. However, this improvement is not guaranteed if the adjustments are estimated without pair fixed effects.
\end{remark}

\begin{remark}
   If we choose $\psi_i$ as a set of sieve basis functions with increasing dimension, then under suitable regularity conditions, the linear adjustments both with and without pair fixed effects achieve the same limiting variance as $\hat \Delta^{\rm ideal}_n$, and thus, the  efficiency bound. In fact, if $\psi_i$ contains sieve bases, then the linear adjustment without pair fixed effects can approximate the true specification $E[Y_i(1) + Y_i(0)|X_i,W_i]$ in the sense that $E[Y_i(1) + Y_i(0)|X_i,W_i]  = \psi_i' \beta^{\rm naive} + R_i$ and $E[R_i^2] = o(1)$. This property implies $\sigma_{1, n}^2(Q)$ in Theorem \ref{thm:main} equals zero. Similarly, the linear adjustment with pair fixed effects can approximate the true specification $E[Y_i(1) + Y_i(0)|X_i,W_i] - E[Y_i(1) + Y_i(0)|X_i]$ in the sense that $E[Y_i(1) + Y_i(0)|X_i,W_i] - E[Y_i(1) + Y_i(0)|X_i] = \tilde \psi_i' \beta^{\rm naive} + \tilde R_i$ and $E [\tilde R_i^2] = o(1)$. This property again implies $\sigma_{1, n}^2(Q)$ in Theorem \ref{thm:main} equals zero. Therefore, in both cases, the adjusted estimator achieves the minimum variance.  In the next section, we consider $\ell_1$-regularized adjustments which may be viewed as providing a way to choose the relevant sieve bases in a data-driven manner.    
\end{remark}

\section{Regularized Adjustments} \label{sec:highdim}
In this section, we study covariate adjustments based on the $\ell_1$-regularized linear regression. Such settings can arise if the covariates $W_i$ are high-dimensional or if the dimension of $W_i$ is fixed but the regressors include many sieve basis functions of $X_i$ and $W_i$. To accommodate situations where the dimension of $W_i$ increases with $n$, we add a subscript and denote it by $W_{n, i}$ instead. Let $k_{w, n}$ denote the dimension of $W_{n, i}$. For $n \geq 1$, let $\psi_{n,i} = \psi_n(X_i,W_{n,i})$, where $\psi_n: \mathbf R^{k_x} \times \mathbf R^{k_{w, n}} \to \mathbf R^{p_n}$ and $p_n$ will be permitted below to be possibly much larger than $n$.

In what follows, we consider a two-step method in the spirit of \cite{cohen2020no-harm}.  In the first step, an intermediate estimator, $\hat \Delta_n^{\rm r}$, is obtained using \eqref{eq:hatdelta} with a ``working model'' obtained through a $\ell_1$-regularized linear regression adjustments $m_{d,n}(X_i,W_i)$ for $d \in \{0,1\}$.  As explained further below in Theorem \ref{thm:LASSO2}, when $m_{d,n}(X_i,W_i)$ is approximately correctly specified, such an estimator is optimal in the sense that it minimizes the limiting variance in Theorem \ref{thm:main}.  When this is not the case, however, for reasons analogous to those put forward in  Remark \ref{rem:sigma_naive1}, it needs not to have a limiting variance weakly smaller than the unadjusted difference-in-means estimator.  In a second step, we therefore consider an estimator by refitting a version of \eqref{eq:pfe} in which the covariates $\psi_i$ are replaced by the regularized estimates of $m_{d,n}(X_i,W_i)$ for $d \in \{0,1\}$.  The resulting estimator, $\hat \Delta_n^{\rm refit}$, has the limiting variance weakly smaller than that of the intermediate estimator and thus remains optimal under approximately correct specification in the same sense. Moreover, it has limiting variance weakly smaller than the unadjusted difference-in-means estimator. \cite{WDTT16} also consider high-dimensional regression adjustments in
randomized experiments using LASSO. We differ from their work by considering the ``matched pairs'' design, and more importantly, discussing when and how regularized adjustments can improve estimation efficiency upon the difference-in-means estimator.

Before proceeding, we introduce some additional notation that will be required in our formal description of the methods. 
We denote by $\psi_{n,i,l}$ the $l$th components of $\psi_{n,i}$. For a vector $a \in \mathbf R^k$ and $0 \leq p \leq \infty$, recall that $$\|a\|_p = \Big ( \sum_{1 \leq l \leq k} |a_l|^p \Big )^{1/p}~,$$ where it is understood that $\|a\|_0 = \sum_{1 \leq l \leq k} I \{a_k \neq 0\}$ and $\|a\|_\infty = \sup_{1 \leq l \leq k} |a_l|$.
Using this notation, we further define $$\Xi_n = \sup_{(x,w) \times \mathrm{supp}(X_i) \times \mathrm{supp}(W_i) }\|\psi_{n,i}(x,w)\|_\infty ~.$$

For $d \in \{0, 1\}$, define
\begin{equation} \label{eq:LASSO2}
(\hat \alpha_{d, n}^{\rm r}, \hat \beta_{d, n}^{\rm r}) \in \argmin_{a \in \mathbf R, b \in \mathbf R^{p_n}} \frac{1}{n} \sum_{1 \leq i \leq 2n: D_i = d} (Y_i - a - \psi_{n,i}' b)^2 + \lambda_{d, n}^{\rm r} \|\hat \Omega_n(d)  b\|_1~,
\end{equation}
where  $\lambda_{d, n}^{\rm r}$ is a penalty parameter that will be disciplined by the assumptions below, $\hat{\Omega}_n(d) = \diag(\hat{\omega}_1(d),\cdots,\break\hat{\omega}_{p_n}(d))$ is a diagonal matrix, and $\hat{\omega}_{n,l}(d)$ is the penalty loading for the $l$th regressor. Let $\hat \Delta_n^{\rm r}$ denote the estimator in \eqref{eq:hatdelta} with $\hat m_{d, n}(X_i,W_{n,i}) = \psi_{n, i}' \hat \beta_{d, n}^{\rm r}$ for $d \in \{0, 1\}$. 

We now proceed with the statement of our assumptions. The first assumption collects a variety of moment conditions that will be used in our formal analysis:
\begin{assumption} 
\begin{enumerate}[(a)]
\item There exist nonrandom quantities $(\alpha_{d,n}^{\rm r}, \beta_{d,n}^{\rm r})$ such that with $\epsilon_{n,i}^{\rm r}(d)$ defined as 
\begin{align*}
   \epsilon_{n,i}^{\rm r}(d) = Y_i(d) -  \alpha_{d, n}^{\rm r} - \psi_{n, i}' \beta_{d, n}^{\rm r}~,
\end{align*}
we have
\begin{equation} \label{eq:foc2}
	\|\Omega_n(d)^{-1}E[\psi_{n, i} \epsilon_{n, i}^{\rm r}(d)]\|_\infty + |E[\epsilon_{n, i}^{\rm r}(d)]|= o\left( \lambda_{d, n}^{\rm r}\right)~,
\end{equation}
where $\Omega_n(d) = \diag(\omega_{n,1}(d),\cdots,\omega_{n,p_n}(d))$ and $\omega_{n,l}^{2}(d) = \var[\psi_{n, i, l} \epsilon_{n,i}^{\rm r}(d)]$. 

\item For some $q > 2$ and constant $C_1$, 
\begin{align*}
\sup_{n \geq 1} \max_{1 \leq l \leq p_n} E[|\psi_{n, i, l}^q| | X_i] &\leq C_1~, \\
\sup_{n \geq 1} |\psi_{n,i}'\beta_{d,n}^{\rm r-pd}| &\leq C_1~, \\
\sup_{n \geq 1} |E[Y_i(a)|X_i,W_{n,i}]| &\leq C_1~,
\end{align*}
with probability one.
\item For some $\ubar c$ and $\bar c$, we require that
\begin{equation} \label{eq:loadings2}
    0<\underline{c} \leq \liminf_{n \rightarrow \infty} \min_{1 \leq l \leq p_n} \hat{\omega}_{n,l}(d)/\omega_{n,l}(d) \leq  \limsup_{n \rightarrow \infty} \max_{1 \leq l \leq p_n} \hat{\omega}_{n,l}(d)/\omega_{n,l}(d) \leq \bar{c} < \infty.
\end{equation}
\item For some $c_0$, $\ubar \sigma$, $\bar \sigma$, the following statements hold with probability one:
\[ 0<\ubar\sigma^2 \leq \liminf_{n \rightarrow \infty}~ \min_{d \in \{0, 1\}, 1 \leq l \leq p_n} \omega_{n,l}^2(d) \leq \limsup_{n \rightarrow \infty}~ \max_{d \in \{0, 1\}, 1 \leq l \leq p_n} \omega_{n,l}^2(d) \leq \bar{\sigma}^2 < \infty~, \]
\begin{align*}
		\sup_{n \geq 1} \max_{d \in \{0, 1\}}E[(\psi_{n, i}'\beta_{d,n}^{\rm r})^2] \leq c_0 & < \infty~, \\ 
		\max_{d \in \{0, 1\}, 1 \leq l \leq p_n} \frac{1}{2n} \sum_{1 \leq i \leq 2n} E[\epsilon_{n, i}^4(d) | X_i] \leq c_0 & < \infty~, \\
  		\sup_{n \geq 1}~ \max_{d \in \{0, 1\}} E[\epsilon_{n, i}^4(d)] \leq c_0 & < \infty~, \\
 \min_{d \in \{0, 1\}} \var[Y_i(d) - \psi_{n,i}'(\beta_{1,n}^{\rm r} + \beta_{0,n}^{\rm r})/2] \geq \ubar \sigma^2 & > 0~, \\
		\min_{1 \leq l \leq p_n} \frac{1}{n} \sum_{1 \leq i \leq 2n} I \{D_i = d\} \var[\psi_{n, i, l} \epsilon_{n, i}(d) | X_i] \geq \ubar \sigma^2 & > 0~, \\
		\min_{1 \leq l \leq p_n} \var[E[\psi_{n, i, l} \epsilon_{n, i}(d) | X_i]] \geq \ubar \sigma^2 & > 0~. 
\end{align*}
\end{enumerate}\label{ass:LASSO2-moments}
\end{assumption}

\begin{remark}
It is instructive to note that \eqref{eq:foc2} in Assumption \ref{ass:LASSO2-moments}(a) is the subgradient condition for a $\ell_1$-penalized regression of the outcome $Y_i(d)$ on $ \psi_{n,i}$ when the penalty is of order $o(\lambda_n^{\rm r})$. Specifically, if $p_n \ll n$, then this condition holds automatically for the $\beta_{d, n}^{\rm r}$ equal to the coefficients of a linear projection of $Y_i(d)$ onto $(1, \psi_{n,i}')$.  When $p_n \gg n$, but $E[Y_i(d)|X_i,W_i]$ is approximately correctly specified in the sense that the approximation error $ R_{n,i}(d) = E[Y_i(d)|X_i,W_i] - \alpha_{d,n}^{\rm r} - \psi_{n, i}'\beta_{d,n}^{\rm r}$ is sufficiently small, then \eqref{eq:foc2} also holds. However, the approximately correct specification is not necessary for \eqref{eq:foc2}. For example, suppose $W_{n,i} = (W_{n,i,1},\cdots,W_{n,i,p_n})$ is a $p_n$ vector of independent standard normal random variables, $W_{n,i}$ is independent of $X_i$, $\psi_{n,i} = (X_i',W_{n,i}')'$, and 
\begin{align*}
    Y_i(d) = \alpha_{d,n}^{\rm r} + \psi_{n,i}'\beta_{d,n}^{\rm r} + \sum_{l=1}^{p_n}\frac{W_{n,i,l}^2-1}{\sqrt{p_n}} + u_{n,i}(d)~,
\end{align*}
where $E (u_{n,i}(d)|X_i,W_{n,i}) = 0$. Then, Assumption \ref{ass:LASSO2-moments}(a) holds with $\epsilon_{n,i}^{\rm r}(d) = \sum_{l=1}^{p_n}\frac{W_{n,i,l}^2-1}{\sqrt{p_n}} + u_{n,i}(d)$. We can impose a sparse restriction on $\beta_{d,n}^{\rm r}$ so that it further satisfies Assumption \ref{ass:LASSO2-penalty}(b) below. On the other hand, the linear regression adjustment is not approximately correctly specified because $R_{n,i}(d) = E (Y_i(d)|X_i,W_{n,i})  - (\alpha_{d,n}^{\rm r} + \psi_{n,i}'\beta_{d,n}^{\rm r})=  \sum_{l=1}^{p_n}\frac{W_{n,i,l}^2-1}{\sqrt{p_n}}$, and we have $E R_{n,i}^2(d) = 2 \nrightarrow 0$. 
\end{remark}

\begin{remark}    
Assumption \ref{ass:LASSO2-moments}(b) and \ref{ass:LASSO2-moments}(d) are standard in the high-dimensional estimation literature; see, for instance, \cite{belloni2017program}. The last four inequalities in Assumption \ref{ass:LASSO2-moments}(d), in particular, permit us to apply the high-dimensional central limit theorem in \citet[Theorem 2.1]{chernozhukov2017central}.     
\end{remark}

\begin{remark}
The penalty loadings in Assumption \ref{ass:LASSO2-moments}(c) can be computed by an iterative procedure proposed by \cite{belloni2017program}. We provide more detail in Algorithm \ref{algo:2} below. We can then verify \eqref{eq:loadings2} under ``matched pairs'' designs following arguments similar to those in \cite{belloni2017program}. 
\end{remark}

Our analysis will, as before, also require some discipline on the way in which pairs are formed.  For this purpose, Assumption \ref{ass:close} will suffice, but we will need an additional Lipshitz-like condition:
\begin{assumption} \label{ass:LASSO2-match}
For some $L > 0$ and any $x_1$ and $x_2$ in the support of $X_i$, we have
\[ |(\Psi(x_1)-\Psi(x_2))'\beta_{d,n}^{\rm r}| \leq L ||x_1-x_2||_2~. \]
\end{assumption}

We next specify our restrictions on the penalty parameter $\lambda_{d, n}^{\rm r}$.
\begin{assumption} \label{ass:LASSO2-penalty}
\begin{enumerate}[(a)]
  \item For some $\ell \ell_n \rightarrow \infty$, 
		\[ \lambda_{d, n}^{\rm r} = \frac{\ell \ell_n }{\sqrt n} \Phi^{-1} \left ( 1 - \frac{0.1}{2 \log(n) p_n} \right )~. \]
  \item $\Xi_n^2 (\log p_n)^7 / n \to 0$ and $(\ell\ell_n s_n \log p_n) / \sqrt{n} \to 0$, where 
  \begin{equation} \label{eq:sparse2}
s_n = \max_{d \in \{0, 1\}} \|\beta_{d,n}^{\rm r}\|_0.
\end{equation}
  \end{enumerate}
\end{assumption}
We note that Assumption \ref{ass:LASSO2-penalty}(b) permits $p_n$ to be much greater than $n$. It also requires sparsity in the sense that $s_n = o(\sqrt{n})$. 

Finally, as is common in the analysis of $\ell_1$-penalized regression, we require a ``restricted eigenvalue'' condition.  This assumption permits us to apply \citet[Lemma 4.1]{bickel2009simultaneous} and establish the error bounds for $|\hat \alpha_{d,n}^{\rm r} - \alpha_{d,n}^{\rm r}|+||\hat \beta_{d,n}^{\rm r} - \beta_{d,n}^{\rm r}||_1$  and $\frac{1}{n}\sum_{1\leq i\leq 2n}I\{D_i=d\}\left(\hat \alpha_{d,n}^{\rm r} - \alpha_{d,n}^{\rm r}+\psi_{n,i}'(\hat \beta_{d,n}^{\rm r} - \beta_{d,n}^{\rm r})\right)^2$.
\begin{assumption} \label{ass:LASSO2-ev}
   For some $\kappa_1 > 0, \kappa_2$ and $\ell_n \to \infty$, the following statements hold with probability approaching one:
		\begin{align*}
		\inf_{d \in \{0, 1\}, v \in \mathbf R^{p_n+1}: \|v\|_0 \leq (s_n+1) \ell_n} (\|v\|_2^2)^{-1} v' \Bigg ( \frac{1}{n} \sum_{1 \leq i \leq 2n} I \{D_i = d\} \breve \psi_{n, i} \breve \psi_{n, i}' \Bigg ) v & \geq \kappa_1 \\
		\sup_{d \in \{0, 1\}, v \in \mathbf R^{p_n+1}: \|v\|_0 \leq (s_n+1) \ell_n} (\|v\|_2^2)^{-1} v' \Bigg ( \frac{1}{n} \sum_{1 \leq i \leq 2n} I \{D_i = d\} \breve \psi_{n, i} \breve \psi_{n, i}' \Bigg ) v & \leq \kappa_2 \\
		\inf_{d \in \{0, 1\}, v \in \mathbf R^{p_n+1}: \|v\|_0 \leq (s_n+1) \ell_n} (\|v\|_2^2)^{-1} v' \Bigg ( \frac{1}{n} \sum_{1 \leq i \leq 2n} I \{D_i = d\} E[\breve \psi_{n, i}\breve \psi_{n, i}' | X_i] \Bigg ) v & \geq \kappa_1 \\
		\sup_{d \in \{0, 1\}, v \in \mathbf R^{p_n+1}: \|v\|_0 \leq (s_n+1) \ell_n} (\|v\|_2^2)^{-1} v' \Bigg ( \frac{1}{n} \sum_{1 \leq i \leq 2n} I \{D_i = d\} E[\breve \psi_{n, i} \breve \psi_{n, i}' | X_i] \Bigg ) v & \leq \kappa_2~,
		\end{align*}
		where $\breve \psi_{n,i} = (1,\psi_{n,i}')'$.
\end{assumption}

Using these assumptions, the following theorem characterizes the behavior of $\hat \Delta_n^{\rm r}$:
\begin{theorem} \label{thm:LASSO2}
Suppose $Q$ satisfies Assumption \ref{ass:Q} and the treatment assignment mechanism satisfies Assumptions \ref{ass:treatment}--\ref{ass:close}.  Further suppose Assumptions \ref{ass:LASSO2-moments}--\ref{ass:LASSO2-ev} hold. Then, \eqref{eq:rate},  \eqref{eq:L2}, and Assumption \ref{ass:md} are satisfied with $\hat m_{d, n}(X_i, W_{n, i}) = \hat \alpha_{d,n}^{\rm r}+ \psi_{n, i}' \hat \beta_{d, n}^{\rm r}$ and
\[ m_{d, n}(X_i, W_{n, i}) = \alpha_{d,n}^{\rm r} + \psi_{n, i}' \beta_{d, n}^{\rm r} \]
for $d \in \{0, 1\}$ and $n \geq 1$. Denote the variance of $\hat \Delta_n^{\rm r}$ by $\sigma_n^{\rm r,2}$. If the regularized adjustment is approximately correctly specified, i.e., $E [Y_i(d)|X_i,W_{n,i}] = \alpha_{d,n}^{\rm r} + \psi_{n,i}'\beta^{\rm r}_{d,n} + R_{n,i}(d)$ and $\max_{d \in \{0, 1\}}E[R_{n,i}^2(d)] = o(1)$, then $\sigma_n^{\rm r,2}$ achieves the minimum variance, i.e., 
\begin{align*}
\lim_{n \to \infty} \sigma_n^{\rm r,2} = \sigma_2^2(Q) + \sigma_3^2(Q)~. 
\end{align*}
\end{theorem}

\begin{remark}
We recommend employing an iterative estimation procedure outlined by \cite{belloni2017program} to estimate $\hat \beta_{d, n}^{\rm r}$, in which the $m$-th step's penalty loadings are estimated based on the $(m-1)$th step's LASSO estimates.  Formally, this iterative procedure is described by the following algorithm:
\begin{algorithm}\label{algo:2} \hspace{1cm}
\begin{enumerate}
    \item[] \underline{Step 0}: Set $\hat \epsilon_{n,i}^{{\rm r},(0)}(d) = Y_i$ if $D_i = d$.
    \item[] \hspace{1in} $\vdots$
    \item[] \underline{Step $m$}: Compute $\hat \omega_{n,l}^{(m)}(d) = \sqrt{\frac{1}{n} \sum_{1 \leq i \leq 2n} I \{D_i = d\} \psi_{n, i, l}^2 (\hat \epsilon_{n, i}^{{\rm r}, (m - 1)}(d))^2}$ and compute $(\hat \alpha_{d, n}^{{\rm r}, (m)},\hat \beta_{d, n}^{{\rm r}, (m)})$ following \eqref{eq:LASSO2} with $\hat \omega_{n,l}^{(m)}$ as the penalty loadings, and $\hat \epsilon_{n,i}^{{\rm r},(m)}(d) = Y_i -\hat \alpha_{d, n}^{{\rm r},(m)}-\psi_i'\hat \beta_{d, n}^{{\rm r}, (m)}$ if $D_i = d$.
    \item[] \hspace{1in} $\vdots$
    \item[] \underline{Step $M$}: $\ldots$
    \item[] \underline{Step $M+1$}: Set $\hat \beta_{d, n}^{\rm r} = \hat \beta_{d, n}^{{\rm r}, (M)}$. 
\end{enumerate}
\end{algorithm}
As suggested by \cite{belloni2017program}, we set $M$ to be 15. We note that {\tt R} package \textbf{hdm} has a built-in option for this iterative procedure.  For this choice of penalty loadings, arguments similar to those in \cite{belloni2017program} can be used to verify \eqref{eq:loadings2} under ``matched pairs'' designs.
\end{remark}

\begin{remark}\label{rem:lasso2_f}
When the $\ell_1$-regularized adjustment is approximately correctly specified, Theorem \ref{thm:LASSO2} shows $\hat \Delta_n^{\rm r}$ achieves the minimum variance derived in Remark \ref{remark:equivalent}, and thus, is guaranteed to be weakly more efficient than the difference-in-means estimator ($\hat \Delta_n^{\rm unadj}$). When $W_{n,i}$ is fixed dimensional and $\psi_{n,i}$ consists of sieve basis functions of $(X_i,W_{n,i})$, the approximately correct specification usually holds. Specifically, under regularity conditions such as the smoothness of $E(Y_i(d)|X_i,W_{n,i})$, we can approximate $E(Y_i(d)|X_i,W_{n,i})$ by $\alpha_{d,n}^{\rm r}+ \psi_{n,i}'\beta^{\rm r}_{d,n}$ and $\beta^{\rm r}_{d,n}$ is automatically sparse in the sense that $||\beta^{\rm r}_{d,n}||_0 \ll n$. This means our regularized regression adjustment can select relevant sieve bases in nonparametric regression adjustments in a data-driven manner and automatically minimize the limiting variance of the corresponding ATE estimator.  
\end{remark}

\begin{remark}\label{rem:lasso2_f2}
When the dimension of $\psi_{n,i}$ is ultra-high (i.e., $p_n \gg n$) and the regularized adjustment is not approximately correctly specified, $\hat \Delta_n^{\rm r}$ suffers from \cite{freedman2008regression}'s critique that, theoretically, it is possible to be less efficient than $\hat \Delta_n^{\rm unadj}$. To overcome this problem, we consider an additional step in which we treat the regularized adjustments $(\psi_{n, i}' \hat \beta_{1, n}^{\rm r}, \psi_{n, i}'\hat \beta_{0, n}^{\rm r})$ as a two-dimensional covariate and refit a linear regression with pair fixed effects. Such a procedure has also been studied by \cite{cohen2020no-harm} in the setting with low-dimensional covariates and complete randomization. In fact, this strategy can improve upon general initial regression adjustments as long as \eqref{eq:rate}, \eqref{eq:L2}, and Assumption \ref{ass:md} are satisfied.    
\end{remark}

Theorem \ref{thm:refit} below shows the ``refit'' estimator for the ATE is weakly more efficient than both $\hat \Delta_n^{\rm unadj}$ and $\hat \Delta_n^{\rm r}$. To state the results, define $\Gamma_{n,i} = (\psi_{n, i}' \beta_{1, n}^{\rm r}, \psi_{n, i}' \beta_{0, n}^{\rm r})'$, $\hat \Gamma_{n,i} = (\psi_{n, i}' \hat \beta_{1, n}^{\rm r}, \psi_{n, i}'\hat \beta_{0, n}^{\rm r})$, and $\hat \Delta_n^{\rm refit}$ as the estimator in \eqref{eq:pfe} with $\psi_i$ replaced by $\hat \Gamma_{n,i}$. Note that $\hat \Delta_n^{\rm refit}$ remains numerically the same if we include the intercept $\hat \alpha_{d,n}^{\rm r}$ in the definition of $\hat \Gamma_{n,i}$. Following Remark \ref{remark:fogarty}, $\hat \Delta_n^{\rm refit}$ is the intercept in the linear regression of $(D_{\pi(2j-1)}-D_{\pi(2j-1)})(Y_{\pi(2j-1)} - Y_{\pi(2j)})$ on constant and $(D_{\pi(2j-1)}-D_{\pi(2j-1)})(\hat \Gamma_{n,\pi(2j-1)} - \hat \Gamma_{n,\pi(2j)})$. Replacing  $\hat \Gamma_{n,i}$ by $\hat \Gamma_{n,i} + (\hat \alpha_{1,n}^{\rm r},\hat \alpha_{0,n}^{\rm r})'$ will not change the regression estimators. 

The following assumption will be employed to control $\Gamma_{n,i}$ in our subsequent analysis:
\begin{assumption}\label{ass:refit}
For some $\kappa_1 > 0$ and $\kappa_2$,
\vspace{-.3cm}
\begin{align*}
& \inf_{n \geq 1}  \inf_{v \in \mathbf R^2} ||v||_2^{-2}v' E [\var[\Gamma_{n,i}|X_i]]v \geq \kappa_1 \\
& \sup_{n \geq 1}  \sup_{v \in \mathbf R^2} ||v||_2^{-2}v' E [\var[\Gamma_{n,i}|X_i]]v \leq \kappa_2~.
\end{align*}
\end{assumption}

The following theorem characterizes the behavior of $\hat \Delta_n^{\rm refit}$:
\begin{theorem} \label{thm:refit}
	Suppose $Q$ satisfies Assumption \ref{ass:Q} and the treatment assignment mechanism satisfies Assumptions \ref{ass:treatment}--\ref{ass:close}.  Further suppose Assumptions \ref{ass:LASSO2-moments}--\ref{ass:refit} hold. Then, \eqref{eq:rate}, \eqref{eq:L2}, and Assumption \ref{ass:md} are satisfied with $\hat m_{d, n}(X_i, W_{n, i}) = \hat \Gamma_{n, i}' \hat \beta_{n}^{\rm refit}$ and
	\[ m_{d, n}(X_i, W_{n, i}) = \Gamma_{n, i}'  \beta_{n}^{\rm refit} \]
	for $d \in \{0, 1\}$ and $n \geq 1$, where $\beta_{n}^{\rm refit} = (2E [\var[\Gamma_{n, i}|X_i]])^{-1}E[\cov[\Gamma_{n, i},Y_i(1)+Y_i(0)|X_i]]$. In addition, denote the asymptotic variance of $\hat \Delta_n^{\rm refit}$ as $\sigma_n^{\rm refit,2}$. Then, $\sigma_n^{\rm unadj,2} \geq \sigma_n^{\rm refit,2}$ and $\sigma_n^{\rm r,2}\geq \sigma_n^{\rm refit,2}$. 
\end{theorem}

\begin{remark} \label{remark:ridge}
It is possible to further relax the full rank condition in Assumption \ref{ass:refit} by running a ridge regression or truncating the minimum eigenvalue of the gram matrix in the refitting step. 
\end{remark}

\section{Simulations} \label{sec:simulations}
In this section, we conduct Monte Carlo experiments to assess the finite-sample performance of the inference methods proposed in the paper. In all cases, we follow \cite{bai2021inference} to consider tests of the hypothesis that 
\begin{align*}
	H_{0}: \Delta(Q) = \Delta_{0} \text{ versus } H_{1}: \Delta(Q) \neq \Delta_{0}~.
\end{align*} 
with $\Delta_{0}=0$ at nominal level $\alpha=0.05$.  

\subsection{Data Generating Processes}
We generate potential outcomes for $d\in\{0,1\}$ and $1\leq i\leq2n$ by the equation
\begin{equation}
	\ensuremath{Y_{i}(d)=\mu_{d}+m_{d}(X_{i},W_{i})+\sigma_{d}(X_{i},W_{i})\epsilon_{d,i}},~d=0,1,\label{eq:simulpart01}
\end{equation}
where $\mu_{d},m_{d}\left(X_{i},W_{i}\right),\sigma_{d}\left(X_{i},W_{i}\right)$,
and $\epsilon_{d,i}$ are specified in each model as follows. In each of the specifications, ($X_{i},W_{i},\epsilon_{0,i},\epsilon_{1,i}$) are i.i.d. across $i$. The number of pairs $n$ is equal to 100 and 200. The number of replications is 10,000.

\begin{description}
	\item [{Model 1}] $\left(X_{i},W_{i}\right)^\top=\left(\Phi\left(V_{i1}\right),\Phi\left(V_{i2}\right)\right)^\top$,
	where $\Phi(\cdot)$ is the standard normal distribution function and
	\[
	V_{i}\sim N\left(\left(\begin{array}{l}
		0\\
		0
	\end{array}\right),\left(\begin{array}{ll}
		1 & \rho\\
		\rho & 1
	\end{array}\right)\right),
	\]
	$m_{0}\left(X_{i},W_{i}\right)=\gamma\left(W_{i}-\frac{1}{2}\right)$; $m_{1}\left(X_{i},W_{i}\right)=m_{0}\left(X_{i},W_{i}\right)$;
	$\epsilon_{d,i}\sim N(0,1)$ for $d=0,1$; $\sigma_{0}\left(X_{i},W_{i}\right)=\sigma_{1}\left(X_{i},W_{i}\right)=1$. We set $\gamma=4$ and $\rho=0.2$.
	
	\item [{Model 2}] $\left(X_{i},W_{i}\right)^\top=\left(\Phi\left(V_{i1}\right),V_{1i}V_{i2}\right)^\top$, where $V_{i}$ is the same as in Model 1. $m_{0}\left(X_{i},W_{i}\right)=m_{1}\left(X_{i},W_{i}\right)=\gamma_{1}\left(W_{i} - \rho\right)+ \gamma_{2}\left(\Phi^{-1}\left(X_{i}\right)^2 - 1\right)$. $\epsilon_{d,i}\sim N(0,1)$ for $d=0,1$; $\sigma_{0}\left(X_{i},W_{i}\right)=\sigma_{1}\left(X_{i},W_{i}\right)=1$. $\left(\gamma_{1},\gamma_{2}\right)^\top=\left(1,2\right)^\top$ and $\rho=0.2$.
	
	\item [{Model 3}] The same as in Model 2, except that $m_{0}\left(X_{i},W_{i}\right)=m_{1}\left(X_{i},W_{i}\right)=\gamma_{1}\left(W_{i} - \rho\right)+ \gamma_{2}\left(\Phi\left(W_{i}\right) - \frac{1}{2}\right) + \gamma_{3}\left(\Phi^{-1}\left(X_{i}\right)^2  - 1\right)$ with $\left(\gamma_{1},\gamma_{2},\gamma_{3}\right)^\top=\left(\frac{1}{4},1,2\right)^\top$.
	
	\item [{Model 4}] $\left(X_{i},W_{i}\right)^\top=\left(V_{i1},V_{1i}V_{i2}\right)^\top$, where $V_{i}$ is the same as in Model 1. $m_{0}\left(X_{i},W_{i}\right)=m_{1}\left(X_{i},W_{i}\right)=\gamma_{1}\left(W_{i} - \rho\right)+ \gamma_{2}\left(\Phi\left(W_{i}\right) - \frac{1}{2}\right) + \gamma_{3}\left(X_{i}^2  - 1\right)$. $\epsilon_{d,i}\sim N(0,1)$ for $d=0,1$; $\sigma_{0}\left(X_{i},W_{i}\right)=\sigma_{1}\left(X_{i},W_{i}\right)=1$. $\left(\gamma_{1},\gamma_{2},\gamma_{3}\right)^\top=\left(2,1,2\right)^\top$ and $\rho=0.2$.
	
	\item [{Model 5}] The same as in Model 4, except that $m_{1}\left(X_{i},W_{i}\right)= m_{0}\left(X_{i},W_{i}\right) + \left(\Phi\left(X_{i}\right) - \frac{1}{2}\right)$.
	
	\item [{Model 6}] The same as in Model 5, except that $\sigma_{0}\left(X_{i},W_{i}\right)=\sigma_{1}\left(X_{i},W_{i}\right)=\left(\Phi\left(X_{i}\right)+0.5\right)$.
	
	\item [{Model 7}] $X_{i}=\left(V_{i1},V_{i2}\right)^\top$ and $W_{i}=\left(V_{i1}V_{i3},V_{i2}V_{i4}\right)^\top$, where $V_{i} \sim N(0,\Sigma)$ with $\text{dim}(V_{i})=4$ and $\Sigma$ consisting of 1 on the diagonal and $\rho$ on all other elements. $m_{0}\left(X_{i},W_{i}\right)=m_{1}\left(X_{i},W_{i}\right)=\gamma_{1}^\prime\left(W_{i} - \rho\right) + \gamma_{2}^\prime\left(\Phi\left(W_{i}\right) - \frac{1}{2}\right) + \gamma_{3}\left(X_{i1}^2 - 1\right)$ with $\gamma_{1}=\left(2,2\right)^\top,\gamma_{2}=\left(1,1\right)^\top, \gamma_{3}=1$. $\epsilon_{d,i}\sim N(0,1)$ for $d=0,1$; $\sigma_{0}\left(X_{i},W_{i}\right)=\sigma_{1}\left(X_{i},W_{i}\right)=1$. $\rho=0.2$.
	
	\item [{Model 8}] The same as in Model 7, except that $m_{1}\left(X_{i},W_{i}\right)= m_{0}\left(X_{i},W_{i}\right) + \left(\Phi\left(X_{i1}\right) - \frac{1}{2}\right)$.
	
	\item [{Model 9}] The same as in Model 8, except that $\sigma_{0}\left(X_{i},W_{i}\right)=\sigma_{1}\left(X_{i},W_{i}\right)=\left(\Phi\left(X_{i1}\right) +0.5\right)$
	
	\item [{Model 10}] $X_{i}=\left(\Phi\left(V_{i1}\right),\cdots,\Phi\left(V_{i4}\right)\right)^\top$ and $W_{i}=\left(V_{i1}V_{i5},V_{i2}V_{i6}\right)^\top$, where $V_{i}\sim N(0,\Sigma)$ with $\text{dim}(V_{i})=6$ and $\Sigma$ consisting of 1 on the diagonal and $\rho$ on all other elements. $m_{0}\left(X_{i},W_{i}\right)=m_{1}\left(X_{i},W_{i}\right)=\gamma_{1}^\prime\left(W_{i} - \rho\right) + \gamma_{2}^\prime\left(\Phi\left(W_{i}\right) - \frac{1}{2}\right) + \gamma_{3}^\prime\left(\left(\Phi^{-1}\left(X_{i1}\right)^2,\Phi^{-1}\left(X_{i2}\right)^2\right)^\top - 1\right)$ with $\gamma_{1}=\left(1,1\right)^\top,\gamma_{2}=\left(\frac{1}{2},\frac{1}{2}\right)^\top, \gamma_{3}=\left(\frac{1}{2},\frac{1}{2}\right)^\top$. $\epsilon_{d,i}\sim N(0,1)$ for $d=0,1$; $\sigma_{0}\left(X_{i},W_{i}\right)=\sigma_{1}\left(X_{i},W_{i}\right)=1$.
	
	\item [{Model 11}] The same as in Model 10, except that $m_{1}\left(X_{i},W_{i}\right)= m_{0}\left(X_{i},W_{i}\right) + \frac{1}{4}\sum_{j=1}^{4}\left(X_{ij} - \frac{1}{2}\right)$.
	
	\item [{Model 12}] $X_{i}=\left(\Phi\left(V_{i1}\right),\cdots,\Phi\left(V_{i4}\right)\right)^\top$ and $W_{i}=\left(V_{i1}V_{i41},\cdots,V_{i40}V_{i80}\right)^\top$, where $V_{i} \sim N(0,\Sigma)$ with $\text{dim}(V_{i})=80$. $\Sigma$ is the Toeplitz matrix
	\begin{align*}
		\Sigma = \begin{pmatrix}
			1 & 0.5 & 0.5^2 &\cdots & 0.5^{79} \\
			0.5 & 1 & 0.5 & \cdots & 0.5^{78} \\
			0.5^2 & 0.5 & 1 & \cdots & 0.5^{77} \\
			\vdots & \vdots & \vdots & \ddots & \vdots \\
			0.5^{79} & 0.5^{78} & 0.5^{77} & \cdots & 1
		\end{pmatrix}.
	\end{align*}
	$m_{0}\left(X_{i},W_{i}\right)=m_{1}\left(X_{i},W_{i}\right)=\gamma_{1}^{\prime}W_{i} + \gamma_{2}^\prime\left(\Phi^{-1}\left(X_{i}\right)^2 - 1\right)$, $\gamma_{1}=\left(\frac{1}{1^2},\frac{1}{2^2},\cdots,\frac{1}{40^2}\right)^\top$ with $\text{dim}(\gamma_{1})=40$, and $\gamma_{2}=\left(\frac{1}{8},\frac{1}{8},\frac{1}{8},\frac{1}{8}\right)^\top$ with $\text{dim}(\gamma_{2})=4$. $\epsilon_{d,i}\sim N(0,1)$ for $d=0,1$; $\sigma_{0}\left(X_{i},W_{i}\right)=\sigma_{1}\left(X_{i},W_{i}\right)=1$.
	
	\item [{Model 13}] The same as in Model 12, except that $m_{0}\left(X_{i},W_{i}\right)=m_{1}\left(X_{i},W_{i}\right)=\gamma_{1}^{\prime}W_{i} + \gamma_{2}^\prime\left(\Phi\left(W_{i}\right) - \frac{1}{2}\right) + \gamma_{3}^\prime\left(\Phi^{-1}\left(X_{i}\right)^2 - 1\right)$, $\gamma_{1}=\left(\frac{1}{1^2},\cdots,\frac{1}{40^2}\right)^\top$,  $\gamma_{2}=\frac{1}{8}\left(\frac{1}{1^2},\cdots,\frac{1}{40^2}\right)^\top$, and $\gamma_{3}=\left(\frac{1}{8},\frac{1}{8},\frac{1}{8},\frac{1}{8}\right)^\top$ with $\text{dim}(\gamma_{1})=\text{dim}(\gamma_{2})=40$ and $\text{dim}(\gamma_{3})=4$.
	
	\item [{Model 14}] The same as in Model 13, except that $m_{1}\left(X_{i},W_{i}\right)= m_{0}\left(X_{i},W_{i}\right) + \sum_{j=1}^{4}\frac{1}{j^2}\left(X_{ij}- \frac{1}{2}\right)$.
	
	\item [{Model 15}] The same as in Model 14, except that $\sigma_{0}\left(X_{i},W_{i}\right)=\sigma_{1}\left(X_{i},W_{i}\right)=\left(X_{i1}+0.5\right)$.

\end{description}

It is worth noting that Models 1, 2, 3, 4,  7, 10, 12, and 13 imply homogeneous treatment effects because $m_{1}\left(X_{i},W_{i}\right) = m_{0}\left(X_{i},W_{i}\right)$. Among them, $E[Y_i(a)|X_i,W_i] - E[Y_i(a)|X_i]$ is linear in $W_i$ in Models 1, 2, and 12. Models 5, 8, 11, and 14 have heterogeneous but homoscedastic treatment effects. In Models 6, 9, and 15, however, the implied treatment effects are both heterogeneous and heteroscedastic. Models 12-15 contain high-dimensional covariates. 

We follow \cite{bai2021inference} to match pairs. Specifically, if $\text{dim}\left(X_i\right)=1$, we match pairs by sorting $X_i, i = 1, \ldots, 2n$. If $\text{dim}\left(X_i\right)>1$, we match pairs by the permutation $\pi$ calculated using the {\tt R} package \emph{nbpMatching}. For more details, see \citet[Section 4]{bai2021inference}. After matching the pairs, we flip coins to randomly select one unit within each pair for treatment and another for control.

\subsection{Estimation and Inference}
\label{sec:sims2}
We set $\mu_0=0$ and $\mu_{1}=\Delta$, where $\Delta=0$ and $1/4$ are used to illustrate the size and power, respectively. Rejection probabilities in percentage points are presented. To further illustrate the efficiency gains obtained by regression adjustments, in Figure \ref{fig:bar}, we plot the average standard error reduction in percentage relative to the standard error of the estimator without adjustments for various estimation methods.  

Specifically, we consider the following adjusted estimators. 
\begin{enumerate}[(i)]
	\item unadj: the estimator with no adjustments. In this case, our standard error is identical to the adjusted standard error proposed by \cite{bai2021inference}.
	\item na\"ive: the linear adjustments with regressors $W_i$ but without pair dummies. 
	\item na\"ive2: the linear adjustments with $X_i$ and $W_i$ regressors but without pair dummies. 
	\item pfe: the linear adjustments with regressors $W_i$ and pair dummies.
        \item refit: refit the $\ell_1$-regularized adjustments by linear regression with pair dummies.
        
\end{enumerate}
See Section \ref{sec:sims-details} in the Online Supplement for the regressors used in the regularized adjustments.

For Models 1-11, we examine the performance of estimators (i)-(v). For Models 12-15, we assess the performance among estimators (i) and (v) in high-dimensional settings. Note that the adjustments are misspecified for almost all the models. The only exception is Model 1, for which the linear adjustment in $W_i$ is correctly specified because $m_d(X_i,W_i)$ is just a linear function of $W_i$. 

\subsection{Simulation Results}

Tables \ref{tab:no_LASSO_tab} and \ref{tab:no_LASSO_tab2} report rejection probabilities at the 0.05 level and power of the different methods for Models 1--11 when $n$ is 100 and 200, respectively. Several patterns emerge. First, for all the estimators, the rejection rates under $H_0$ are close to the nominal level even when $n=100$ and with misspecified adjustments. This result is expected because all the estimators take into account the dependence structure arising in the ``matched pairs'' design, consistent with the findings in \cite{bai2021inference}.

Second, in terms of power,  ``pfe'' is higher than ``unadj'', ``na\"ive'', and ``na\"ive2'' for all eleven models, as predicted by our theory. This finding confirms that ``pfe'' is the optimal linear adjustment and will not degrade the precision of the ATE estimator. In contrast, we observe that ``na\"ive'' and ``na\"ive2'' in Model 3 are even less powerful than the unadjusted estimator ``unadj''. Figure \ref{fig:bar} further confirms that these two methods inflate the estimation standard error. This result echoes Freedman's critique \citep{freedman2008regression} that careless regression adjustments may degrade the estimation precision. Our ``pfe'' addresses this issue because it has been proven to be weakly more efficient than the unadjusted estimator.

Third, the improvement of power for ``pfe'' is mainly due to the reduction of estimation standard errors, which can be more than 50\% as shown in Figure \ref{fig:bar} for Models 4--9. This means that the length of the confidence interval of the ``pfe'' estimator is just half of that for the ``unadj'' estimator.  Note the standard error of the ``unadj'' estimator is the one proposed by \cite{bai2021inference}, which has already been adjusted to account for the cross-sectional dependence created in pair matching. The extra 50\% reduction is therefore produced purely by the regression adjustment. For Models 10-11, the reduction of standard errors achieved by ``pfe'' is more than 40\% as well. For Model 1, the linear regression is correctly specified so that all three methods achieve the global minimum asymptotic variance and maximum power. For Model 2, $m_d(X_i,W_i) - E[m_d(X_i,W_i)|X_i] = \gamma (W_i - E[W_i|X_i])$ so that the linear adjustment $\gamma W_i$ satisfies the conditions in Theorem \ref{thm:main}. Therefore, ``pfe'', as the best linear adjustment, is also the best adjustment globally, achieving the global minimum asymptotic variance and maximum power. In contrast, ``na\"ive'' and ``na\"ive2'' are not the best linear adjustment and therefore less powerful than ``pfe'' because of the omitted pair dummies.

Finally,  the ``refit'' method has the best power for most models as they automatically achieve the global minimum asymptotic variance when the dimension of $W_i$ is fixed. 

Tables \ref{tab:LASSO_tab} and \ref{tab:LASSO_tab2} report the size and power for the ``refit'' adjustments when both $W_i$ and $X_i$ are high-dimensional. We see that the size under the null is close to the nominal 5\% while the power for the adjusted estimator is higher than the unadjusted one. Figure \ref{fig:bar} further illustrates the reduction of the standard error is more than 30\% for all high-dimensional models. 

\newcolumntype{L}{>{\raggedright\arraybackslash}X}
\newcolumntype{C}{>{\centering\arraybackslash}X}

\begin{table}[t!]
	\caption{Rejection probabilities for Models 1-11 when $n=100$}
	\vspace{1ex}
	
	\centering{}%
	\begin{tabularx}{1\textwidth}{LCCCCCCCCCC}
		\toprule
		& \multicolumn{5}{c}{$H_{0}$: $\Delta=0$} & \multicolumn{5}{c}{$H_{1}$: $\Delta=1/4$}\\
		\cmidrule(lr){2-6} \cmidrule(lr){7-11}
		Model & \multicolumn{1}{c}{unadj} & \multicolumn{1}{c}{na\"ive} & \multicolumn{1}{c}{na\"ive2} & \multicolumn{1}{c}{pfe} & \multicolumn{1}{c}{refit}  & \multicolumn{1}{c}{unadj}  & \multicolumn{1}{c}{na\"ive} & \multicolumn{1}{c}{na\"ive2} & \multicolumn{1}{c}{pfe} & \multicolumn{1}{c}{refit} \\
		\midrule
1                       & 5.47      & 5.57      & 5.63      & 5.76      & 5.84      & 22.48     & 43.89     & 43.95     & 43.91     & 43.92      \\ 
2                       & 4.96      & 5.26      & 5.30      & 5.47      & 5.32      & 23.32     & 28.02     & 27.96     & 37.21     & 33.12      \\ 
3                       & 4.99      & 5.28      & 5.24      & 5.48      & 5.27      & 32.19     & 27.88     & 27.96     & 37.34     & 36.29      \\ 
4                       & 5.31      & 5.28      & 5.28      & 5.48      & 5.79      & 11.78     & 27.88     & 28.03     & 37.34     & 43.28      \\ 
5                       & 5.43      & 5.09      & 5.08      & 5.49      & 5.78      & 11.87     & 27.72     & 27.88     & 36.69     & 43.08      \\ 
6                       & 5.28      & 5.43      & 5.41      & 5.58      & 5.79      & 11.78     & 26.67     & 26.72     & 34.71     & 40.29      \\ 
7                       & 5.64      & 5.63      & 5.62      & 5.98      & 6.04      & 9.24      & 34.55     & 34.65     & 37.96     & 42.08      \\ 
8                       & 5.63      & 5.54      & 5.51      & 6.03      & 6.17      & 9.28      & 34.11     & 34.42     & 37.22     & 41.29      \\ 
9                       & 5.74      & 5.69      & 5.76      & 6.19      & 5.89      & 8.99      & 32.39     & 32.30     & 35.42     & 38.75      \\ 
10                      & 5.24      & 5.78      & 5.73      & 6.05      & 6.04      & 14.27     & 30.80     & 30.75     & 32.02     & 32.51      \\ 
11                      & 5.19      & 5.78      & 5.72      & 6.07      & 5.95      & 14.36     & 30.60     & 30.49     & 32.21     & 32.81      \\ 
		\bottomrule
	\end{tabularx}
	
	\label{tab:no_LASSO_tab}
\end{table}

\begin{table}[ht!]
	\caption{Rejection probabilities for Models 12-15 when $n=100$}
	\vspace{1ex}
	
	\centering{}%
	\begin{tabularx}{1\textwidth}{LCCCC}
		\toprule
		& \multicolumn{2}{c}{$H_{0}$: $\Delta=0$} & \multicolumn{2}{c}{$H_{1}$: $\Delta=1/4$}\\
		\cmidrule(lr){2-3} \cmidrule(lr){4-5}
		& \multicolumn{1}{c}{unadj} & \multicolumn{1}{c}{refit}  & \multicolumn{1}{c}{unadj} & \multicolumn{1}{c}{refit} \\
		\midrule
12                      & 5.35      & 6.12      & 22.01     & 42.56      \\ 
13                      & 5.31      & 6.11      & 21.47     & 42.47      \\ 
14                      & 5.24      & 6.07      & 21.39     & 41.14      \\ 
15                      & 5.31      & 6.23      & 20.73     & 38.67      \\   

		\bottomrule
	\end{tabularx}
	
	\label{tab:LASSO_tab}
\end{table}

\begin{table}[ht!]
	\caption{Rejection probabilities for Models 1-11 when $n=200$}
	\vspace{1ex}
	
	\centering{}%
	\begin{tabularx}{1\textwidth}{LCCCCCCCCCC}
		\toprule
		& \multicolumn{5}{c}{$H_{0}$: $\Delta=0$} & \multicolumn{5}{c}{$H_{1}$: $\Delta=1/4$}\\
		\cmidrule(lr){2-6} \cmidrule(lr){7-11}
		Model & \multicolumn{1}{c}{unadj} & \multicolumn{1}{c}{na\"ive} & \multicolumn{1}{c}{na\"ive2} & \multicolumn{1}{c}{pfe} & \multicolumn{1}{c}{refit}  & \multicolumn{1}{c}{unadj}  & \multicolumn{1}{c}{na\"ive} & \multicolumn{1}{c}{na\"ive2} & \multicolumn{1}{c}{pfe} & \multicolumn{1}{c}{refit} \\
		\midrule
1                       & 5.08      & 5.04      & 5.10      & 5.21      & 5.31      & 38.94     & 70.35     & 70.36     & 70.32     & 70.30      \\ 
2                       & 5.69      & 5.28      & 5.28      & 5.24      & 5.40      & 40.31     & 49.25     & 49.32     & 65.36     & 57.87      \\ 
3                       & 5.44      & 5.29      & 5.30      & 5.35      & 5.41      & 56.89     & 49.43     & 49.51     & 64.96     & 62.42      \\ 
4                       & 5.45      & 5.29      & 5.29      & 5.35      & 5.20      & 18.55     & 49.43     & 49.67     & 64.96     & 69.96      \\ 
5                       & 5.45      & 5.24      & 5.18      & 5.19      & 5.29      & 18.41     & 48.65     & 48.80     & 64.11     & 69.09      \\ 
6                       & 5.62      & 5.32      & 5.31      & 5.35      & 5.43      & 18.19     & 46.71     & 46.67     & 61.09     & 65.98      \\ 
7                       & 5.24      & 5.51      & 5.46      & 5.34      & 5.49      & 11.86     & 60.73     & 60.63     & 65.14     & 69.24      \\ 
8                       & 5.23      & 5.49      & 5.47      & 5.35      & 5.65      & 11.84     & 60.00     & 60.10     & 64.93     & 68.02      \\ 
9                       & 5.30      & 5.58      & 5.57      & 5.66      & 5.81      & 11.90     & 57.25     & 57.28     & 61.61     & 64.88      \\ 
10                      & 5.34      & 5.19      & 5.15      & 5.25      & 5.31      & 23.95     & 55.49     & 55.44     & 56.64     & 56.43      \\ 
11                      & 5.41      & 5.36      & 5.32      & 5.34      & 5.41      & 23.88     & 55.01     & 55.05     & 56.31     & 56.18      \\ 
		\bottomrule
	\end{tabularx}
	
	\label{tab:no_LASSO_tab2}
\end{table}

\begin{table}[ht!]
	\caption{Rejection probabilities for Models 12-15 when $n=200$}
	\vspace{1ex}
	
	\centering{}%
	\begin{tabularx}{1\textwidth}{LCCCC}
	\toprule
		& \multicolumn{2}{c}{$H_{0}$: $\Delta=0$} & \multicolumn{2}{c}{$H_{1}$: $\Delta=1/4$}\\
		\cmidrule(lr){2-3} \cmidrule(lr){4-5}
		& \multicolumn{1}{c}{unadj} & \multicolumn{1}{c}{refit}  & \multicolumn{1}{c}{unadj} & \multicolumn{1}{c}{refit} \\
		
		\midrule
12                      & 4.97      & 5.22      & 38.91     & 68.10      \\ 
13                      & 4.95      & 5.19      & 38.04     & 68.06      \\ 
14                      & 5.01      & 5.24      & 37.65     & 66.69      \\ 
15                      & 5.15      & 5.40      & 36.61     & 63.79      \\  
		\bottomrule
	\end{tabularx}
	\label{tab:LASSO_tab2}
\end{table}

\begin{figure}[ht!]
	\includegraphics[width=5.5in, height=5.0in]{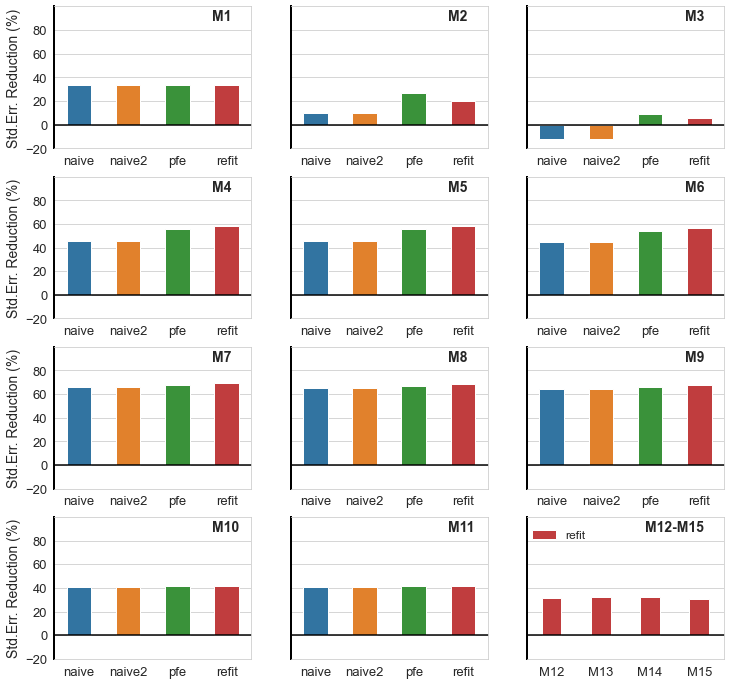}
	\centering
	\caption{Average Standard Error Reduction in Percentage under $H_{1}$ when $n=200$}
	\label{fig:bar}
	\vspace{-1ex}
	\justify
	Notes: The figure plots average standard error reduction in percentage achieved by regression adjustments relative to ``unadj'' under $H_{1}$ for Models 1-15 when $n=200$. 
\end{figure}

\section{Empirical Illustration} \label{sec:empirical}

In this section, we revisit the randomized experiment with a matched pairs design conducted in \cite{groh2016macroinsurance}. In the paper, they examined the impact of macroinsurance on microenterprises. Here, we apply the covariate adjustment methods developed in this paper to their data and reinvestigate the average effect of macroinsurance on three outcome variables: the microenterprises' monthly profits, revenues, and investment.

The subjects in the experiment are microenterprise owners, who were the clients of the largest microfinance institution in Egypt. In the randomization, after an exact match of gender and the institution's branch code, those clients were grouped into pairs by applying an optimal greedy algorithm to additional 13 matching variables. Within each pair, a macroinsurance product was then offered to one randomly assigned client, and the other acted as a control. Based on the pair identities and all the matching variables, we re-order the pairs in our sample according to the procedure described in Section 5.1 of \cite{Jiang2022QTE}. The resulting sample contains 2824 microenterprise owners, that is, 1412 pairs of them.\footnote{See \cite{groh2016macroinsurance} and \cite{Jiang2022QTE} for more details.}

Table \ref{tab:emp_ate} reports the ATEs with the standard errors (in parentheses) estimated by different methods. Among them, ``GM'' corresponds to the method used in \cite{groh2016macroinsurance}.\footnote{\cite{groh2016macroinsurance} estimated the effect by regression with regressors including some baseline variables, a dummy for missing observations, and dummies for the pairs. Specifically, for profits and revenues, the regressors are the baseline value for the outcome of interest, a dummy for missing observations, and pair dummies; for investment, the regressors only include pair dummies. The standard errors for the ``GM'' ATE estimate are calculated by the usual heteroskedastity-consistent estimator. The ``GM'' results in Table \ref{tab:emp_ate} were obtained by applying the Stata code provided by \cite{groh2016macroinsurance}.} The description of other methods is similar to that in Section \ref{sec:sims2}.\footnote{Specifically:
\begin{enumerate}[(i)]
    \item $X_i$ includes gender and 13 additional matching variables for all adjustments. Three of the matching variables are continuous, and the others are dummies.
    \item To maintain comparability, we keep $X_i$ and $W_i$ consistent across all adjustments except for ``refit'' for each outcome variable. For profits and revenue, $W_i$ includes the baseline value for the outcome of interest, a dummy for whether the firm is above the 95th percentile of the control firms' distributions of the outcome variable, and a dummy for missing observations. For investment, $W_i$ includes all the covariates used for the first two outcome variables. 
    \item For ``refit'', we intentionally expand the dimensions of $W_i$. In addition to the baseline values used in the other adjustments and the dummy variables for missing observations, the $W_i$ used in ``refit'' also includes the interaction of the continuous original $W_i$ variables with three continuous variables and the first three discrete variables in $X_i$. 
    \item All the continuous variables in $X_i$ and $W_i$ are standardized initially when the regression-adjusted estimators are employed.
\end{enumerate} } The results in this table prompt the following observations. 

First, aligning with our theoretical and simulation findings, we observe that the standard errors associated with the covariate-adjusted ATEs, particularly those for the ``naïve2'' and ``pfe'' estimates, are generally lower compared to the ATE estimate without any adjustment. This pattern is consistent across nearly all the outcome variables. To illustrate, when examining the revenue outcome, the standard errors for the ``pfe'' estimates are 10.2\% smaller than those for the unadjusted ATE estimate. 

Second, the standard errors of the ``refit'' estimates are consistently smaller than those of the unadjusted ATE estimate across all the outcome variables. For example, when profits are the outcome variable, the ``refit'' estimates exhibit standard errors 7.5\% smaller than those of the unadjusted ATE estimate. Moreover, compared with those of the ``pfe'' estimates, the standard errors of ``refit'' are slightly smaller.

\begin{table}[H]
	\centering
	\caption{Impacts of Macronsurance for Microenterprises}
	\vspace{1ex}
	\begin{tabularx}{1\textwidth}{LCCCCCCC}
		\toprule
		Y & n & unadj & GM & na\"ive & na\"ive2 & pfe & refit\\
		\midrule
Profits                 & 1322      & -85.65    & -50.88    & -41.69    & -50.97    & -51.60    & -55.13     \\ 
                        &           & (49.43)   & (46.46)   & (47.22)   & (45.49)   & (46.94)   & (45.71)    \\ 
Revenue                 & 1318      & -838.60   & -660.16   & -611.75   & -610.80   & -635.80   & -600.97    \\ 
                        &           & (319.02)  & (284.02)  & (286.93)  & (282.93)  & (286.50)  & (284.60)   \\ 
Investment              & 1410      & -66.60    & -66.60    & -49.37    & -50.72    & -67.31    & -58.77     \\ 
                        &           & (118.93)  & (118.66)  & (119.23)  & (118.97)  & (118.88)  & (118.84)   \\

		\bottomrule
	\end{tabularx} \\
	
	\vspace{-1ex}
	\justify
	Notes: The table reports the ATE estimates of the effect of macroinsurance for microenterprises. Standard errors are in parentheses. 
	\label{tab:emp_ate}
\end{table}

\section{Conclusion}
This paper considers covariate adjustment for the estimation of average treatment effect in ``matched pairs'' designs when covariates other than the matching variables are available. When the dimension of these covariates is low, we suggest estimating the average treatment effect by a linear regression of the outcome on treatment status and covariates, controlling for pair fixed effects.  We show that this estimator is no worse than the simple difference-in-means estimator in terms of efficiency. When the dimension of these covariates is high, we suggest a two-step estimation procedure: in the first step, we run $\ell_1$-regularized regressions of outcome on covariates for the treated and control groups separately and obtain the fitted values for both potential outcomes, and in the second step, we estimate the average treatment effect by refitting a linear regression of outcome on treatment status and regularized adjustments from the first step, controlling for the pair fixed effects. We show that the final estimator is no worse than the simple difference-in-means estimator in terms of efficiency. When the conditional mean models are approximately correctly specified, this estimator further achieves the minimum variance as if all relevant covariates are used to form pairs in the experiment design stage.  We take the choice of variables to use in forming pairs as given and focus on how to obtain more efficient estimators of the average treatment effect in the analysis stage. Our paper is therefore silent on the important question of how to choose the relevant matching variables in the design stage. This topic is left for future research. 

\clearpage
\newpage
\appendix

\section{Proofs of Main Results} \label{sec:proof_main}
In the appendix, we use $a_n \lesssim b_n$ to denote there exists $c > 0$ such that $a_n \leq c b_n$.

\subsection{Proof of Theorem \ref{thm:main}}
\underline{Step 1: Decomposition by recursive conditioning}

To begin, note
\begin{align}
\nonumber \hat \mu_n(1) & = \frac{1}{2n} \sum_{1 \leq i \leq 2n} (2 D_i (Y_i(1) - \hat m_{1, n}(X_i, W_i)) + \hat m_{1, n}(X_i, W_i)) \\
\nonumber & = \frac{1}{2n} \sum_{1 \leq i \leq 2n} (2 D_i Y_i(1) - (2 D_i - 1) \hat m_{1, n}(X_i, W_i)) \\
\nonumber & = \frac{1}{2n} \sum_{1 \leq i \leq 2n} (2 D_i Y_i(1) - (2 D_i - 1) m_{1, n}(X_i, W_i)) + o_P(n^{-1/2}) \\
\label{eq:mu1} & = \frac{1}{2n} \sum_{1 \leq i \leq 2n} (2 D_i Y_i(1) - D_i m_{1, n}(X_i, W_i) - (1 - D_i) m_{1, n}(X_i, W_i)) + o_P(n^{-1/2})~,
\end{align}
where the third equality follows from \eqref{eq:rate}. Similarly,
\begin{equation} \label{eq:mu0}
\hat \mu_n(0) = \frac{1}{2n} \sum_{1 \leq i \leq 2n} (2 (1 - D_i) Y_i(0) - D_i m_{0, n}(X_i, W_i) - (1 - D_i) m_{0, n}(X_i, W_i)) + o_P(n^{-1/2})~.
\end{equation}
It follows from \eqref{eq:mu1}--\eqref{eq:mu0} that
\begin{equation} \label{eq:delta-phi}
\hat \Delta_n = \frac{1}{n} \sum_{1 \leq i \leq 2n} D_i \phi_{1, n, i} - \frac{1}{n} \sum_{1 \leq i \leq 2n} (1 - D_i) \phi_{0, n, i} + o_P(n^{-1/2})~,
\end{equation}
where
\begin{align*}
\phi_{1, n, i} & = Y_i(1) - \frac{1}{2} (m_{1, n}(X_i, W_i) + m_{0, n}(X_i, W_i)) \\
\phi_{0, n, i} & = Y_i(0) - \frac{1}{2} (m_{1, n}(X_i, W_i) + m_{0, n}(X_i, W_i))~.
\end{align*}
Next, consider
\begin{align*}
\mathbb L_n = \frac{1}{2\sqrt n} \sum_{1 \leq i \leq 2n} (2 D_i - 1) E[m_{1, n}(X_i, W_i) + m_{0, n}(X_i, W_i) | X_i]~.
\end{align*}
For simplicity, define $M_{d, n}(X_i) = E[m_{d, n}(X_i, W_i) | X_i]$ for $d \in \{0, 1\}$. It follows from Assumption \ref{ass:treatment} that $E[\mathbb L_n | X^{(n)}] = 0$. On the other hand,
\begin{align*}
\var[\mathbb L_n | X^{(n)}] & = \frac{1}{4n} \sum_{1 \leq j \leq n} \left ( M_{1, n}(X_{\pi(2j - 1)}) + M_{0, n}(X_{\pi(2j - 1)}) - (M_{1, n}(X_{\pi(2j)}) + M_{0, n}(X_{\pi(2j)})) \right )^2 \\
& \lesssim \frac{1}{n} \sum_{1 \leq j \leq n} |M_{1, n}(X_{\pi(2j - 1)}) - M_{1, n}(X_{\pi(2j)})|^2 + \frac{1}{n} \sum_{1 \leq j \leq n} |M_{0, n}(X_{\pi(2j - 1)}) - M_{0, n}(X_{\pi(2j)})|^2 \\
& \stackrel{P}{\to} 0~,
\end{align*}
where the inequality follows from $(a + b)^2 \leq 2(a^2 + b^2)$ and the convergence follows from Assumptions \ref{ass:close} and \ref{ass:md}(c). By Markov's inequality and the fact that $E[\mathbb L_n | X^{(n)}] = 0$, for any $\epsilon > 0$,
\[ P \{|\mathbb L_n| > \epsilon | X^{(n)}\} \leq \frac{\var[\mathbb L_n | X^{(n)}]}{\epsilon^2} \stackrel{P}{\to} 0~. \]
Since probabilities are bounded, we have $\mathbb L_n = o_P(1)$. This fact, together with \eqref{eq:delta-phi}, imply
\[ \sqrt n(\hat \Delta_n - \Delta(Q)) = A_n - B_n + C_n - D_n~, \]
where
\begin{align*}
A_n & = \frac{1}{\sqrt n} \sum_{1 \leq i \leq 2n} \left ( D_i \phi_{1, n, i} - E[D_i \phi_{1, n, i} | X^{(n)}, D^{(n)}] \right ) \\
B_n & = \frac{1}{\sqrt n} \sum_{1 \leq i \leq 2n} \left ( (1 - D_i) \phi_{0, n, i} - E[(1 - D_i) \phi_{0, n, i} | X^{(n)}, D^{(n)}] \right ) \\
C_n & = \frac{1}{\sqrt n} \sum_{1 \leq i \leq 2n} D_i (E[Y_i(1) | X_i] - E[Y_i(1)]) \\
D_n & = \frac{1}{\sqrt n} \sum_{1 \leq i \leq 2n} (1 - D_i)(E[Y_i(0) | X_i] - E[Y_i(0)])~.
\end{align*}
Note that conditional on $X^{(n)}$ and $D^{(n)}$, $A_n$ and $B_n$ are independent while $C_n$ and $D_n$ are constants.

\noindent \underline{Step 2: Conditional central limit theorems}

We first analyze the limiting behavior of $A_n$. Define
\[ s_n^2 = \sum_{1 \leq i \leq 2n} D_i \var[\phi_{1, n, i} | X_i]~. \]
Note by Assumption \ref{ass:treatment} that $s_n^2 = n \var[A_n | X^{(n)}, D^{(n)}]$. We proceed verify the Lindeberg condition for $A_n$ conditional on $X^{(n)}$ and $D^{(n)}$, i.e., we show that for every $\epsilon > 0$,
\begin{equation} \label{eq:lind}
\frac{1}{s_n^2} \sum_{1 \leq i \leq 2n} E [ |D_i (\phi_{1, n, i} - E[\phi_{1, n, i} | X_i])|^2 I \{|D_i (\phi_{1, n, i} - E[\phi_{1, n, i} | X_i])| > \epsilon s_n\} | X^{(n)}, D^{(n)} ] \stackrel{P}{\to} 0~.
\end{equation}
To that end, first note Lemma \ref{lem:sn} implies
\begin{equation} \label{eq:sn-limit}
\frac{s_n^2}{n E[\var[\phi_{1, n, i} | X_i]]} \stackrel{P}{\to} 1~.
\end{equation}
\eqref{eq:sn-limit} and Assumption \ref{ass:md}(a) imply that for all $\lambda > 0$,
\begin{equation} \label{eq:sn-eps}
P \{\epsilon s_n > \lambda\} \stackrel{P}{\to} 1~.
\end{equation}
Furthermore, for some $c > 0$,
\begin{equation} \label{eq:sn-n}
P \left \{ \frac{s_n^2}{n} > c \right \} \to 1~. 
\end{equation}
Next, note for any $\lambda > 0$ and $\delta_1 > 0$, the left-hand side of \eqref{eq:lind} can be written as
\begin{align}
\nonumber & \frac{1}{s_n^2 / n} \frac{1}{n} \sum_{1 \leq i \leq 2n: D_i = 1} E[ |\phi_{1, n, i} - E[\phi_{1, n, i} | X_i]|^2 I \{|\phi_{1, n, i} - E[\phi_{1, n, i} | X_i]| > \epsilon s_n\} | X^{(n)}, D^{(n)}] \\
\nonumber & \leq \frac{1}{s_n^2 / n} \frac{1}{n} \sum_{1 \leq i \leq 2n} E[ |\phi_{1, n, i} - E[\phi_{1, n, i} | X_i]|^2 I \{|\phi_{1, n, i} - E[\phi_{1, n, i} | X_i]| > \epsilon s_n\} | X^{(n)}, D^{(n)}] \\
\nonumber & \leq \frac{1}{c} \frac{1}{n} \sum_{1 \leq i \leq 2n} E[ |\phi_{1, n, i} - E[\phi_{1, n, i} | X_i]|^2 I \{|\phi_{1, n, i} - E[\phi_{1, n, i} | X_i]| > \lambda\} | X^{(n)}, D^{(n)}] + o_P(1) \\
\label{eq:lind1} & \leq \frac{2}{c} \frac{1}{2n} \sum_{1 \leq i \leq 2n} E[ |\phi_{1, n, i} - E[\phi_{1, n, i} | X_i]|^2 I \{|\phi_{1, n, i} - E[\phi_{1, n, i} | X_i]| > \lambda\} | X_i] + o_P(1)~,
\end{align}
where the first inequality follows by inspection, the second follows from \eqref{eq:sn-eps}--\eqref{eq:sn-n}, and the last follows from Assumption \ref{ass:treatment}. We then argue
\begin{multline} \label{eq:lind2}
\frac{1}{2n} \sum_{1 \leq i \leq 2n} E[ |\phi_{1, n, i} - E[\phi_{1, n, i} | X_i]|^2 I \{|\phi_{1, n, i} - E[\phi_{1, n, i} | X_i]| > \lambda\} | X_i] \\
= E[|\phi_{1, n, i} - E[\phi_{1, n, i} | X_i]|^2 I \{|\phi_{1, n, i} - E[\phi_{1, n, i} | X_i]| > \lambda\}] + o_P(1)~.
\end{multline}
To this end, we once again verify the Lindeberg condition in Lemma 11.4.2 of \cite{lehmann:romano:tsh:2005}. Note
\[ |\phi_{1, n, i} - E[\phi_{1, n, i} | X_i]|^2 I \{|\phi_{1, n, i} - E[\phi_{1, n, i} | X_i]| > \lambda\} \leq |\phi_{1, n, i} - E[\phi_{1, n, i} | X_i]|^2~. \]
Therefore, in light of Lemma \ref{lem:ui}, we only need to verify
\begin{equation} \label{eq:lind-ui}
\lim_{\gamma \to \infty} \limsup_{n \to \infty} E[|\phi_{1, n, i} - E[\phi_{1, n, i} | X_i]|^2 I \{|\phi_{1, n, i} - E[\phi_{1, n, i} | X_i]|^2 > \gamma\}] = 0~,
\end{equation}
which follows immediately from Lemma \ref{lem:lind-ui}.

Another application of \eqref{eq:lind-ui} implies \eqref{eq:lind}. Lindeberg's central limit theorem and \eqref{eq:sn-limit} then imply that
\[ \sup_{t \in \mathbf R} |P \{A_n / \sqrt{E[\var[\phi_{1, n, i} | X_i]]} \leq t | X^{(n)}, D^{(n)}\} - \Phi(t)| \stackrel{P}{\to} 0~. \]
Similar arguments lead to
\[ \sup_{t \in \mathbf R} |P \{B_n / \sqrt{E[\var[\phi_{0, n, i} | X_i]]} \leq t | X^{(n)}, D^{(n)}\} - \Phi(t)| \stackrel{P}{\to} 0~. \]

\noindent \underline{Step 3: Combining conditional and unconditional components}

Meanwhile, it follows from the same arguments as those in (S.22)--(S.25) of \cite{bai2021inference} that
\[ C_n - D_n \stackrel{d}{\to} N \left (0, \frac{1}{2} E \left [ (E[Y_i(1) | X_i] - E[Y_i(1)] - (E[Y_i(0) | X_i] - E[Y_i(0)]))^2 \right ] \right )~. \]
To establish \eqref{eq:normal}, define $\nu_n^2 = \nu_{1, n}^2 + \nu_{0, n}^2 + \nu_2^2$, where
\begin{align*}
\nu_{1, n}^2 & = E[\var[\phi_{1, n, i} | X_i]] \\
\nu_{0, n}^2 & = E[\var[\phi_{0, n, i} | X_i]] \\
\nu^2 & = \frac{1}{2} E \left [ (E[Y_i(1) | X_i] - E[Y_i(1)] - (E[Y_i(0) | X_i] - E[Y_i(0)]))^2 \right ]
\end{align*}
Note
\[ \frac{\sqrt n (\hat \Delta_n - \Delta(Q))}{\nu_n} = \frac{A_n}{\nu_{1, n}} \frac{\nu_{1, n}}{\nu_n} - \frac{B_n}{\nu_{0, n}} \frac{\nu_{0, n}}{\nu_n} + \frac{C_n - D_n}{\nu_2} \frac{\nu_2}{\nu_n}~. \]
Further note $\nu_n, \nu_{1, n}, \nu_{0, n}, \nu_2$ are all constants conditional on $X^{(n)}$ and $D^{(n)}$. Suppose by contradiction that $\frac{\sqrt n (\hat \Delta_n - \Delta(Q))}{\nu_n}$ does not converge in distribution to $N(0, 1)$. Then, there exists $\epsilon > 0$ and a subsequence $\{n_k\}$ such that
\begin{equation} \label{eq:contra}
\sup_{t \in \mathbf R} |P \{\sqrt n_k (\hat \Delta_{n_k} - \Delta(Q)) / \nu_{n_k} \leq t\} - \Phi(t)| \to \epsilon~.
\end{equation}
Because the sequence $\nu_{1, n_k}$ and $\nu_{0, n_k}$ are bounded by Assumptions \ref{ass:md}(b), there is a further subsequence, which with some abuse of notation we still denote by $\{n_k\}$, along which $\nu_{1, n_k} \to \nu_1^\ast$ and $\nu_{0, n_k} \to \nu_0^\ast$ for some $\nu_1^\ast, \nu_0^\ast \geq 0$. Then, $\nu_{1, n_k} / \nu_{n_k}, \nu_{0, n_k} / \nu_{n_k}, \nu_2 / \nu_{n_k}$ all converge to constants. Therefore, it follows from Lemma S.1.2 of \cite{bai2021inference} that
\[ \sqrt n_k (\hat \Delta_{n_k} - \Delta(Q)) / \nu_{n_k} \stackrel{d}{\to} N(0, 1)~, \]
a contradiction to \eqref{eq:contra}. Therefore, the desired convergence in Theorem \ref{thm:main} follows.



\noindent \underline{Step 4: Rearranging the variance formula}

To conclude the proof with the the variance formula as stated in the theorem, note
\begin{align}
\nonumber & \var \Big [ Y_i(0) - \frac{1}{2} (m_{1, n}(X_i, W_i) + m_{0, n}(X_i, W_i)) \Big | X_i \Big ] \\
\nonumber & = \var \Big [ E \Big [ Y_i(0) - \frac{1}{2} (m_{1, n}(X_i, W_i) + m_{0, n}(X_i, W_i)) \Big | X_i, W_i \Big ] \Big | X_i \Big ] \\
\nonumber & \hspace{3em} + E \Big [ \var \Big [ Y_i(0) - \frac{1}{2} (m_{1, n}(X_i, W_i) + m_{0, n}(X_i, W_i)) \Big | X_i, W_i \Big ] \Big | X_i \Big ] \\
\nonumber & = \var \Big [ E \Big [ \frac{Y_i(1) + Y_i(0)}{2}  \Big | X_i, W_i \Big ] - \frac{1}{2} (m_{1, n}(X_i, W_i) + m_{0, n}(X_i, W_i)) - E \Big [ \frac{Y_i(1) - Y_i(0)}{2}  \Big | X_i, W_i \Big ] \Big | X_i \Big ] \\
\nonumber & \hspace{3em}  + E [\var[Y_i(0) | X_i, W_i] | X_i] \\
\nonumber & = \var \Big [ E \Big [ \frac{Y_i(1) + Y_i(0)}{2}  \Big | X_i, W_i \Big ] - \frac{1}{2} (m_{1, n}(X_i, W_i) + m_{0, n}(X_i, W_i)) \Big | X_i \Big ] \\
\nonumber & \hspace{3em} + \var \Big [E \Big [ \frac{Y_i(1) - Y_i(0)}{2}  \Big | X_i, W_i \Big ] \Big | X_i \Big ] \\
\nonumber & \hspace{3em} - 2 \mathrm{Cov} \Big [ E \Big [ \frac{Y_i(1) + Y_i(0)}{2}  \Big | X_i, W_i \Big ] - \frac{1}{2} (m_{1, n}(X_i, W_i) + m_{0, n}(X_i, W_i)), E \Big [ \frac{Y_i(1) - Y_i(0)}{2}  \Big | X_i, W_i \Big ] \Big | X_i \Big ] \\
\label{eq:var0} & \hspace{3em}  + E [\var[Y_i(0) | X_i, W_i] | X_i]~,
\end{align}
where the first equality follows from the law of total variance, the second one follows by direct calculation, and the last one follows by expanding the variance of the sum. Similarly,
\begin{align}
\nonumber & \var \Big [ Y_i(1) - \frac{1}{2} (m_{1, n}(X_i, W_i) + m_{0, n}(X_i, W_i)) \Big | X_i \Big ] \\
\nonumber & = \var \Big [ E \Big [ \frac{Y_i(1) + Y_i(0)}{2}  \Big | X_i, W_i \Big ] - \frac{1}{2} (m_{1, n}(X_i, W_i) + m_{0, n}(X_i, W_i)) \Big | X_i \Big ] \\
\nonumber & \hspace{3em} + \var \Big [E \Big [ \frac{Y_i(1) - Y_i(0)}{2}  \Big | X_i, W_i \Big ] \Big | X_i \Big ] \\
\nonumber & \hspace{3em} + 2 \mathrm{Cov} \Big [ E \Big [ \frac{Y_i(1) + Y_i(0)}{2}  \Big | X_i, W_i \Big ] - \frac{1}{2} (m_{1, n}(X_i, W_i) + m_{0, n}(X_i, W_i)), E \Big [ \frac{Y_i(1) - Y_i(0)}{2}  \Big | X_i, W_i \Big ] \Big | X_i \Big ] \\
\label{eq:var1} & \hspace{3em}  + E [\var[Y_i(1) | X_i, W_i] | X_i]~.
\end{align}
It follows that
\begin{align}
\nonumber \sigma_n^2(Q) & = \frac{1}{2} E[\var[E[Y_i(1) + Y_i(0) | X_i, W_i] - (m_{1, n}(X_i, W_i) + m_{0, n}(X_i, W_i)) | X_i]] \\
\nonumber & \hspace{3em} + \frac{1}{2} E[\var[E[Y_i(1) - Y_i(0) | X_i, W_i] | X_i]] + \frac{1}{2} \var[E[Y_i(1) - Y_i(0) | X_i]] \\
\nonumber & \hspace{3em} + E [\var[Y_i(0) | X_i, W_i] | X_i] + E [\var[Y_i(1) | X_i, W_i] | X_i] \\
\nonumber & = \frac{1}{2} E[\var[E[Y_i(1) + Y_i(0) | X_i, W_i] - (m_{1, n}(X_i, W_i) + m_{0, n}(X_i, W_i)) | X_i]] \\
\nonumber & \hspace{3em} + \frac{1}{2} E[(E[Y_i(1) - Y_i(0) | X_i, W_i] - E[Y_i(1) - Y_i(0) | X_i])^2] \\
\nonumber & \hspace{3em} + \frac{1}{2} E[(E[Y_i(1) - Y_i(0) | X_i] - E[Y_i(1) - Y_i(0)])^2] \\
\nonumber & \hspace{3em} + E[(Y_i(0) - E[Y_i(0) | X_i, W_i])^2] + E[(Y_i(1) - E[Y_i(1) | X_i, W_i])^2] \\
\nonumber & = \frac{1}{2} E[\var[E[Y_i(1) + Y_i(0) | X_i, W_i] - (m_{1, n}(X_i, W_i) + m_{0, n}(X_i, W_i)) | X_i]] \\
\nonumber & \hspace{3em} + \frac{1}{2} \var[E[Y_i(1) - Y_i(0) | X_i, W_i]] + E [\var[Y_i(0) | X_i, W_i]] + E [\var[Y_i(1) | X_i, W_i]]~,
\end{align}
where the first equality follows by definition, the second one follows from \eqref{eq:var0}--\eqref{eq:var1}, the third one again follows by definition, and the last one follows because by the law of iterated expectations,
\[ E[(E[Y_i(1) - Y_i(0) | X_i, W_i] - E[Y_i(1) - Y_i(0) | X_i]) (E[Y_i(1) - Y_i(0) | X_i] - E[Y_i(1) - Y_i(0)])] = 0~. \]
The conclusion then follows. \qed

\subsection{Proof of Theorem \ref{thm:var}}
Theorem \ref{thm:main} implies $\hat \Delta_n \stackrel{P}{\to} \Delta(Q)$. Next, we show
\begin{equation} \label{eq:tau}
\hat \tau_n^2 - E[\var[\phi_{1, n, i} | X_i]] + E[\var[\phi_{0, n, i} | X_i]] + E[(E[Y_i(1) | X_i] - E[Y_i(0) | X_i])^2] \stackrel{P}{\to} 0~.
\end{equation}
To that end, define
\[ \mathring Y_i = Y_i - \frac{1}{2} (m_{1, n}(X_i, W_i) + m_{0, n}(X_i, W_i))~. \]
Note
\begin{align*}
\hat \tau_n^2 & = \frac{1}{n} \sum_{1 \leq j \leq n} \left ( \mathring Y_{\pi(2j - 1)} - \mathring Y_{\pi(2j)} + (\tilde Y_{\pi(2j - 1)} - \tilde Y_{\pi(2j)} - (\mathring Y_{\pi(2j - 1)} - \mathring Y_{\pi(2j)}) ) \right )^2 \\
& = \frac{1}{n} \sum_{1 \leq j \leq n} (\mathring Y_{\pi(2j - 1)} - \mathring Y_{\pi(2j)})^2 + \frac{1}{n} \sum_{1 \leq j \leq n} (\tilde Y_{\pi(2j - 1)} - \tilde Y_{\pi(2j)} - (\mathring Y_{\pi(2j - 1)} - \mathring Y_{\pi(2j)}) )^2 \\
& \hspace{3em} + \frac{2}{n} \sum_{1 \leq j \leq n} (\tilde Y_{\pi(2j - 1)} - \tilde Y_{\pi(2j)} - (\mathring Y_{\pi(2j - 1)} - \mathring Y_{\pi(2j)}) ) (\mathring Y_{\pi(2j - 1)} - \mathring Y_{\pi(2j)})~.
\end{align*}
Therefore, to establish \eqref{eq:tau}, we first show
\begin{equation} \label{eq:check}
\frac{1}{n} \sum_{1 \leq j \leq n} (\mathring Y_{\pi(2j - 1)} - \mathring Y_{\pi(2j)})^2 - E[\var[\phi_{1, n, i} | X_i]] + E[\var[\phi_{0, n, i} | X_i]] + E[(E[Y_i(1) | X_i] - E[Y_i(0) | X_i])^2] \stackrel{P}{\to} 0
\end{equation}
and
\begin{equation} \label{eq:checktilde}
\frac{1}{n} \sum_{1 \leq j \leq n} (\tilde Y_{\pi(2j - 1)} - \tilde Y_{\pi(2j)} - (\mathring Y_{\pi(2j - 1)} - \mathring Y_{\pi(2j)}) )^2 \stackrel{P}{\to} 0~.
\end{equation}
\eqref{eq:checktilde} immediately follows from repeated applications of the inequality $(a - b)^2 \leq 2 (a^2 + b^2)$ and \eqref{eq:L2}. To verify \eqref{eq:check}, note
\[ \frac{1}{n} \sum_{1 \leq j \leq n} (\mathring Y_{\pi(2j - 1)} - \mathring Y_{\pi(2j)})^2 = \frac{1}{n} \sum_{1 \leq i \leq 2n} \mathring Y_i^2 - \frac{2}{n} \sum_{1 \leq j \leq n} \mathring Y_{\pi(2j - 1)} \mathring Y_{\pi(2j)}~. \]
It follows from similar arguments to those in the proof of Lemma \ref{lem:sn} below that
\[ \frac{1}{n} \sum_{1 \leq i \leq 2n} \mathring Y_i^2 - E[\phi_{1, n, i}^2] + E[\phi_{0, n, i}^2] \stackrel{P}{\to} 0~. \]
Similarly, it follows from the proof of the same lemma that
\[ \frac{2}{n} \sum_{1 \leq j \leq n} \mathring Y_{\pi(2j - 1)} \mathring Y_{\pi(2j)} - 2 E[E[\phi_{1, n, i} | X_i] E[\phi_{0, n, i} | X_i]] \stackrel{P}{\to} 0~. \]
To establish \eqref{eq:check}, note
\begin{align*}
& E[\phi_{1, n, i}^2] + E[\phi_{0, n, i}^2] - 2 E[E[\phi_{1, n, i} | X_i] E[\phi_{0, n, i} | X_i]] \\
& = E[\var[\phi_{1, n, i} | X_i]] + E[\var[\phi_{0, n, i} | X_i]] + E[E[\phi_{1, n, i} | X_i]^2] + E[E[\phi_{0, n, i} | X_i]^2] - 2 E[E[\phi_{1, n, i} | X_i] E[\phi_{0, n, i} | X_i]] \\
& = E[\var[\phi_{1, n, i} | X_i]] + E[\var[\phi_{0, n, i} | X_i]] + E[(E[\phi_{1, n, i} | X_i] - E[\phi_{0, n, i} | X_i])^2] \\
& = E[\var[\phi_{1, n, i} | X_i]] + E[\var[\phi_{0, n, i} | X_i]] + E[(E[Y_i(1) | X_i] - E[Y_i(0) | X_i])^2]~,
\end{align*}
where the last equality follows from the definition of $\phi_{1, n, i}$ and $\phi_{0, n, i}$. It then follows from the Cauchy-Schwarz inequality that
\begin{multline*}
\left | \frac{1}{n} \sum_{1 \leq j \leq n} (\tilde Y_{\pi(2j - 1)} - \tilde Y_{\pi(2j)} - (\mathring Y_{\pi(2j - 1)} - \mathring Y_{\pi(2j)}) ) (\mathring Y_{\pi(2j - 1)} - \mathring Y_{\pi(2j)}) \right | \\
\leq \left ( \frac{1}{n} \sum_{1 \leq j \leq n} (\mathring Y_{\pi(2j - 1)} - \mathring Y_{\pi(2j)})^2 \right ) \left ( \frac{1}{n} \sum_{1 \leq j \leq n} (\tilde Y_{\pi(2j - 1)} - \tilde Y_{\pi(2j)} - (\mathring Y_{\pi(2j - 1)} - \mathring Y_{\pi(2j)}) )^2 \right ) \stackrel{P}{\to} 0~,
\end{multline*}
which, together with \eqref{eq:check}--\eqref{eq:checktilde} as well as Assumptions \ref{ass:Q}(b) and \ref{ass:md}(b), imply \eqref{eq:tau}.

Next, we show
\begin{equation} \label{eq:lambda}
\hat \lambda_n \stackrel{P}{\to} E[(E[Y_i(1) | X_i] - E[Y_i(0) | X_i])^2]~.
\end{equation}
Note
\begin{align}
\label{eq:lambdadecomp} & \hat \lambda_n - \frac{2}{n} \sum_{1 \leq j \leq \lfloor \frac{n}{2} \rfloor} (\mathring Y_{\pi(4j - 3)} - \mathring Y_{\pi(4j - 2)}) (\mathring Y_{\pi(4j - 1)} - \mathring Y_{\pi(4j)}) (D_{\pi(4j - 3)} - D_{\pi(4j - 2)}) (D_{\pi(4j - 1)} - D_{\pi(4j)}) \\
\nonumber & = \frac{2}{n} \sum_{1 \leq j \leq \lfloor \frac{n}{2} \rfloor} (\tilde Y_{\pi(4j - 3)} - \mathring Y_{\pi(4j - 3)} - (\tilde Y_{\pi(4j - 2)} - \mathring Y_{\pi(4j - 2)})) (\mathring Y_{\pi(4j - 1)} - \mathring Y_{\pi(4j)}) \\
\nonumber & \hspace{5em} \times (D_{\pi(4j - 3)} - D_{\pi(4j - 2)}) (D_{\pi(4j - 1)} - D_{\pi(4j)}) \\
\nonumber & \hspace{3em} + \frac{2}{n} \sum_{1 \leq j \leq \lfloor \frac{n}{2} \rfloor} (\mathring Y_{\pi(4j - 3)} - \mathring Y_{\pi(4j - 2)}) (\tilde Y_{\pi(4j - 1)} - \mathring Y_{\pi(4j - 1)}) - (\tilde Y_{\pi(4j)} - \mathring Y_{\pi(4j)})) \\
\nonumber & \hspace{5em} \times (D_{\pi(4j - 3)} - D_{\pi(4j - 2)}) (D_{\pi(4j - 1)} - D_{\pi(4j)}) \\
\nonumber & \hspace{3em} + \frac{2}{n} \sum_{1 \leq j \leq \lfloor \frac{n}{2} \rfloor} (\tilde Y_{\pi(4j - 3)} - \mathring Y_{\pi(4j - 3)} - (\tilde Y_{\pi(4j - 2)} - \mathring Y_{\pi(4j - 2)})) \\
\nonumber & \hspace{5em} \times (\tilde Y_{\pi(4j - 1)} - \mathring Y_{\pi(4j - 1)}) - (\tilde Y_{\pi(4j)} - \mathring Y_{\pi(4j)})) (D_{\pi(4j - 3)} - D_{\pi(4j - 2)}) (D_{\pi(4j - 1)} - D_{\pi(4j)})~.
\end{align}
In what follows, we show
\begin{align}
\label{eq:check1} & \frac{2}{n} \sum_{1 \leq j \leq \lfloor \frac{n}{2} \rfloor} (\mathring Y_{\pi(4j - 3)} - \mathring Y_{\pi(4j - 2)})^2 = O_P(1) \\
\label{eq:check2} & \frac{2}{n} \sum_{1 \leq j \leq \lfloor \frac{n}{2} \rfloor} (\mathring Y_{\pi(4j - 1)} - \mathring Y_{\pi(4j)})^2 = O_P(1) \\
\label{eq:checktilde1} & \frac{2}{n} \sum_{1 \leq j \leq \lfloor \frac{n}{2} \rfloor} (\tilde Y_{\pi(4j - 3)} - \mathring Y_{\pi(4j - 3)} - (\tilde Y_{\pi(4j - 2)} - \mathring Y_{\pi(4j - 2)}))^2 = o_P(1) \\
\label{eq:checktilde2} & \frac{2}{n} \sum_{1 \leq j \leq \lfloor \frac{n}{2} \rfloor} (\tilde Y_{\pi(4j - 1)} - \mathring Y_{\pi(4j - 1)}) - (\tilde Y_{\pi(4j)} - \mathring Y_{\pi(4j)}))^2 = o_P(1) \\
\nonumber & \frac{2}{n} \sum_{1 \leq j \leq \lfloor \frac{n}{2} \rfloor} (\mathring Y_{\pi(4j - 3)} - \mathring Y_{\pi(4j - 2)}) (\mathring Y_{\pi(4j - 1)} - \mathring Y_{\pi(4j)}) (D_{\pi(4j - 3)} - D_{\pi(4j - 2)}) (D_{\pi(4j - 1)} - D_{\pi(4j)}) \\
\label{eq:checkpop} & \hspace{5em} \stackrel{P}{\to} E[(E[Y_i(1) | X_i] - E[Y_i(0) | X_i])^2]~.
\end{align}
To establish \eqref{eq:check1}--\eqref{eq:check2}, note they follow directly from \eqref{eq:check} and Assumptions \ref{ass:Q}(b) and \ref{ass:md}(b). Next, note \eqref{eq:checktilde1} follows from repeated applications of the inequality $(a + b)^2 \leq 2 (a^2 + b^2)$ and \eqref{eq:L2}. \eqref{eq:checktilde2} can be established by similar arguments. \eqref{eq:checkpop} follows from similar arguments to those in the proof of Lemma S.1.7 of \cite{bai2021inference}, with the uniform integrability arguments replaced by arguments similar to those in the proof of Lemma \ref{lem:sn}, together with Assumptions \ref{ass:Q}--\ref{ass:pairsclose} and \ref{ass:md}. \eqref{eq:lambdadecomp}--\eqref{eq:checkpop} imply \eqref{eq:lambda} immediately.

Finally, note we have shown
\[ \hat \sigma_n^2 - \sigma_n^2 \stackrel{P}{\to} 0~. \]
Assumption \ref{ass:md}(a) implies $\sigma_n^2$ is bounded away from zero, so
\[ \frac{\hat \sigma_n}{\sigma_n} \stackrel{P}{\to} 1~. \]
The conclusion of the theorem then follows.
\qed

\subsection{Proof of Theorem \ref{thm:naive}}
We will apply the Frisch-Waugh-Lovell theorem to obtain an expression for $\hat \beta_n^{\rm naive}$. Consider the linear regression of $\psi_i$ on $1$ and $D_i$. Define
\[ \hat \mu_{\psi, n}(d) = \frac{1}{n} \sum_{1 \leq i \leq 2n} \psi_i I \{D_i = d\} \]
for $d \in \{0, 1\}$ and
\[  \hat \Delta_{\psi, n} = \hat \mu_{\psi, n}(1) - \hat \mu_{\psi, n}(0)~. \]
The $i$th residual based on the OLS estimation of this linear regression model is given by
\[ \tilde \psi_i = \psi_i - \hat \mu_{\psi, n}(0) - \hat \Delta_{\psi, n} D_i~. \]
$\hat \beta_n^{\rm naive}$ is then given by the OLS estimator of the coefficient in the linear regression of $Y_i$ on $\tilde \psi_i$. Note
\begin{align*}
\frac{1}{2n} \sum_{1 \leq i \leq 2n} \tilde \psi_i \tilde \psi_i' & = \frac{1}{2n} \sum_{1 \leq i \leq 2n} (\psi_i - \hat \mu_{\psi, n}(1)) (\psi_i - \hat \mu_{\psi, n}(1))' D_i + \frac{1}{2n} \sum_{1 \leq i \leq 2n} (\psi_i - \hat \mu_{\psi, n}(0)) (\psi_i - \hat \mu_{\psi, n}(0))' (1 - D_i) \\
& = \frac{1}{2n} \sum_{1 \leq i \leq 2n} \psi_i \psi_i' - \frac{1}{2} \hat \mu_{\psi, n}(1) \hat \mu_{\psi, n}(1)' - \frac{1}{2} \hat \mu_{\psi, n}(0) \hat \mu_{\psi, n}(0)'~.
\end{align*}
It follows from Assumption \ref{ass:psi}(b) and the weak law of large number that
\[ \frac{1}{2n} \sum_{1 \leq i \leq 2n} \psi_i \psi_i' \stackrel{P}{\to} E[\psi_i \psi_i']~. \]
On the other hand, it follows from Assumptions \ref{ass:treatment}--\ref{ass:close} and \ref{ass:psi}(b)--(c) as well as similar arguments to those in the proof of Lemma S.1.5 of \cite{bai2021inference} that
\[ \hat \mu_{\psi, n}(d) \stackrel{P}{\to} E[\psi_i] \]
for $d \in \{0, 1\}$. Therefore,
\[ \frac{1}{2n} \sum_{1 \leq i \leq 2n} \tilde \psi_i \tilde \psi_i' \stackrel{P}{\to} \var[\psi_i]~. \]

Next, 
\begin{align*}
\frac{1}{2n} \sum_{1 \leq i \leq 2n} \tilde \psi_i Y_i & = \frac{1}{2n} \sum_{1 \leq i \leq 2n} (\psi_i - \hat \mu_{\psi, n}(1)) Y_i(1) D_i + \frac{1}{2n} \sum_{1 \leq i \leq 2n} (\psi_i - \hat \mu_{\psi, n}(0)) Y_i(0) (1 - D_i)
\end{align*}
It follows from similar arguments as above as well as Assumptions \ref{ass:Q}(b), \ref{ass:treatment}--\ref{ass:close}, and \ref{ass:psi}(b)--(c) that
\[ \frac{1}{2n} \sum_{1 \leq i \leq 2n} \tilde \psi_i Y_i \stackrel{P}{\to} \cov[\psi_i, Y_i(1) + Y_i(0)]~. \]
The convergence of $\hat \beta_n^{\rm naive}$ therefore follows from the continuous mapping theorem and Assumption \ref{ass:psi}(a).

To see \eqref{eq:L2} is satisfied, note
\[ \frac{1}{2n} \sum_{1 \leq i \leq 2n} (\hat m_{d, n}(X_i,W_i) - m_{d, n}(X_i,W_i))^2 = (\hat \beta_n^{\rm naive} - \beta^{\rm naive})' \left ( \frac{1}{2n} \sum_{1 \leq i \leq 2n} \psi_i\psi_i' \right ) (\hat \beta_n^{\rm naive} - \beta^{\rm naive})~. \]
\eqref{eq:L2} then follows from the fact that $\hat \beta_n^{\rm naive} \stackrel{P}{\to} \beta$, Assumption \ref{ass:psi}(b), and the weak law of large numbers. To establish \eqref{eq:rate}, first note
\[ \frac{1}{\sqrt {2n}} \sum_{1 \leq i \leq 2n} (2D_i - 1)(\hat m_{d, n}(X_i,W_i) - m_{d, n}(X_i,W_i)) = \frac{1}{\sqrt {2}} \sqrt n \hat \Delta_{\psi, n}' (\hat \beta_n^{\rm naive} - \beta^{\rm naive})~. \]
In what follows, we establish
\begin{equation} \label{eq:rootnpsi}
\sqrt n \hat \Delta_{\psi, n} = O_P(1)~,
\end{equation}
from which \eqref{eq:rate} follows immediately because $\hat \beta_n^{\rm naive} - \beta^{\rm naive} = o_P(1)$. Note by Assumption \ref{ass:treatment} that $E[\sqrt n \hat \Delta_{\psi, n} | X^{(n)}] = 0$. Also note
\[ \sqrt n \hat \Delta_{\psi, n} = F_n - G_n + H_n~, \]
where
\begin{align*}
F_n & = \frac{1}{\sqrt n} \sum_{1 \leq i \leq 2n} (\psi_i - E[\psi_i | X_i]) D_i~, \\
G_n & = \frac{1}{\sqrt n} \sum_{1 \leq i \leq 2n} (\psi_i - E[\psi_i | X_i]) (1 - D_i)~,\quad \text{and} \\
H_n & = \frac{1}{\sqrt n} \sum_{1 \leq j \leq n} (E[\psi_{\pi(2j - 1)} | X_{\pi(2j - 1)}] - E[\psi_{\pi(2j)} | X_{\pi(2j)}]) (D_{\pi(2j - 1)} - D_{\pi(2j)})~.
\end{align*}
We will argue $F_n, G_n, H_n$ are all $O_P(1)$. Since this could be carried out separately for each entry of $F_n$ and $G_n$, we assume without loss of generality that $k_\psi = 1$. First, it follows from Assumptions \ref{ass:treatment}--\ref{ass:close} and \ref{ass:psi}(c) as well as similar arguments to those in the proof of Lemma S.1.4 of \cite{bai2021inference} that
\[ \mathrm{Var}[F_n | X^{(n)}, D^{(n)}] = \frac{1}{n}\sum_{1 \leq i \leq 2n} \mathrm{Var}[\psi_i | X_i] D_i \stackrel{P}{\to} E[\mathrm{Var}[\psi_i | X_i]] > 0~. \]
It then follows from similar arguments using the Lindeberg central limit theorem as in the proof of Lemma S.1.4 of \cite{bai2021inference} that $F_n = O_P(1)$. Similar arguments establish $G_n = O_P(1)$. Finally, we show $H_n = O_P(1)$. Note that $E[H_n | X^{(n)}] = 0$ and by Assumptions \ref{ass:treatment}--\ref{ass:close} and \ref{ass:psi}(c),
\[ \mathrm{Var}[H_n | X^{(n)}] = \frac{1}{n} \sum_{1 \leq j \leq n} (E[\psi_{\pi(2j - 1)} | X_{\pi(2j - 1)}] - E[\psi_{\pi(2j)} | X_{\pi(2j)}])^2 \stackrel{P}{\to} 0~. \]
Therefore, for any fixed $\epsilon > 0$, Markov's inequality implies
\[ P \{|H_n - E[H_n | X^{(n)}]| > \epsilon | X^{(n)}\} \leq \frac{\mathrm{Var}[H_n | X^{(n)}]}{\epsilon^2} \stackrel{P}{\to} 0~. \]
Since probabilities are bounded and therefore uniformly integrable, we have that
\[ P \{|H_n - E[H_n | X^{(n)}]| > \epsilon\} \to 0~. \]
Therefore, \eqref{eq:rootnpsi} follows. Finally, it is straightforward to see Assumption \ref{ass:md} is implied by Assumption \ref{ass:psi}. \qed

\subsection{Proof of Theorem \ref{thm:pfe}}
By the Frisch-Waugh-Lovell theorem, $\hat \beta_n^{\rm pfe}$ is equal to the OLS estimator in the linear regression of $\{(Y_{\pi(2j - 1)} - Y_{\pi(2j)}, Y_{\pi(2j)} - Y_{\pi(2j - 1)}): 1 \leq j \leq n\}$ on $\{(2 D_{\pi(2j - 1)} - 1, 2 D_{\pi(2j)} - 1): 1 \leq j \leq n\}$ and $\{(\psi_{\pi(2j - 1)} - \psi_{\pi(2j)}, \psi_{\pi(2j)} - \psi_{\pi(2j - 1)}): 1 \leq j \leq n\}$. To apply the Frisch-Waugh-Lovell theorem again, we study the linear regression of $\{(\psi_{\pi(2j - 1)} - \psi_{\pi(2j)}, \psi_{\pi(2j)} - \psi_{\pi(2j - 1)}): 1 \leq j \leq n\}$ on $\{(2 D_{\pi(2j - 1)} - 1, 2 D_{\pi(2j)} - 1): 1 \leq j \leq n\}$. The OLS estimator of the regression coefficient in such a regression equals
\begin{align*}
\hat \Delta_{\psi, n} = \frac{1}{n} \sum_{1 \leq j \leq n} (D_{\pi(2j - 1)} - D_{\pi(2j)}) (\psi_{\pi(2j - 1)} - \psi_{\pi(2j)})~.
\end{align*}
The residual is therefore
$\{(\psi_{\pi(2j - 1)} - \psi_{\pi(2j)} - (2 D_{\pi(2j - 1)} - 1) \hat \Delta_{\psi, n}, \psi_{\pi(2j)} - \psi_{\pi(2j - 1)} - (2 D_{\pi(2j)} - 1) \hat \Delta_{\psi, n}): 1 \leq j \leq n\}$. $\hat \beta_n^{\rm pfe}$ equals the OLS estimator of the coefficient in the linear regression of $\{(Y_{\pi(2j - 1)} - Y_{\pi(2j)}, Y_{\pi(2j)} - Y_{\pi(2j - 1)}): 1 \leq j \leq n\}$ on those residuals. Define
\begin{align*}
\delta_{Y, j} & = (D_{\pi(2j - 1)} - D_{\pi(2j)}) (Y_{\pi(2j - 1)} - Y_{\pi(2j)}) \quad \text{and}\\
\delta_{\psi, j} & = (D_{\pi(2j - 1)} - D_{\pi(2j)}) (\psi_{\pi(2j - 1)} - \psi_{\pi(2j)})
\end{align*}
Apparently $\hat \Delta_{\psi, n} = \frac{1}{n} \sum_{1 \leq j \leq n} \delta_{\psi, j}$. A moment's thought reveals that $\hat \beta_n^{\rm pfe}$ further equals the coefficient estimate using least squares in the linear regression of $\delta_{Y, j}$ on $\delta_{\psi, j} - \hat \Delta_{\psi, n}$ for $1 \leq j \leq n$. It follows from Assumptions \ref{ass:Q}(b)--(c), \ref{ass:treatment}--\ref{ass:close}, and \ref{ass:psi}(b)--(c) as well as similar arguments to those in the proof of Lemma S.1.5 of \cite{bai2021inference} that
\begin{align}
\label{eq:psi0} \hat \Delta_{\psi, n} & \stackrel{P}{\to} 0 \quad \text{and}\\
\nonumber \frac{1}{n} \sum_{1 \leq j \leq n} \delta_{Y, j} & \stackrel{P}{\to} \Delta(Q)~.
\end{align}
Next, note that
\begin{align}
\nonumber & \frac{1}{n} \sum_{1 \leq j \leq n} \delta_{\psi, j} \delta_{\psi, j}' \\
\nonumber & = \frac{1}{n} \sum_{1 \leq j \leq n} (\psi_{\pi(2j - 1)} - \psi_{\pi(2j)}) (\psi_{\pi(2j - 1)} - \psi_{\pi(2j)})' \\
\label{eq:psi} & = \frac{1}{n} \sum_{1 \leq i \leq 2n} \psi_i \psi_i' - \frac{1}{n} \sum_{1 \leq j \leq n} (\psi_{\pi(2j - 1)} \psi_{\pi(2j)}' + \psi_{\pi(2j)} \psi_{\pi(2j - 1)}')~.
\end{align}
For convenience, we introduce the following notation:
\begin{align*}
\mu_d(X_i) & = E[Y_i(d) | X_i] \\
\Psi(X_i) & = E[\psi_i | X_i] \\
\xi_d(X_i) & = E[\psi_i Y_i(d) | X_i]~.
\end{align*}
The first term in \eqref{eq:psi} converges in probability to $2 E[\psi_i \psi_i']$ by the weak law of large numbers. For the second term, we have that
\begin{align*}
& E \Big [ \frac{1}{n} \sum_{1 \leq j \leq n} (\psi_{\pi(2j - 1)} \psi_{\pi(2j)}' + \psi_{\pi(2j)} \psi_{\pi(2j - 1)}') \Big | X^{(n)} \Big ] \\
& = \frac{1}{n} \sum_{1 \leq i \leq 2n} \Psi(X_i) \Psi(X_i)' - \frac{1}{n} \sum_{1 \leq j \leq n} (\Psi(X_{\pi(2j - 1)}) - \Psi(X_{\pi(2j)})) (\Psi(X_{\pi(2j - 1)}) - \Psi(X_{\pi(2j)}))' \\
& \stackrel{P}{\to} 2 E[\Psi(X_i) \Psi(X_i)']~,
\end{align*}
where the convergence in probability holds because of Assumptions \ref{ass:treatment}--\ref{ass:close} and \ref{ass:psi}(c). It follows from Assumptions \ref{ass:treatment}--\ref{ass:close} and \ref{ass:psi}(b)--(c) as well as similar arguments to those in the proof of Lemma S.1.6 of \cite{bai2021inference} that
\[ \Big |\frac{1}{n} \sum_{1 \leq j \leq n} (\psi_{\pi(2j - 1)} \psi_{\pi(2j)}' + \psi_{\pi(2j)} \psi_{\pi(2j - 1)}') - E \Big [ \frac{1}{n} \sum_{1 \leq j \leq n} (\psi_{\pi(2j - 1)} \psi_{\pi(2j)}' + \psi_{\pi(2j)} \psi_{\pi(2j - 1)}') \Big | X^{(n)} \Big ] \Big | \stackrel{P}{\to} 0~. \]
Therefore,
\[ \frac{1}{n} \sum_{1 \leq j \leq n} \delta_{\psi, j} \delta_{\psi, j}' \stackrel{P}{\to} 2 E[\var[\psi_i | X_i]]~. \]
We now turn to
\[ \frac{1}{n} \sum_{1 \leq j \leq n} \delta_{\psi, j} \delta_{Y, j} = \frac{1}{n} \sum_{1 \leq j \leq n} (\psi_{\pi(2j - 1)} - \psi_{\pi(2j)}) (Y_{\pi(2j - 1)} - Y_{\pi(2j)})~. \]
Note that
\begin{align*}
E[\psi_{\pi(2j - 1)} Y_{\pi(2j - 1)} | X^{(n)}] & = \frac{1}{2} \xi_1(X_{\pi(2j - 1)}) + \frac{1}{2} \xi_0(X_{\pi(2j - 1)}) \\
E[\psi_{\pi(2j - 1)} Y_{\pi(2j)} | X^{(n)}] & = \frac{1}{2} \Psi(X_{\pi(2j - 1)}) (\mu_1(X_{\pi(2j)}) + \mu_0(X_{\pi(2j)}))~.
\end{align*}
It follows from Assumptions \ref{ass:Q}(b)--(c), \ref{ass:treatment}--\ref{ass:close}, \ref{ass:psi}(b)--(c) as well as similar arguments to those in the proof of Lemma S.1.6 of \cite{bai2021inference} that
\[ \frac{1}{n} \sum_{1 \leq j \leq n} \delta_{\psi, j} \delta_{Y, j} \stackrel{P}{\to} E[\psi_i (Y_i(1) + Y_i(0))] - E[\Psi(X_i) (\mu_1(X_i) + \mu_0(X_i))]~. \]
The convergence in probability of $\hat \beta_n^{\rm pfe}$ now follows from Assumption \ref{ass:psi}(a) and the continuous mapping theorem. \eqref{eq:rate}--\eqref{eq:L2} can be established using similar arguments to those in the proof of Theorem \ref{thm:naive}. Finally, it is straightforward to see Assumption \ref{ass:md} is implied by Assumption \ref{ass:psi}. \qed

\subsection{Proof of Theorem \ref{thm:LASSO2}}
We divide the proof into three steps. In the first step, we show 
\begin{equation} \label{eq:beta-hd2}
|\hat \alpha_{d,n}^{\rm r} - \alpha_{d,n}^{\rm r}|+\|\hat \beta_{d, n}^{\rm r} - \beta_{d, n}^{\rm r}\|_1 = O_P \left ( s_n  \lambda_n^{\rm r} \right )~.
\end{equation}
In the second step, we show \eqref{eq:rate},  \eqref{eq:L2}, and Assumption \ref{ass:md} hold. In the third step, we show the asymptotic variance achieves the minimum under the approximately correct specification condition in Theorem \ref{thm:LASSO2}. 

\underline{Step 1: Proof of \eqref{eq:beta-hd2}} 

Note that 
\begin{align*}
& \frac{1}{n} \sum_{1 \leq i \leq 2n} I \{D_i = d\} (Y_i(d) - \hat \alpha_{d, n}^{\rm r} - \psi_{n, i}' \hat \beta_{d, n}^{\rm r})^2 + \lambda_{d, n}^{\rm r} \|\hat{\Omega}_n(d) \hat \beta_{d, n}^{\rm r}\|_1 \\
&\leq \frac{1}{n} \sum_{1 \leq i \leq 2n} I \{D_i = d\} (Y_i(d) - \alpha_{d, n}^{\rm r}-\psi_{n, i}' \beta_{d, n}^{\rm r})^2 + \lambda_{d, n}^{\rm r} \|\hat{\Omega}_n(d) \beta_{d, n}^{\rm r}\|_1~.     
\end{align*}

Rearranging the terms, we then have
\begin{align}\label{eq:basic2}
& \frac{1}{n} \sum_{1 \leq i \leq 2n} I \{D_i = d\} (\hat \alpha_{d, n}^{\rm r} - \alpha_{d, n}^{\rm r} +\psi_{n, i}' (\hat \beta_{d, n}^{\rm r} - \beta_{d, n}^{\rm r}))^2 + \lambda_{d, n}^{\rm r} \|\hat{\Omega}_n(d) \hat \beta_{d, n}^{\rm r}\|_1 \notag \\
& \leq \left ( \frac{2}{n} \sum_{1 \leq i \leq 2n} I \{D_i = d\} \epsilon_{n, i}(d) \psi_{n, i}' \right ) (\hat \beta_{d, n}^{\rm r} - \beta_{d, n}^{\rm r}) + \left ( \frac{2}{n} \sum_{1 \leq i \leq 2n} I \{D_i = d\} \epsilon_{n, i}(d) \right ) (\hat \alpha_{d, n}^{\rm r} - \alpha_{d, n}^{\rm r}) \notag \\
& + \lambda_{d, n}^{\rm r} \|\hat{\Omega}_n(d) \beta_{d, n}^{\rm r}\|_1  \end{align}

Next, define
\[ \mathbb U_n(d) = \Omega_n^{-1}(d)\frac{1}{n } \sum_{1 \leq i \leq 2n} I \{D_i = d\} (\psi_{n, i} \epsilon_{n, i}(d) - E[\psi_{n, i} \epsilon_{n, i}(d)]) \]
and
\[ \mathcal E_n(d) = \left \{ \|\mathbb U_n(d)\|_\infty \leq \frac{6 \bar{\sigma}}{\ubar\sigma} \sqrt{\frac{\log (2n p_n)}{n}}, ~\left\vert \frac{1}{n}\sum_{i \in [2n]}I\{D_i=d\}\epsilon_{n,i}(d) - E[\epsilon_{n,i}(d)] \right\vert\leq \sqrt{\frac{\log (2n p_n)}{n}} \right \}~. \]
Lemma \ref{lem:En'} implies $P \{\mathcal E_n(d)\} \to 1$ for $d \in \{0, 1\}$.

On the event $\mathcal E_n(d)$, we have
\begin{align*}
& \left | \left ( \frac{2}{n} \sum_{1 \leq i \leq 2n} I \{D_i = d\} \epsilon_{n, i}(d) \psi_{n, i}' \right ) (\hat \beta_{d, n}^{\rm r} - \beta_{d, n}^{\rm r}) \right | \\
& \leq \left \|\Omega_n^{-1}(d)  \frac{2}{n} \sum_{1 \leq i \leq 2n} I \{D_i = d\} \epsilon_{n, i}(d) \psi_{n, i} \right \|_\infty \|\Omega_n(d)  (\hat \beta_{d, n}^{\rm r} - \beta_{d, n}^{\rm r})\|_1 \\
& \leq 2 \|\mathbb U_n(d)\|_\infty \|\Omega_n(d)(\hat \beta_{d, n}^{\rm r} - \beta_{d, n}^{\rm r})\|_1 + \left \| \Omega_n^{-1}(d)\frac{2}{n} \sum_{1 \leq i \leq 2n} I \{D_i = d\} E[\epsilon_{n, i}(d) \psi_{n, i}] \right \|_\infty \|\Omega_n(d)(\hat \beta_{d, n}^{\rm r} - \beta_{d, n}^{\rm r})\|_1 \\
& \leq 2 \|\mathbb U_n(d)\|_\infty \|\Omega_n(d)(\hat \beta_{d, n}^{\rm r} - \beta_{d, n}^{\rm r})\|_1 + \left \| \Omega_n^{-1}(d)2E[\epsilon_{n, i}(d) \psi_{n, i}] \right \|_\infty \|\Omega_n(d)(\hat \beta_{d, n}^{\rm r} - \beta_{d, n}^{\rm r})\|_1 \\
& \leq \left(\frac{12 \bar{\sigma}}{\ubar\sigma \ell \ell_n} + d_n\right)\lambda_{d, n}^{\rm r} \|\Omega_n(d)(\hat \beta_{d, n}^{\rm r} - \beta_{d, n}^{\rm r})\|_1~,
\end{align*}
where $d_n = o(1)$ and the last inequality follows from \eqref{eq:foc2} and the fact that
\[ \lambda_{d, n}^{\rm r} \geq \ell \ell_n  \sqrt{\frac{\log (2 n p_n)}{n}}~. \]
Next, define
\[ \hat \delta_{d, n} = \hat \beta_{d, n}^{\rm r} - \beta_{d, n}^{\rm r} \]
and let $S_{d, n}$ be the support of $\beta_{d, n}^{\rm r}$. Then, we have
\begin{align*}
\|\hat \Omega_n(d) \hat \beta_{d, n}^{\rm r}\|_1 & = \|(\hat \Omega_n(d) \hat \beta_{d, n}^{\rm r})_{S_{d, n}}\|_1 + \|(\hat \Omega_n(d)\hat \beta_{d, n}^{\rm r})_{S_{d, n}^c}\|_1 = \|(\hat \Omega_n(d)\hat \beta_{d, n}^{\rm r})_{S_{d, n}}\|_1 + \|(\hat \Omega_n(d) \hat \delta_{d, n})_{S_{d, n}^c}\|_1~, \\
\|\hat \Omega_n(d)\beta_{d, n}^{\rm r}\|_1 & = \|(\hat \Omega_n(d)\beta_{d, n}^{\rm r})_{S_{d, n}}\|_1 \leq \|(\hat \Omega_n(d)\hat \beta_{d, n}^{\rm r})_{S_{d, n}}\|_1 + \|(\hat \Omega_n(d) \hat \delta_{d, n}^{\rm r})_{S_{d, n}}\|_1~,
\end{align*}
and thus, 
\begin{align*}
& \left(\frac{12 \bar{\sigma}}{\ubar\sigma \ell \ell_n} + d_n\right) \|\Omega_n(d)\hat \delta_{d, n}\|_1 +\|\hat \Omega_n(d)\beta_{d, n}^{\rm r}\|_1 - \|\hat \Omega_n(d)\hat \beta_{d, n}^{\rm r}\|_1 \\
& = \left(\frac{12 \bar{\sigma}}{\ubar\sigma \ell \ell_n} + d_n\right) \|( \Omega_n(d) \hat \delta_{d, n})_{S_{d, n}}\|_1 + \left(\frac{12 \bar{\sigma}}{\ubar\sigma \ell \ell_n} + d_n\right) \|( \Omega_n(d) \hat \delta_{d, n})_{S_{d, n}^c}\|_1+ \|\hat \Omega_n(d)\beta_{d, n}^{\rm r}\|_1 - \|\hat \Omega_n(d)\hat \beta_{d, n}^{\rm r}\|_1 \\
& \leq \left(\frac{12 \bar{\sigma}}{\ubar\sigma \ell \ell_n} + d_n\right) \|( \Omega_n(d) \hat \delta_{d, n})_{S_{d, n}}\|_1 + \left(\frac{12 \bar{\sigma}}{\ubar\sigma \ell \ell_n} + d_n\right) \|( \Omega_n(d) \hat \delta_{d, n})_{S_{d, n}^c}\|_1+ \|(\hat \Omega_n(d) \hat \delta_{d, n}^{\rm r})_{S_{d, n}}\|_1 - \|(\hat \Omega_n(d) \hat \delta_{d, n}^{\rm r})_{S_{d, n}^c}\|_1~.
\end{align*}

Further define $\breve \delta_{d,n} = (\hat \alpha_{d, n}^{\rm r} - \alpha_{d, n}^{\rm r}, \hat \delta_{d,n}' )'$ and $\breve S_{d,n} = \{1,S_{d,n}+1\}$\footnote{For example, if $S_{d,n} = \{2,4,9\}$, we have $\breve S_{d,n} = \{1,3,5,10\}$.} and recall $\breve \psi_{n,i} = (1,\psi_{n,i}')'$. Then, 
together with \eqref{eq:basic2}, we have
\begin{align}\label{eq:basic3}
0 & \leq \frac{1}{n} \sum_{1 \leq i \leq 2n} I \{D_i = d\} (\breve \psi_{n, i}' \breve \delta_{d, n})^2 \notag \\
& \leq \lambda_{d, n}^{\rm r} \left[ \left(\frac{12 \bar{\sigma}}{\ubar\sigma \ell \ell_n} + d_n + \bar{c}\right) \|( \Omega_n(d) \hat \delta_{d, n})_{S_{d, n}}\|_1 - \left(\underline{c} - \frac{12 \bar{\sigma}}{\ubar\sigma \ell \ell_n} - d_n\right) \|( \Omega_n(d) \hat \delta_{d, n})_{S_{d, n}^c}\|_1\right] \notag \\
& + \lambda_{d, n}^{\rm r}(1/\ell_n + d_n)|\hat \alpha_{d, n}^{\rm r} - \alpha_{d, n}^{\rm r}| \notag \\
& \leq \lambda_{d, n}^{\rm r} \left[ \left(\frac{12 \bar{\sigma}}{\ubar\sigma \ell \ell_n} + d_n + \bar{c}\right)\bar{\sigma} \|(\hat \delta_{d, n})_{S_{d, n}}\|_1 - \left(\underline{c} - \frac{12 \bar{\sigma}}{\ubar\sigma \ell \ell_n} - d_n\right) \ubar\sigma \|( \hat \delta_{d, n})_{S_{d, n}^c}\|_1\right] \notag \\
& + \lambda_{d, n}^{\rm r}(1/\ell_n + d_n)|\hat \alpha_{d, n}^{\rm r} - \alpha_{d, n}^{\rm r}| \notag \\
&\leq \lambda_{d, n}^{\rm r} \left[ \left(\frac{12 \bar{\sigma}}{\ubar\sigma \ell \ell_n} + d_n + \bar{c}\right)\bar{\sigma} \|(\breve \delta_{d, n})_{\breve S_{d, n}}\|_1 - \left(\underline{c} - \frac{12 \bar{\sigma}}{\ubar\sigma \ell \ell_n} - d_n\right) \ubar\sigma \|( \breve \delta_{d, n})_{\breve S_{d, n}^c}\|_1\right]~. 
\end{align}
Define
\[ \mathcal C_n = \left\{u \in \mathbf R^{p_n+1}: \|u_{\breve S_{d, n}^c}\|_1 \leq \frac{2 \bar{\sigma} \bar{c}}{\ubar\sigma \underline{c}} \|u_{\breve S_{d, n}}\|_1\right\}~. \]
For sufficiently large $n$, we have $\breve \delta_{d,n} \in \mathcal C_n$. It follows from \cite{bickel2009simultaneous} and Assumption \ref{ass:LASSO2-ev} that
\[ \inf_{u \in \mathcal C_n} (\|u_{\breve S_{d, n}}\|_1)^{-2} (s_n+1) u' \left ( \frac{1}{n} \sum_{1 \leq i \leq 2n} I \{D_i = d\} \breve \psi_{n, i} \breve \psi_{n, i}' \right ) u \geq 0.25 \kappa_1^2~. \]
Therefore, we have
\begin{align*}
0.25 \kappa_1^2 \|(\breve \delta_{d, n})_{\breve S_{d, n}}\|_1^2 & \leq \frac{1}{n} \sum_{1 \leq i \leq 2n} I \{D_i = d\} (\breve \psi_{n, i}' \breve \delta_{d, n})^2 \\
& \leq \left(\frac{12 \bar{\sigma}}{\ubar\sigma \ell \ell_n} + d_n + \bar{c}\right)\lambda_{d, n}^{\rm r} (s_n+1) \|(\breve \delta_{d, n})_{S_{d, n}}\|_1~,     
\end{align*}

which implies
\[ \|(\breve \delta_{d, n})_{S_{d, n}}\|_1 \leq 4\left(\frac{12 \bar{\sigma}}{\ubar\sigma \ell \ell_n} + d_n + \bar{c}\right) (s_n+1) \lambda_{d, n}^{\rm r} / \kappa_1^2~. \]
We then have
\begin{align*}
|\hat \alpha_{d,n}^{\rm r} - \alpha_{d,n}^{\rm r}|+\|\hat \beta_{d, n}^{\rm r} - \beta_{d, n}^{\rm r}\|_1 & \leq \|(\breve \delta_{d, n})_{\breve S_{d, n}}\|_1 + \|(\breve \delta_{d, n})_{\breve S_{d, n}^c}\|_1 \\
& \leq \left(1 + \frac{2 \bar{\sigma} \bar{c}}{\ubar\sigma \underline{c}}\right) \|(\breve \delta_{d, n})_{\breve S_{d, n}}\|_1 \\
& \leq 4\left(1 + \frac{2 \bar{\sigma} \bar{c}}{\ubar\sigma \underline{c}}\right)\left(\frac{12 \bar{\sigma}}{\ubar\sigma \ell \ell_n} + d_n + \bar{c}\right) (s_n+1) \lambda_{d, n}^{\rm r} / \kappa_1^2~.    
\end{align*}
Then, we have \eqref{eq:beta-hd2} holds because $P \{\mathcal E_n(d)\} \to 1$.

\underline{Step 2: Verifying \eqref{eq:rate},  \eqref{eq:L2}, and Assumption \ref{ass:md}}

By \eqref{eq:basic3}, on $\mathcal E_n(d)$ we have 
\begin{align*}
  \frac{1}{n} \sum_{1 \leq i \leq 2n} I \{D_i = d\} (\hat \alpha_{d,n}^{\rm r}-\alpha_{d,n}^{\rm r} + \psi_{n, i}' \hat \delta_{d, n})^2 & \leq    \lambda_{d, n}^{\rm r} \left(\frac{12 \bar{\sigma}}{\ubar\sigma \ell \ell_n} + d_n + \bar{c}\right)\bar{\sigma} \|(\breve \delta_{d, n})_{\breve S_{d, n}}\|_1 \\
  & \leq 4\left(\frac{12 \bar{\sigma}}{\ubar\sigma \ell \ell_n} + d_n + \bar{c}\right)^2 \bar\sigma (s_n+1) \lambda_{d, n}^{\rm r,2} / \kappa_1^2~.
\end{align*}
Because $P \{\mathcal E_n(d)\} \to 1$, we have
\[ \frac{1}{n} \sum_{1 \leq i \leq 2n} I \{D_i = d\} (\hat \alpha_{d,n}^{\rm r}-\alpha_{d,n}^{\rm r} +\psi_{n, i}' (\hat \beta_{d, n} - \beta_{d, n}))^2 = O_P \left ( s_n  (\lambda_n^{\rm r})^2 \right ) = o_P(1)~, \]
which implies \eqref{eq:L2} holds. 

Next, we show \eqref{eq:rate} for $\hat \beta_{d, n}^{\rm r}$. First note
\begin{align*}
& \left | \frac{1}{\sqrt{2n}} \sum_{1 \leq i \leq 2n} (2 D_i - 1) (\hat m_{d, n}(X_i, W_{n, i}) - m_{d, n}(X_i, W_{n, i})) \right | \\
& \left | \frac{1}{\sqrt{2n}} \sum_{1 \leq i \leq 2n} (2 D_i - 1) \psi_{n,i}'(\hat \beta_{d,n}^{\rm r} - \beta_{d,n}^{\rm r}) \right | \\
& \leq \left \| \frac{1}{\sqrt{2n}} \sum_{1 \leq i \leq n} (2 D_i - 1) \psi_{n, i} \right \|_\infty \|\hat \beta_{d, n}^{\rm r} - \beta_{d, n}^{\rm r}\|_1~.
\end{align*}
Next, note that it follows from Assumption \ref{ass:treatment} that conditional on $X^{(n)}$ and $W_n^{(n)}$,
\[ \{D_{\pi(2j - 1)} - D_{\pi(2j)}: 1 \leq j \leq n\} \]
is a sequence of independent Rademacher random variables. Therefore, Hoeffding's inequality implies
\begin{align*}
& P \left \{ \left \| \frac{1}{\sqrt{2n}} \sum_{1 \leq i \leq 2n} (2 D_i - 1) \psi_{n, i} \right \|_\infty > t \Bigg | X^{(n)}, W_n^{(n)} \right \} \\
& \leq \sum_{1 \leq l \leq p_n} P \left \{ \left\vert \frac{1}{\sqrt{2n}} \sum_{1 \leq j \leq n} (\psi_{n, \pi(2j - 1)} - \psi_{n, \pi(2j)}) (D_{\pi(2j - 1)} - D_{\pi(2j)})\right\vert  > t \Bigg | X^{(n)}, W_n^{(n)} \right \} \\
& \leq \sum_{1 \leq l \leq p_n} 2 \exp \left ( - \frac{t^2}{\frac{1}{n} \sum_{1 \leq j \leq n} (\psi_{n, \pi(2j - 1)} - \psi_{n, \pi(2j)})^2} \right )~.
\end{align*}
Define
\[ \nu_n^2 = \max_{1 \leq l \leq p_n} \frac{1}{n} \sum_{1 \leq i \leq 2n} \psi_{n, i, l}^2~. \]
We then have
\begin{equation} \label{eq:ratepsi2}
P \left \{ \left \| \frac{1}{\sqrt{2n}} \sum_{1 \leq i \leq 2n} (2 D_i - 1) \psi_{n, i} \right \|_\infty > \nu_n \sqrt{2 \log(p_n \vee n)} \Bigg | X^{(n)}, W_n^{(n)} \right \} \leq (p_n \vee n)^{-1}~.
\end{equation}
Next, we determine the order of $\nu_n^2$. Note
\begin{align*}
E[\nu_n^2] & \leq \max_{1 \leq l \leq p_n} 2 E[\psi_{n, i, l}^2] + 2 E \left [ \frac{1}{2n} \sum_{1 \leq i \leq 2n} (\psi_{n, i, l}^2 - E[\psi_{n, i, l}^2]) \right ] \\
& \lesssim 1 + E \left [ \max_{1 \leq l \leq p_n} \left | \frac{1}{2n} \sum_{1 \leq i \leq 2n} e_i \psi_{n, i, l}^2 \right | \right ] \\
& \lesssim 1 + \Xi_n E \left [ \max_{1 \leq l \leq p_n} \left | \frac{1}{2n} \sum_{1 \leq i \leq 2n} e_i \psi_{n, i, l} \right | \right ] \\
& \lesssim 1 + \Xi_n E \left [ \sup_{f \in \mathcal F_n} \left | \frac{1}{2n} \sum_{1 \leq i \leq 2n} f(e_i, \psi_{n, i, l}) \right | \right ]
\end{align*}
where $\{e_i: 1 \leq i \leq n\}$ is an i.i.d.\ sequence of Rademacher random variables, 
\[ \mathcal F_n = \{f: \mathbf R \times \mathbf R^{p_n} \mapsto \mathbf R,  f(e, \psi) = e\psi_l, 1 \leq l \leq p_n\}~, \]
and $\psi_l$ is the $l$th element of $\psi$. Note the second inequality follows from Lemma 2.3.1 of \cite{vandervaart:wellner:1996}, the third inequality follows from Theorem 4.12 of \cite{ledoux1991probability} and the definition of $\Xi_n$, and the last follows from Assumption \ref{ass:LASSO2-moments}. Note also $\mathcal F_n$ has an envelope $F = \Xi_n$ and
\[ \sup_{n \geq 1} \sup_{f \in \mathcal F_n} E[f^2] < \infty \]
because of Assumption \ref{ass:LASSO2-moments}. Because the cardinality of $\mathcal F_n$ is $p_n$, for any $\epsilon < 1$ we have that
\[ \sup_{Q: Q \text{ is a discrete distribution with finite support}} \mathcal N(\epsilon \|F\|_{Q, 2}, \mathcal F, L_2(Q)) \leq \frac{p_n}{\epsilon}~, \]
where $\mathcal N(\epsilon, \mathcal F, L_2(Q))$ is the covering number for class $\mathcal F$ under the metric $L_2(Q)$ using balls of radius $\epsilon$. Therefore, Corollary 5.1 of \cite{chernozhukov2014gaussian} implies
\[ E \left [ \sup_{f \in \mathcal F_n} \left | \frac{1}{2n} \sum_{1 \leq i \leq 2n} e_i \psi_{n, i, l}  \right | \right ] \lesssim \sqrt{\frac{\log p_n}{n}} + \frac{\Xi_n \log p_n}{n} = o(\Xi_n^{-1})~. \]
Therefore, $\nu_n = O_p(1)$. Together with \eqref{eq:ratepsi2}, they imply
\[ \left \| \frac{1}{\sqrt{2n}} \sum_{1 \leq i \leq 2n} (2 D_i - 1) \psi_{n, i} \right \|_\infty = O_P \left ( \sqrt{\log (p_n \vee n)} \right )~. \]
In light of \eqref{eq:beta-hd2} and Assumption \ref{ass:LASSO2-penalty}, we have
\[ \left \| \frac{1}{\sqrt{2n}} \sum_{1 \leq i \leq 2n} (2 D_i - 1) \psi_{n, i} \right \|_\infty \|\hat \beta_{d, n}^{\rm r} - \beta_{d, n}^{\rm r}\|_1 = O_P \left ( \frac{s_n \ell\ell_n\log (p_n \vee n)}{\sqrt n} \right ) = o_P(1)~. \]

Next, note that Assumption \ref{ass:md}(a) and \ref{ass:md}(b) follow Assumption \ref{ass:LASSO2-moments}, and Assumption \ref{ass:md}(c) follows Assumptions \ref{ass:LASSO2-moments} and \ref{ass:LASSO2-match}.

\underline{Step 3: Asymptotic variance} 

Suppose the true specification is approximately sparse as specified in Theorem \ref{thm:LASSO2}. Let $\tilde Y_i(d) = Y_i(d) - \mu_d(X_i)$, $\tilde \psi_{n, i} = \psi_{n,i} - E[\psi_{n, i}|X_i]$, and $\tilde R_{n,i}(d) = R_{n,i}(d) - E[R_{n,i}(d)|X_i]$. Then, we have 
\begin{align*}
& E \left[(E[\tilde Y_i(1) + \tilde Y_i(0)|W_{n,i},X_i] - \tilde \psi_{n,i}'(\beta_{1,n}^{\rm r} + \beta_{0,n}^{\rm r}))^2\right] = E[(\tilde R_{n,i}(1)+ \tilde R_{n,i}(0))^2] = o(1)~. 
\end{align*}
This concludes the proof.  \qed

\subsection{Proof of Theorem \ref{thm:refit}}
We divide the proof into three steps. In the first step, we show
\begin{align}
\hat \beta_n^{\rm refit} - \beta_n^{\rm refit} = o_P(1)~.
\label{eq:beta^pfehd}
\end{align}
In the second step, we show \eqref{eq:rate}, \eqref{eq:L2}, and Assumption \ref{ass:md} hold. In the third step, we show that  $\sigma_n^{\rm na,2} \geq \sigma_n^{\rm refit,2}$ and $\sigma_n^{\rm r,2}\geq \sigma_n^{\rm refit,2}$. 


\underline{Step 1: Proof of \eqref{eq:beta^pfehd}} 

Let 
\begin{align*}
\hat \Delta_{\Gamma, n} & = \frac{1}{n} \sum_{1 \leq j \leq n} (D_{\pi(2j - 1)} - D_{\pi(2j)}) (\Gamma_{n,\pi(2j - 1)} - \Gamma_{n,\pi(2j)})~,\\
\hat \Delta_{\hat \Gamma, n} & = \frac{1}{n} \sum_{1 \leq j \leq n} (D_{\pi(2j - 1)} - D_{\pi(2j)}) (\hat \Gamma_{n,\pi(2j - 1)} - \hat \Gamma_{n,\pi(2j)})~,\\
\delta_{\Gamma, j} & = (D_{\pi(2j - 1)} - D_{\pi(2j)})(\Gamma_{n,\pi(2j - 1)} - \Gamma_{n,\pi(2j)})~,\\
\delta_{\hat \Gamma, j} & = (D_{\pi(2j - 1)} - D_{\pi(2j)})(\hat \Gamma_{n,\pi(2j - 1)} - \hat \Gamma_{n,\pi(2j)})~.
\end{align*}
Then, by the proof of Theorem \ref{thm:pfe}, we have $\hat \beta_n^{\rm refit}$ equals the coefficient estimate using least squares in the linear regression of $\delta_{Y, j}$ on $\delta_{\hat \Gamma, j} - \hat{\Delta}_{\hat \Gamma, n}$. Then, for any $u \in \mathbf R^2$ such that $||u||_2 =1$, we have 
\begin{align*}
& \left|\left( \frac{1}{n} \sum_{1\leq j \leq n} ((\delta_{\hat \Gamma, j} - \hat \Delta_{\hat \Gamma, n})'u)^2 \right)^{1/2} - \left( \frac{1}{n} \sum_{1\leq j \leq n} ((\delta_{\Gamma, j}- \hat \Delta_{ \Gamma, n})'u)^2 \right)^{1/2} \right| \\
& \leq  \left( \frac{1}{n} \sum_{1\leq j \leq n} ((\delta_{\hat \Gamma, j} - \delta_{\Gamma, j})'u)^2 - ((\hat \Delta_{\hat \Gamma, n}- \hat \Delta_{ \Gamma, n})'u)^2 \right)^{1/2} \\
& \leq \left( \frac{2}{n} \sum_{1\leq i \leq 2n} \left\|\hat\Gamma_{n,i} + (\hat \alpha_{1,n}^{\rm r},\hat \alpha_{0,n}^{\rm r})' - \Gamma_{n,i} - (\alpha_{1,n}^{\rm r},\alpha_{0,n}^{\rm r})'\right\|_2^2 \right)^{1/2} \\
& \lesssim \sum_{d \in \{0, 1\}} \frac{1}{2n}\sum_{1 \leq i \leq 2n}(\hat \alpha_{d,n}^{\rm r} - \alpha_{d,n}^{\rm r}+ \psi_{n, i}'(\hat \beta_{d, n}^{\rm r} - \beta_{d, n}^{\rm r}))^2 = o_P(1)~, 
\end{align*}
where the second inequality is by the fact that 
\begin{align*}
&    \delta_{\Gamma,j} = (D_{\pi(2j - 1)} - D_{\pi(2j)})(\Gamma_{n,\pi(2j - 1)} + ( \alpha_{1,n}^{\rm r},\alpha_{0,n}^{\rm r})' - \Gamma_{n,\pi(2j)} - (\alpha_{1,n}^{\rm r},\alpha_{0,n}^{\rm r})')~,\\
&    \delta_{\hat \Gamma,j} = (D_{\pi(2j - 1)} - D_{\pi(2j)})(\hat \Gamma_{n,\pi(2j - 1)} + (\hat \alpha_{1,n}^{\rm r},\hat \alpha_{0,n}^{\rm r})' - \hat \Gamma_{n,\pi(2j)} - (\hat \alpha_{1,n}^{\rm r},\hat \alpha_{0,n}^{\rm r})')~,
\end{align*}
and the last equality is by the proof of Theorem \ref{thm:LASSO2}. This implies 
\begin{align*}
& \frac{1}{n} \sum_{1\leq j \leq n} (\delta_{\hat \Gamma, j} - \hat \Delta_{\hat \Gamma, n})(\delta_{\hat \Gamma, j} - \hat \Delta_{\hat \Gamma, n})' - 2 E[\var[\Gamma_{n,i}|X_i]] \\
& = \frac{1}{n} \sum_{1\leq j \leq n} (\delta_{\hat \Gamma, j} - \hat \Delta_{\hat \Gamma, n})(\delta_{\hat \Gamma, j} - \hat \Delta_{\hat \Gamma, n})'-\frac{1}{n} \sum_{1\leq j \leq n} (\delta_{ \Gamma, j} - \hat \Delta_{ \Gamma, n})(\delta_{ \Gamma, j} - \hat \Delta_{\Gamma, n})' \\
& + \frac{1}{n} \sum_{1\leq j \leq n} (\delta_{ \Gamma, j} - \hat \Delta_{ \Gamma, n})(\delta_{ \Gamma, j} - \hat \Delta_{\Gamma, n})'- 2 E[\var[\Gamma_{n,i}|X_i]] = o_P(1)~,
\end{align*}
where the last equality holds due to the same argument as used in the proof of Theorem \ref{thm:pfe}. 
Similarly, we can show that 
\begin{align*}
\frac{1}{n} \sum_{1\leq j \leq n} \delta_{Y, j}(\delta_{\hat \Gamma, j} - \hat \Delta_{\hat \Gamma, n}) - E [\cov[\Gamma_{n,i},Y_i(1) + Y_i(0)|X_i]] = o_P(1)~,
\end{align*}
which leads to \eqref{eq:beta^pfehd}. 

\underline{Step 2: Verifying \eqref{eq:rate}, \eqref{eq:L2}, and Assumption \ref{ass:md}}

We first show \eqref{eq:rate}. We have
\begin{align*}
& \frac{1}{\sqrt{2n}}\sum_{1 \leq i \leq 2n} (2D_i -1)(\hat m_{d, n}(X_i, W_{n, i}) -  m_{d, n}(X_i, W_{n, i})) \\
& = \frac{1}{\sqrt{2n}}\sum_{1 \leq i \leq 2n} (2D_i -1)(\hat \Gamma_{n, i} - \Gamma_{n, i})'\hat \beta_{ n}^{\rm refit} + \frac{1}{\sqrt{2n}}\sum_{1 \leq i \leq 2n} (2D_i -1) \Gamma_{n, i}'(\hat \beta_{ n}^{\rm refit} - \beta_{ n}^{\rm refit} ) \\
& = 
\begin{pmatrix}
\frac{1}{\sqrt{2n}}\sum_{1 \leq i \leq 2n} (2D_i -1)\psi_{n,i}'(\hat \beta_{1, n}^{\rm r} - \beta_{1, n}^{\rm r}),\frac{1}{\sqrt{2n}}\sum_{1 \leq i \leq 2n} (2D_i -1)\psi_{n,i}'(\hat \beta_{0, n}^{\rm r} - \beta_{0, n}^{\rm r})
\end{pmatrix}\hat \beta_{ n}^{\rm refit} \\
&  + \begin{pmatrix}
\frac{1}{\sqrt{2n}}\sum_{1 \leq i \leq 2n} (2D_i -1)\psi_{n,i}' \beta_{1, n}^{\rm r}, \frac{1}{\sqrt{2n}}\sum_{1 \leq i \leq 2n} (2D_i -1)\psi_{n,i}' \beta_{0, n}^{\rm r}
\end{pmatrix} (\hat \beta_{1, n}^{\rm refit} - \beta_{1, n}^{\rm refit}) \\
& = o_P(1)~,
\end{align*}
where the last equality holds by \eqref{eq:beta^pfehd} and the facts that 
\begin{align*}
\begin{pmatrix}
\frac{1}{\sqrt{2n}}\sum_{1 \leq i \leq 2n} (2D_i -1)\psi_{n,i}'(\hat \beta_{1, n}^{\rm r} - \beta_{1, n}^{\rm r}),\frac{1}{\sqrt{2n}}\sum_{1 \leq i \leq 2n} (2D_i -1)\psi_{n,i}'(\hat \beta_{0, n}^{\rm r} - \beta_{0, n}^{\rm r})
\end{pmatrix} = o_P(1)
\end{align*}
as shown in Theorem \ref{thm:LASSO2} and 
\begin{align*}
 \begin{pmatrix}
\frac{1}{\sqrt{2n}}\sum_{1 \leq i \leq 2n} (2D_i -1)\psi_{n,i}' \beta_{1, n}^{\rm r}, \frac{1}{\sqrt{2n}}\sum_{1 \leq i \leq 2n} (2D_i -1)\psi_{n,i}' \beta_{0, n}^{\rm r}
\end{pmatrix} = O_P(1)~. 
\end{align*}

Next, we show \eqref{eq:L2}. We note that 
\begin{align*}
& \frac{1}{2n}\sum_{1 \leq i \leq 2n}(\hat m_{d, n}(X_i, W_{n, i}) -  m_{d, n}(X_i, W_{n, i}))^2 \\
& \lesssim \frac{1}{2n}\sum_{1 \leq i \leq 2n} ((\hat \Gamma_{n, i} - \Gamma_{n, i})'\hat \beta_{ n}^{\rm refit})^2 + \frac{1}{2n}\sum_{1 \leq i \leq 2n} (\Gamma_{n, i}'(\hat \beta_{n}^{\rm refit} - \beta_{n}^{\rm refit}))^2 \\
& \lesssim \frac{1}{2n}\sum_{1 \leq i \leq 2n} ((\psi_{n,i}'(\beta_{1,n}^{\rm r}-\hat \beta_{1,n}^{\rm r}))^2 + (\psi_{n, i}'(\beta_{0,n}^{\rm r}-\hat \beta_{0,n}^{\rm r}))^2) ||\hat \beta_{ n}^{\rm refit}||^2_2 \\
& + \frac{1}{2n}\sum_{1 \leq i \leq 2n}[ (\psi_{n,i}'\beta_{1,n}^{\rm r})^2 + (\psi_{n, i}'\beta_{0,n}^{\rm r})^2] ||\hat \beta_{n}^{\rm refit} - \beta_{n}^{\rm refit}||_2^2 \\
& \lesssim \sum_{d=0,1}\frac{1}{2n}\sum_{1 \leq i \leq 2n} \left[( \alpha_{d,n}^{\rm r}-\hat \alpha_{d,n}^{\rm r} + \psi_{n,i}'(\beta_{d,n}^{\rm r}-\hat \beta_{d,n}^{\rm r}))^2 + (\alpha_{d,n}^{\rm r}-\hat \alpha_{d,n}^{\rm r})^2\right] ||\hat \beta_{ n}^{\rm refit}||^2_2 + o_P(1) \\
& = o_P(1)~. 
\end{align*}

Last, Assumption \ref{ass:md}(1) can be verified in the same manner as we did in the proof of Theorem \ref{thm:LASSO2}. 

\underline{Step 3: Asymptotic variance}

Recall $\sigma_2^2(Q)$ and $\sigma_3^2(Q)$ defined in Theorem \ref{thm:main}. As we have already verified \eqref{eq:rate} for $\hat m_{d,n}(X_i,W_{n,i}) =\hat \Gamma_{n, i} \hat \beta_{ n}^{\rm refit}$ and $ m_{d,n}(X_i,W_{n,i}) = \Gamma_{n, i} \beta_{ n}^{\rm refit}$, we have, for $b \in \{\rm{unadj, r, refit}\}$, that  
\begin{align*}
\sigma_n^{\rm b,2} - \sigma_2^2(Q)-\sigma_3^2(Q) = \frac{1}{2} E\left[ \var[E[Y_i(1)+Y_i(0)|X_i,W_{n,i}] - \Gamma_{n, i}' \gamma^b|X_i]\right]
\end{align*}
with
\begin{align*}
\gamma^{\rm unadj} = (0,0)'~, \quad \gamma^{\rm r} = (1,1)'~, \quad \text{and} \quad \gamma^{\rm refit} = \beta_{ n}^{\rm refit}~.
\end{align*}
In addition, we note  that 
\begin{align*}
\frac{1}{2} E\left[ \var[E[Y_i(1)+Y_i(0)|X_i,W_{n,i}] - \Gamma_{n, i}' \gamma|X_i]\right]
\end{align*}
is minimized at $\gamma = \beta_{ n}^{\rm refit}$, which leads to the desired result. \qed

\section{Auxiliary Lemmas} \label{sec:aux_lemmas}
\begin{lemma} \label{lem:ui}
Suppose $\phi_n, n \geq 1$ is a sequence of random variables satisfying
\begin{equation} \label{eq:ui}
\lim_{\lambda \to \infty} \limsup_{n \to \infty} E[|\phi_n| I \{|\phi_n| > \lambda\}] = 0~.
\end{equation}
Suppose $X$ is another random variable defined on the same probability space with $\phi_n, n \geq 1$. Then,
\begin{equation} \label{eq:uic}
\lim_{\gamma \to \infty} \limsup_{n \to \infty} E[ E[|\phi_n| | X] I \{E[|\phi_n| | X] > \gamma\}] = 0~.
\end{equation}
\end{lemma}

\begin{proof}
Fix $\epsilon > 0$. We will show there exists $\gamma > 0$ so that
\begin{equation} \label{eq:goal}
\limsup_{n \to \infty} E[ E[|\phi_n| | X] I \{E[|\phi_n| | X] > \gamma\}] < \epsilon~.
\end{equation}
First note the event $\{E[|\phi_n| | X] > \gamma\}$ is measurable with respect to the $\sigma$-algebra generated by $X$, and therefore
\begin{equation} \label{eq:mes}
E[ E[|\phi_n| | X] I \{E[|\phi_n| | X] > \gamma\}] = E[|\phi_n| I \{E[|\phi_n| | X] > \gamma\}]~.
\end{equation}
Next, by Theorem 10.3.5 of \cite{dudley:1989}, \eqref{eq:ui} implies that there exists a $\delta > 0$ such that for any sequence of events $A_n$ such that $\limsup_{n \to \infty} P \{A_n\} < \delta$, we have
\begin{equation} \label{eq:eps}
\limsup_{n \to \infty} E[|\phi_n| I \{A_n\}] < \epsilon~.
\end{equation}
In light of the previous result, note
\begin{align*}
P \{E[|\phi_n| | X] > \gamma\} & \leq \frac{E[E[|\phi_n| | X]]}{\gamma} = \frac{E[|\phi_n|]}{\gamma}
\end{align*}
By Theorem 10.3.5 of \cite{dudley:1989} again, \eqref{eq:ui} implies $\limsup_{n \to \infty} E[|\phi_n|] < \infty$, so by choosing $\gamma$ large enough, we can make sure
\[ \limsup_{n \to \infty} P \{E[|\phi_n| | X] > \gamma\} < \delta \text{ for all } n~. \]
\eqref{eq:goal} then follows from \eqref{eq:mes}--\eqref{eq:eps}.
\end{proof}

\begin{lemma} \label{lem:sn}
Suppose Assumptions \ref{ass:Q}--\ref{ass:close} and \ref{ass:md} hold. Then,
\[ \frac{s_n^2}{n E[\var[\phi_{1, n, i} | X_i]]} \stackrel{P}{\to} 1~. \]
\end{lemma}

\begin{proof}
To begin, note it follows from Assumption \ref{ass:treatment} and $Q_n = Q^{2n}$ that
\begin{multline} \label{eq:varA} 
\frac{1}{n} \sum_{1 \leq i \leq 2n} D_i \var[\phi_{1, n, i} | X_i] = \frac{1}{2n} \sum_{1 \leq i \leq 2n} \var[\phi_{1, n, i} | X_i] \\
+ \frac{1}{2n} \sum_{1 \leq i \leq 2n: D_i = 1} \var[\phi_{1, n, i} | X_i] - \frac{1}{2n} \sum_{1 \leq i \leq 2n: D_i = 0} \var[\phi_{1, n, i} | X_i]~.
\end{multline}
Next,
\begin{multline} \label{eq:varphi1}
\left | \frac{1}{2n} \sum_{1 \leq i \leq 2n: D_i = 1} \var[\phi_{1, n, i} | X_i] - \frac{1}{2n} \sum_{1 \leq i \leq 2n: D_i = 0} \var[\phi_{1, n, i} | X_i] \right | \\
\leq \frac{1}{2n} \sum_{1 \leq j \leq n} |\var[\phi_{1, n, \pi(2j - 1)} | X_{\pi(2j - 1)}] - \var[\phi_{1, n, \pi(2j)} | X_{\pi(2j)}]|~.
\end{multline}
In what follows, we will show
\begin{multline*}
\frac{1}{n} \sum_{1 \leq j \leq n} |\cov[Y_{\pi(2j - 1)}(1), m_{1, n}(X_{\pi(2j - 1)}, W_{\pi(2j - 1)}) | X_{\pi(2j - 1)}]] \\
- \cov[Y_{\pi(2j)}(1), m_{1, n}(X_{\pi(2j)}, W_{\pi(2j)}) | X_{\pi(2j)}]]| \stackrel{P}{\to} 0~.
\end{multline*}
To that end, first note from Assumptions \ref{ass:close} and \ref{ass:md}(c) that
\begin{align*}
& \frac{1}{n} \sum_{1 \leq j \leq n} |E[Y_{\pi(2j - 1)}(1) m_{1, n}(X_{\pi(2j - 1)}, W_{\pi(2j - 1)}) | X_{\pi(2j - 1)}] - E[Y_{\pi(2j)}(1) m_{1, n}(X_{\pi(2j)}, W_{\pi(2j)}) | X_{\pi(2j)}]| \\
& \lesssim \frac{1}{n} \sum_{1 \leq j \leq n} |X_{\pi(2j - 1)} - X_{\pi(2j)}| \stackrel{P}{\to} 0~.
\end{align*}
Next, note
\begin{align*}
& \frac{1}{n} \sum_{1 \leq j \leq n} |E[Y_{\pi(2j - 1)}(1) | X_{\pi(2j - 1)}] E[m_{1, n}(X_{\pi(2j - 1)}, W_{\pi(2j - 1)}) | X_{\pi(2j - 1)}] \\
& \hspace{3em} - E[Y_{\pi(2j)}(1) | X_{\pi(2j)}] E[m_{1, n}(X_{\pi(2j)}, W_{\pi(2j)}) | X_{\pi(2j)}]| \\
& \leq \frac{1}{n} \sum_{1 \leq j \leq n} |E[Y_{\pi(2j - 1)}(1) | X_{\pi(2j - 1)}]| |E[m_{1, n}(X_{\pi(2j - 1)}, W_{\pi(2j - 1)}) | X_{\pi(2j - 1)}] - E[m_{1, n}(X_{\pi(2j)}, W_{\pi(2j)}) | X_{\pi(2j)}]| \\
& \hspace{3em} + \frac{1}{n} \sum_{1 \leq j \leq n} |E[Y_{\pi(2j - 1)}(1) | X_{\pi(2j - 1)}] - E[Y_{\pi(2j)}(1) | X_{\pi(2j)}]| |E[m_{1, n}(X_{\pi(2j)}, W_{\pi(2j)}) | X_{\pi(2j)}]| \\
& \leq \left ( \frac{1}{n} \sum_{1 \leq j \leq n} |E[Y_{\pi(2j - 1)}(1) | X_{\pi(2j - 1)}]|^2 \right )^{1/2} \\
& \hspace{5em} \times \left ( \frac{1}{n} \sum_{1 \leq j \leq n} |E[m_{1, n}(X_{\pi(2j - 1)}, W_{\pi(2j - 1)}) | X_{\pi(2j - 1)}] - E[m_{1, n}(X_{\pi(2j)}, W_{\pi(2j)}) | X_{\pi(2j)}]|^2 \right )^{1/2} \\
& \hspace{3em} + \left ( \frac{1}{n} \sum_{1 \leq j \leq n} |E[m_{1, n}(X_{\pi(2j)}, W_{\pi(2j)}) | X_{\pi(2j)}]|^2 \right )^{1/2} \\
& \hspace{5em} \times \left ( \frac{1}{n} \sum_{1 \leq j \leq n} |E[Y_{\pi(2j - 1)}(1) | X_{\pi(2j - 1)}] - E[Y_{\pi(2j)}(1) | X_{\pi(2j)}]|^2 \right )^{1/2} \\
& \lesssim \left ( \frac{1}{n} \sum_{1 \leq i \leq 2n} |E[Y_i(1) | X_i]|^2 \right )^{1/2} \left ( \frac{1}{n} \sum_{1 \leq j \leq n} |X_{\pi(2j - 1)} - X_{\pi(2j)}|^2 \right )^{1/2} \\
& \hspace{3em} + \left ( \frac{1}{n} \sum_{1 \leq i \leq 2n} |E[m_{1, n}(X_i, W_i) | X_i]|^2 \right )^{1/2} \left ( \frac{1}{n} \sum_{1 \leq j \leq n} |X_{\pi(2j - 1)} - X_{\pi(2j)}|^2 \right )^{1/2} \stackrel{P}{\to} 0~,
\end{align*}
where the first inequality follows from the triangle inequality, the second follows from the Cauchy-Schwarz inequality, the last follows from Assumptions \ref{ass:Q}(c) and \ref{ass:md}(c). To see the convergence holds, first note because
\[ E[|E[Y_i(1) | X_i]|^2] \leq E[E[Y_i^2(1) | X_i]] = E[Y_i^2(1)] < \infty~, \]
the weak law of large numbers implies
\[ \frac{1}{n} \sum_{1 \leq i \leq 2n} |E[Y_i(1) | X_i]|^2 \stackrel{P}{\to} 2 E[|E[Y_i(1) | X_i]|^2] < \infty~. \]
On the other hand,
\[ \frac{1}{2n} \sum_{1 \leq i \leq 2n} |E[m_{1, n}(X_i, W_i) | X_i]|^2 \leq \frac{1}{2n} \sum_{1 \leq i \leq 2n} E[m_{1, n}^2(X_i, W_i) | X_i]~. \]
Assumption \ref{ass:md}(b) and Lemma \ref{lem:ui} imply
\[ \lim_{\lambda \to \infty} \limsup_{n \to \infty} E[E[m_{1, n}^2(X_i, W_i) | X_i] I \{E[m_{1, n}^2(X_i, W_i) | X_i] > \lambda\}] = 0~. \]
Therefore, Lemma 11.4.2 of \cite{lehmann:romano:tsh:2005} implies
\[ \frac{1}{2n} \sum_{1 \leq i \leq 2n} E[m_{1, n}^2(X_i, W_i) | X_i] - E[E[m_{1, n}^2(X_i, W_i) | X_i]] \stackrel{P}{\to} 0~. \]
Finally, note $E[E[m_{1, n}^2(X_i, W_i) | X_i]] = E[m_{1, n}^2(X_i, W_i)]$ is bounded for $n \geq 1$ by Assumption \ref{ass:md}(b), so
\[ \frac{1}{n} \sum_{1 \leq i \leq 2n} |E[m_{1, n}(X_i, W_i) | X_i]|^2 = O_P(1)~. \]
The desired convergence therefore follows.

Similar arguments applied termwise imply the right-hand side of \eqref{eq:varphi1} is $o_P(1)$. \eqref{eq:varA}--\eqref{eq:varphi1} then imply
\begin{equation} \label{eq:sn-d}
\frac{s_n^2}{n} - \frac{1}{2n} \sum_{1 \leq i \leq 2n} \var[\phi_{1, n, i} | X_i] \to 0~.
\end{equation}

Next, we argue
\begin{equation} \label{eq:sn-wlln}
\frac{1}{2n} \sum_{1 \leq i \leq 2n} \var[\phi_{1, n, i} | X_i] - E[\var[\phi_{1, n, i} | X_i]] \to 0~.
\end{equation}
To establish \eqref{eq:sn-wlln}, we verify the uniform integrability condition in Lemma 11.4.2 of \cite{lehmann:romano:tsh:2005}. To that end, we will repeatedly use the inequality
\begin{align}
\label{eq:ui-indicator1} \left | \sum_{1 \leq j \leq k} a_j \right | I \left \{ \left | \sum_{1 \leq j \leq k} a_j \right | > \lambda \right \} & \leq \sum_{1 \leq j \leq k} k |a_j| I \left \{ |a_j| > \frac{\lambda}{k} \right \} \\
\label{eq:ui-indicator2} |ab| I \{|ab| > \lambda\} & \leq |a|^2 I \{|a| > \sqrt \lambda\} + |b|^2 I \{|b| > \sqrt \lambda\}~.
\end{align}
Note
\begin{align*}
& E[|\var[\phi_{1, n, i} | X_i] - E[\var[\phi_{1, n, i} | X_i]]| I \{|\var[\phi_{1, n, i} | X_i] - E[\var[\phi_{1, n, i} | X_i]]| > \lambda\}] \\
& \lesssim E \left [ |\var[\phi_{1, n, i} | X_i]| I \left \{ |\var[\phi_{1, n, i} | X_i]| > \frac{\lambda}{2} \right \} \right ] + E[\var[\phi_{1, n, i} | X_i]] I \left \{ E[\var[\phi_{1, n, i} | X_i]] > \frac{\lambda}{2} \right \} \\
& \leq E \left [ E[\phi_{1, n, i}^2 | X_i] I \left \{ E[\phi_{1, n, i}^2 | X_i] > \frac{\lambda}{2} \right \} \right ] + E[\phi_{1, n, i}^2] I \left \{ E[\phi_{1, n, i}^2] > \frac{\lambda}{2} \right \}~,
\end{align*}
where in the second inequality we use the fact that the variance of a random variable is bounded by its second moment. Note Assumption \ref{ass:md} implies $E[\phi_{1, n, i}^2]$ is bounded for $n \geq 1$, and therefore
\[ \lim_{\lambda \to \infty} \limsup_{n \to \infty} E[\phi_{1, n, i}^2] I \left \{ E[\phi_{1, n, i}^2] > \frac{\lambda}{2} \right \} = 0~. \]
On the other hand
\begin{align}
\label{eq:sn-ui} & E \left [ E[\phi_{1, n, i}^2 | X_i] I \left \{ E[\phi_{1, n, i}^2 | X_i] > \frac{\lambda}{2} \right \} \right ] \\
\nonumber & \lesssim E \left [ E[Y_i^2(1) | X_i] I \left \{ E[Y_i^2(1) | X_i] > \frac{\lambda}{12} \right \} \right ] + E \left [ E[m_{1, n}^2(X_i, W_i) | X_i] I \left \{ E[m_{1, n}^2(X_i, W_i) | X_i] > \frac{\lambda}{3} \right \} \right ] \\
\nonumber & \hspace{3em} + E \left [ E[m_{0, n}^2(X_i, W_i) | X_i] I \left \{ E[m_{0, n}^2(X_i, W_i) | X_i] > \frac{\lambda}{3} \right \} \right ] \\
\nonumber & \hspace{3em} + E \left [ |E[Y_i(1) m_{1, n}(X_i, W_i) | X_i]| I \left \{ |E[Y_i(1) m_{1, n}(X_i, W_i) | X_i]| > \frac{\lambda}{12} \right \} \right ] \\
\nonumber & \hspace{3em} + E \left [ |E[Y_i(1) m_{0, n}(X_i, W_i) | X_i]| I \left \{ |E[Y_i(1) m_{0, n}(X_i, W_i) | X_i]| > \frac{\lambda}{12} \right \} \right ] \\
\nonumber & \hspace{3em} + E \left [ |E[m_{1, n}(X_i, W_i) m_{0, n}(X_i, W_i) | X_i]| I \left \{ |E[m_{1, n}(X_i, W_i) m_{0, n}(X_i, W_i) | X_i]| > \frac{\lambda}{6} \right \} \right ]~.
\end{align}
It follows from Assumptions \ref{ass:Q}(b) and \ref{ass:md}(b) together with Lemma \ref{lem:ui} that
\begin{align*}
\lim_{\lambda \to \infty} \limsup_{n \to \infty} E \left [ E[Y_i^2(1) | X_i] I \left \{ E[Y_i^2(1) | X_i] > \frac{\lambda}{12} \right \} \right ] & = 0 \\
\lim_{\lambda \to \infty} \limsup_{n \to \infty} E \left [ E[m_{1, n}^2(X_i, W_i) | X_i] I \left \{ E[m_{1, n}^2(X_i, W_i) | X_i] > \frac{\lambda}{3} \right \} \right ] & = 0 \\
\lim_{\lambda \to \infty} \limsup_{n \to \infty} E \left [ E[m_{0, n}^2(X_i, W_i) | X_i] I \left \{ E[m_{0, n}^2(X_i, W_i) | X_i] > \frac{\lambda}{3} \right \} \right ] & = 0~.
\end{align*}
For the last term in \eqref{eq:sn-ui}, note
\begin{align*}
& E \left [ |E[m_{1, n}(X_i, W_i) m_{0, n}(X_i, W_i) | X_i]| I \left \{ |E[m_{1, n}(X_i, W_i) m_{0, n}(X_i, W_i) | X_i]| > \frac{\lambda}{6} \right \} \right ] \\
& \leq E \left [ E[|m_{1, n}(X_i, W_i) m_{0, n}(X_i, W_i)| | X_i] I \left \{ E[|m_{1, n}(X_i, W_i) m_{0, n}(X_i, W_i)| | X_i] > \frac{\lambda}{6} \right \} \right ]~.
\end{align*}
Meanwhile,
\begin{align*}
& E \left [ E[|m_{1, n}(X_i, W_i) m_{0, n}(X_i, W_i)| I \left \{ |m_{1, n}(X_i, W_i) m_{0, n}(X_i, W_i)| > \lambda \right \} \right ] \\
& \leq E[m_{1, n}^2(X_i, W_i) I \{|m_{1, n}(X_i, W_i)| > \sqrt \lambda \}] + E[m_{0, n}^2(X_i, W_i) I \{|m_{0, n}(X_i, W_i)| > \sqrt \lambda \}]~.
\end{align*}
It then follows from the previous two inqualities, Assumption \ref{ass:md}(b), and Lemma \ref{lem:ui} that
\[ \lim_{\lambda \to \infty} \limsup_{n \to \infty} E \left [ |E[m_{1, n}(X_i, W_i) m_{0, n}(X_i, W_i) | X_i]| I \left \{ |E[m_{1, n}(X_i, W_i) m_{0, n}(X_i, W_i) | X_i]| > \frac{\lambda}{6} \right \} \right ] = 0~. \]
Similar arguments establish
\begin{align*}
\lim_{\lambda \to \infty} \limsup_{n \to \infty} E \left [ |E[Y_i(1) m_{1, n}(X_i, W_i) | X_i]| I \left \{ |E[Y_i(1) m_{1, n}(X_i, W_i) | X_i]| > \frac{\lambda}{12} \right \} \right ] & = 0 \\
\lim_{\lambda \to \infty} \limsup_{n \to \infty} E \left [ |E[Y_i(1) m_{0, n}(X_i, W_i) | X_i]| I \left \{ |E[Y_i(1) m_{0, n}(X_i, W_i) | X_i]| > \frac{\lambda}{12} \right \} \right ] & = 0~.
\end{align*}
Therefore, \eqref{eq:sn-wlln} follows. The conclusion then follows from \eqref{eq:sn-d}--\eqref{eq:sn-wlln} and Assumption \ref{ass:md}(a).
\end{proof}

\begin{lemma} \label{lem:lind-ui}
Suppose Assumptions \ref{ass:Q}--\ref{ass:close} and \ref{ass:md} hold. Then,
\[ \lim_{\gamma \to \infty} \limsup_{n \to \infty} E[|\phi_{1, n, i} - E[\phi_{1, n, i} | X_i]|^2 I \{|\phi_{1, n, i} - E[\phi_{1, n, i} | X_i]|^2 > \gamma\}] = 0~. \]
\end{lemma}

\begin{proof}
Note
\begin{align*}
& E[|\phi_{1, n, i} - E[\phi_{1, n, i} | X_i]|^2 I \{|\phi_{1, n, i} - E[\phi_{1, n, i} | X_i]|^2 > \gamma\}] \\
& \lesssim E \left [ (\phi_{1, n, i}^2 + E[\phi_{1, n, i} | X_i]^2)  I \left \{ \phi_{1, n, i}^2 + E[\phi_{1, n, i} | X_i]^2 > \frac{\gamma}{2} \right \} \right ] \\
& \lesssim E \left [ \phi_{1, n, i}^2  I \left \{ \phi_{1, n, i}^2 > \frac{\gamma}{4} \right \} \right ] + E \left [ E[\phi_{1, n, i} | X_i]^2  I \left \{ E[\phi_{1, n, i} | X_i]^2 > \frac{\gamma}{4} \right \} \right ]~.
\end{align*}
where the first inequality follows from $(a + b)^2 \leq 2(a^2 + b^2)$ and the second inequality follows from \eqref{eq:ui-indicator1}. Next, note
\begin{align*}
& E \left [ E[\phi_{1, n, i} | X_i]^2  I \left \{ E[\phi_{1, n, i} | X_i]^2 > \frac{\gamma}{4} \right \} \right ] \\
& \lesssim E \left [ E[Y_i(1) | X_i]^2  I \left \{ E[Y_i(1) | X_i]^2 > \frac{\gamma}{24} \right \} \right ] + E \left [ E[m_{1, n}(X_i, W_i) | X_i]^2  I \left \{ E[m_{1, n}(X_i, W_i) | X_i]^2 > \frac{\gamma}{6} \right \} \right ] \\
& \hspace{3em} + E \left [ E[m_{0, n}(X_i, W_i) | X_i]^2  I \left \{ E[m_{0, n}(X_i, W_i) | X_i]^2 > \frac{\gamma}{6} \right \} \right ] \\
& \hspace{3em} + E \left [ |E[Y_i(1) | X_i] E[m_{1, n}(X_i, W_i) | X_i]|  I \left \{ |E[Y_i(1) | X_i] E[m_{1, n}(X_i, W_i) | X_i]| > \frac{\gamma}{24} \right \} \right ] \\
& \hspace{3em} + E \left [ |E[Y_i(1) | X_i] E[m_{0, n}(X_i, W_i) | X_i]|  I \left \{ |E[Y_i(1) | X_i] E[m_{0, n}(X_i, W_i) | X_i]| > \frac{\gamma}{24} \right \} \right ] \\
& \hspace{3em} + E \left [ |E[m_{1, n}(X_i, W_i) | X_i] E[m_{0, n}(X_i, W_i) | X_i]|  I \left \{ |E[m_{1, n}(X_i, W_i) | X_i] E[m_{0, n}(X_i, W_i) | X_i]| > \frac{\gamma}{12} \right \} \right ] \\
& \lesssim E \left [ E[Y_i^2(1) | X_i]  I \left \{ E[Y_i^2(1) | X_i] > \frac{\gamma}{24} \right \} \right ] + E \left [ E[m_{1, n}^2(X_i, W_i) | X_i]  I \left \{ E[m_{1, n}^2(X_i, W_i) | X_i] > \frac{\gamma}{6} \right \} \right ] \\
& \hspace{3em} + E \left [ E[m_{0, n}^2(X_i, W_i) | X_i]  I \left \{ E[m_{0, n}^2(X_i, W_i) | X_i] > \frac{\gamma}{6} \right \} \right ] \\
& \hspace{3em} + E \left [ |E[Y_i(1) | X_i]| I \left \{ |E[Y_i(1) | X_i]| > \sqrt{\frac{\gamma}{24}} \right \} \right ] \\
& \hspace{3em} + E \left [ |E[m_{1, n}(X_i, W_i) | X_i]| I \left \{ |E[m_{1, n}(X_i, W_i) | X_i]| > \sqrt{\frac{\gamma}{24}} \right \} \right ] \\
& \hspace{3em} + E \left [ |E[m_{0, n}(X_i, W_i) | X_i]| I \left \{ |E[m_{0, n}(X_i, W_i) | X_i]| > \sqrt{\frac{\gamma}{24}} \right \} \right ] \\
& \hspace{3em} + E \left [ |E[m_{1, n}(X_i, W_i) | X_i]| I \left \{ |E[m_{1, n}(X_i, W_i) | X_i]| > \sqrt{\frac{\gamma}{12}} \right \} \right ] \\
& \hspace{3em} + E \left [ |E[m_{0, n}(X_i, W_i) | X_i]| I \left \{ |E[m_{0, n}(X_i, W_i) | X_i]| > \sqrt{\frac{\gamma}{12}} \right \} \right ] \\
& \leq E \left [ E[Y_i^2(1) | X_i]  I \left \{ E[Y_i^2(1) | X_i] > \frac{\gamma}{24} \right \} \right ] + E \left [ E[m_{1, n}^2(X_i, W_i) | X_i]  I \left \{ E[m_{1, n}^2(X_i, W_i) | X_i] > \frac{\gamma}{6} \right \} \right ] \\
& \hspace{3em} + E \left [ E[m_{0, n}^2(X_i, W_i) | X_i]  I \left \{ E[m_{0, n}^2(X_i, W_i) | X_i] > \frac{\gamma}{6} \right \} \right ] \\
& \hspace{3em} + E \left [ E[Y_i^2(1) | X_i] I \left \{ E[Y_i^2(1) | X_i] > \frac{\gamma}{24} \right \} \right ] \\
& \hspace{3em} + E \left [ E[m_{1, n}^2(X_i, W_i) | X_i] I \left \{ E[m_{1, n}^2(X_i, W_i) | X_i] > \frac{\gamma}{24} \right \} \right ] \\
& \hspace{3em} + E \left [ E[m_{0, n}^2(X_i, W_i) | X_i] I \left \{ E[m_{0, n}^2(X_i, W_i) | X_i] > \sqrt{\frac{\gamma}{24}} \right \} \right ] \\
& \hspace{3em} + E \left [ E[m_{1, n}^2(X_i, W_i) | X_i] I \left \{ E[m_{1, n}^2(X_i, W_i) | X_i] > \sqrt{\frac{\gamma}{12}} \right \} \right ] \\
& \hspace{3em} + E \left [ E[m_{0, n}^2(X_i, W_i) | X_i] I \left \{ E[m_{0, n}^2(X_i, W_i) | X_i] > \sqrt{\frac{\gamma}{12}} \right \} \right ]~,
\end{align*}
where the first inequality follows from \eqref{eq:ui-indicator1}, the second one follows from the conditional Jensen's inequality and \eqref{eq:ui-indicator2}, and the third one follows again from the conditional Jensen's inequality. It then follows from Lemma \ref{lem:ui} together with Assumptions \ref{ass:Q}(b) and \ref{ass:md}(b) that
\[ \lim_{\gamma \to \infty} \limsup_{n \to \infty} E \left [ E[\phi_{1, n, i} | X_i]^2  I \left \{ E[\phi_{1, n, i} | X_i]^2 > \frac{\gamma}{4} \right \} \right ] = 0~. \]
Similar arguments lead to
\[ \lim_{\gamma \to \infty} \limsup_{n \to \infty} E \left [ \phi_{1, n, i}^2  I \left \{ \phi_{1, n, i}^2 > \frac{\gamma}{4} \right \} \right ] = 0~. \]
The conclusion then follows.
\end{proof}

\begin{lemma} \label{lem:En'}
	Suppose Assumptions in Theorem \ref{thm:LASSO2} hold. Then,
	\[P \left\{ \left\vert \frac{1}{n}\sum_{i \in [2n]}I\{D_i=d\}\epsilon_{n,i}(d) - E[\epsilon_{n,i}(d)] \right\vert \leq  \sqrt{\frac{\log (2n p_n)}{n}}\right\} \rightarrow 1\]
	and
	\[ P \left \{ \left \|\Omega_n^{-1}(d)\frac{1}{n} \sum_{1 \leq i \leq 2n} I \{D_i = d\} (\psi_{n, i} \epsilon_{n, i}(d) - E[\psi_{n, i} \epsilon_{n, i}(d)]) \right \|_\infty \leq \frac{6\bar{\sigma}}{\ubar\sigma}  \sqrt{\frac{\log (2n p_n)}{n}} \right \} \to 1~. \]
\end{lemma}

\begin{proof}
For the first result, we note that 
\begin{align*}
\left\vert \frac{1}{n}\sum_{i \in [2n]}I\{D_i=d\}\epsilon_{n,i}(d) - E[\epsilon_{n,i}(d)] \right\vert & \leq \left\vert \frac{1}{n}\sum_{i \in [2n]}I\{D_i=d\}(\epsilon_{n,i}(d) - E[\epsilon_{n,i}(d)|X_i]) \right\vert \\
& + \left\vert \frac{1}{n}\sum_{i \in [2n]}(I\{D_i=d\}-1/2)( E[\epsilon_{n,i}(d)|X_i] - E[\epsilon_{n,i}(d)]) \right\vert \\
&+ \left\vert \frac{1}{2n}\sum_{i \in [2n]}( E[\epsilon_{n,i}(d)|X_i] - E[\epsilon_{n,i}(d)]) \right\vert.
\end{align*}
The first two terms on the RHS of the above display are $O_P(1/\sqrt{n})$. The last term on the RHS is also $O_P(1/\sqrt{n})$ by Chebyshev's inequality. This implies the desired result. 

For the second result, define
	\begin{align*}
	\mathcal E_{n,0}(d)  = \begin{pmatrix}
	& \max_{d \in \{0, 1\}} \frac{1}{2n} \sum_{1 \leq i \leq 2n} E[\epsilon_{n, i}^4(d) | X_i] \leq c_0 < \infty~, \\
	& \min_{1 \leq l \leq p_n} \frac{1}{n} \sum_{1 \leq i \leq 2n} I \{D_i = d\} \var[\psi_{n, i, l} \epsilon_{n, i}(d) | X_i] \geq \ubar \sigma^2 > 0~,
	\end{pmatrix}
	\end{align*}

	\begin{align*}
	\mathcal E_{n,1}(d) = \left\{\left\Vert\frac{1}{n}\sum_{1 \leq i \leq 2n}I \{D_i = d\} (\psi_{n,i}\epsilon_{n, i}(d) - E[\psi_{n,i}\epsilon_{n, i}(d)|X_i]) \right\Vert_\infty \leq 2.04 \overline{\sigma} \sqrt{\log (2n p_n)/n}   \right\}~,
	\end{align*}
	\begin{align*}
	\mathcal E_{n,2}(d) = \left\{\left\Vert\frac{1}{n}\sum_{1 \leq i \leq 2n}I \{D_i = d\} ( E[\psi_{n,i}\epsilon_{n, i}(d)|X_i]-E[\psi_{n,i}\epsilon_{n, i}(d)]) \right\Vert_\infty \leq 3.96 \overline{\sigma} \sqrt{\log (2n p_n)/n}   \right\}~,
	\end{align*}
	\begin{align*}
	\mathcal E_{n,3}(d) = \left\{\max_{1 \leq l \leq p_n}\left \vert \frac{1}{2n} \sum_{1 \leq i \leq 2n} (E[\psi_{n,i,l}^2\epsilon_{n, i}^2(d)|X_i] - E[\psi_{n,i,l}^2\epsilon_{n, i}^2(d)])\right\vert^{1/2}  \leq 0.01 \overline{\sigma}\right\}~,
	\end{align*}
	and
	\begin{align*}
	\mathcal E_{n,4}(d) = \left\{\max_{1 \leq l \leq p_n}\left \vert \frac{1}{2n} \sum_{1 \leq i \leq 2n} (2I\{D_i = d\}-1) (E[\epsilon_{n, i}^2(d)\psi_{n, i,l}^2|X_i] - E[\epsilon_{n, i}^2(d)\psi_{n, i,l}^2] )\right\vert^{1/2}  \leq 0.01 \overline{\sigma}\right\}~.
	\end{align*}

	We aim to show that $P \{\mathcal E_{n,1}(d)\} \to 1$ and $P \{\mathcal E_{n,2}(d)\} \to 1$. Then, by letting $C = 6\bar{\sigma}/\ubar\sigma$ which implies
	\begin{align*}
	P \{\mathcal E_{n}(d)\} & = 1-P \{\mathcal E_{n}^c(d)\} \\
	& \geq 1- P \left \{ \left \|\frac{1}{n} \sum_{1 \leq i \leq 2n} I \{D_i = d\} (\psi_{n, i} \epsilon_{n, i}(d) - E[\psi_{n, i} \epsilon_{n, i}(d)]) \right \|_\infty \geq C\ubar\sigma\sqrt{\frac{\log (2n p_n)}{n}} \right \}\\
	& = 1- P \left \{ \left \|\frac{1}{n} \sum_{1 \leq i \leq 2n} I \{D_i = d\} (\psi_{n, i} \epsilon_{n, i}(d) - E[\psi_{n, i} \epsilon_{n, i}(d)]) \right \|_\infty \geq 6\bar{\sigma}\sqrt{\frac{\log (2n p_n)}{n}} \right \}\\
	& \geq 1 - P \{\mathcal E_{n,1}^c(d)\}-P\{\mathcal E_{n,2}^c(d)\} \to 1~.
	\end{align*}
	
	First, we show $P \{\mathcal E_{n,3}(d)\} \to 1$. Let
	\[ t_n = C\sqrt{\frac{\log(n p_n) \Xi_n^2 }{n}} \to 0 \]
	for some sufficiently large constant $C>0$ and $\{e_i\}_{1 \leq i \leq 2n}$ be a sequence of i.i.d.\ Rademacher random variables independent of everything else. Then, for any fixed $t>0$, we have
	\begin{align*}
	& \left(1 - \frac{4\max_{1 \leq l \leq p_n} \var[E[\psi_{n,i,l}^2\epsilon_{n, i}^2(d)|X_i]]}{2nt^2}\right)P\left\{\max_{1 \leq l \leq p_n}\left \vert \frac{1}{2n} \sum_{1 \leq i \leq 2n} \left[ E[\psi_{n,i,l}^2\epsilon_{n, i}^2(d)|X_i] - E[\psi_{n,i,l}^2\epsilon_{n, i}^2(d)]\right]\right\vert \geq t \right\} \\
	& \leq 2P\left\{\max_{1 \leq l \leq p_n}\left \vert \frac{1}{2n} \sum_{1 \leq i \leq 2n} 4e_iE[\psi_{n,i,l}^2\epsilon_{n, i}^2(d)|X_i] \right\vert  \geq t \right\} \\
	& = o(1)+2E \left[ P\left\{\max_{1 \leq l \leq p_n}\left \vert \frac{1}{2n} \sum_{1 \leq i \leq 2n} 4e_iE[\psi_{n,i,l}^2\epsilon_{n, i}^2(d)|X_i] \right\vert  \geq t \bigg|X^{(n)}\right\}I\{\mathcal{E}_{n,0}(d)\} \right] \\
	& \lesssim o(1) + p_n  \exp\left(-\frac{nt^2}{\Xi_n^2 c} \right) = o(1),
	\end{align*}
	where the first inequality is by \citet[Lemma 2.3.7]{vandervaart:wellner:1996}, the second inequality is by the Hoeffding's inequality conditional on $X^{(n)}$ and the fact that, on $\mathcal{E}_{n,0}(d)$, 
	\begin{align*}
	\frac{1}{2n}\sum_{1 \leq i \leq 2n} (E[\psi_{n,i,l}^2\epsilon_{n, i}^2(d)|X_i])^2 & \leq \frac{1}{2n}\sum_{1 \leq i \leq 2n} E[\psi_{n,i,l}^4|X_i] E[\epsilon_{n, i}^4(d)|X_i] \\
	& \leq \frac{\Xi_n^2}{2n}\sum_{1 \leq i \leq 2n} E[\psi_{n,i,l}^2|X_i] E[\epsilon_{n, i}^4(d)|X_i] \leq \Xi_n^2 Cc_0,
	\end{align*}
	where $C$ is a fixed constant, and the last equality is by the fact that $\log(p_n) \Xi_n^2 = o(n)$. Furthermore, we note that 
	\begin{align*}
	& \frac{4\max_{1 \leq l \leq p_n} \var[E[\psi_{n,i,l}^2\epsilon_{n, i}^2(d)|X_i]]}{2n} \\
	& \lesssim \frac{\max_{1 \leq l \leq p_n} E \left[ E[\psi_{n,i,l}^4|X_i]E[\epsilon_{n, i}^4(d)|X_i]\right]}{n} \\
	& \lesssim \frac{\Xi_n^2 \max_{1 \leq l \leq p_n} E \left[ E[\psi_{n,i,l}^2|X_i]E[\epsilon_{n, i}^4(d)|X_i]\right]}{n} \\
	& \lesssim \frac{\Xi_n^2 E[\epsilon_{n, i}^4(d)]}{n} \\
	& = o(1).
	\end{align*}
	Therefore, we have 
	\begin{align*}
	P\left\{\max_{1 \leq l \leq p_n}\left \vert \frac{1}{2n} \sum_{1 \leq i \leq 2n} \left[ E[\psi_{n,i,l}^2\epsilon_{n, i}^2(d)|X_i] - E[\psi_{n,i,l}^2\epsilon_{n, i}^2(d)]\right]\right\vert \geq t \right\} = o(1)
	\end{align*}
	for any fixed $t>0$, which is the desired result.

	Next, we show $P\{\mathcal E_{n,4}(d)\} \to 1$. Define $a_{n,i,l} = E[\epsilon_{n, i}^2(d)\psi_{n, i,l}^2|X_i] - E[\epsilon_{n, i}^2(d)\psi_{n, i,l}^2]$. Then, we have
	\begin{align*}
	& P \left \{\max_{1 \leq l \leq p_n}\left \vert \frac{1}{2n} \sum_{1 \leq i \leq 2n} (2I\{D_i = d\}-1) (E[\epsilon_{n, i}^2(d)\psi_{n, i,l}^2|X_i] - E[\epsilon_{n, i}^2(d)\psi_{n, i,l}^2] )\right\vert  >t\bigg|X^{(n)} \right \}I\{\mathcal{E}_{n,0}(d)\} \\
	& \leq \sum_{1 \leq l \leq p_n} P \left \{  \left\vert \frac{1}{2n}\sum_{1 \leq j \leq n}(I\{D_{\pi(2j-1)}=d\} - I\{ D_{\pi(2j)}=d\})(a_{n,\pi(2j-1),l} - a_{n,\pi(2j),l})\right\vert > t \bigg|X^{(n)} \right\}I\{\mathcal{E}_{n,0}(d)\} \\
	& \leq \sum_{1 \leq l \leq p_n} \exp \left ( - \frac{2n t^2}{\frac{1}{n} \sum_{1 \leq j \leq n} (a_{n, \pi(2j - 1),l} - a_{n, \pi(2j),l})^2} \right )I\{\mathcal{E}_{n,0}(d)\} \\
	& \leq \exp\left(\log(p_n) - \frac{2n t^2}{\Xi_n^2 c^2} \right)~,
	\end{align*}
	where, conditional on $X^{(n)}$, $\{I\{D_{\pi(2j-1)}=d\} - I\{ D_{\pi(2j)}=d\}\}_{1 \leq j \leq n}$ is a sequence of i.i.d. Rademacher random variables, the second last inequality is by Hoeffding's inequality, and the last inequality is by  that, on $\mathcal{E}_{n,0}(d)$, 
	\begin{align*}
	& \left(\frac{1}{n} \sum_{1 \leq j \leq n} (a_{n, \pi(2j - 1),l} - a_{n, \pi(2j),l})^2\right)^{1/2} \\
	& \leq \left(\frac{1}{n} \sum_{1 \leq j \leq n} (E[\psi_{n,\pi(2j-1),l}^2\epsilon_{n, i}^2(d)|X_{\pi(2j-1)}])^2\right)^{1/2} + \left(\frac{1}{n} \sum_{1 \leq j \leq n} (E[\psi_{n,\pi(2j),l}^2\epsilon_{n, i}^2(d)|X_{\pi(2j)}])^2\right)^{1/2} \\
	& \leq \left(\frac{2}{n} \sum_{1 \leq i \leq 2n} (E[\psi_{n,i,l}^2\epsilon_{n, i}^2(d)|X_i])^2\right)^{1/2} \\
	& \leq \Xi_n c~.
	\end{align*}
	By letting $t = C\sqrt{\frac{ \log(p_n) \Xi_n^2 }{n}}$ for some sufficiently large $C$ and noting that $P\{\mathcal{E}_{n,0}(d)\} \rightarrow 1$, we have 
	\begin{align*}
	\max_{1 \leq l \leq p_n}\left \vert  \sum_{1 \leq i \leq 2n} \frac{(2I\{D_i = d\}-1) (E[\epsilon_{n, i}^2(d)\psi_{n, i,l}^2|X_i] - E[\epsilon_{n, i}^2(d)\psi_{n, i,l}^2] )}{2n}\right\vert = O_p\left(\sqrt{\frac{ \log(p_n) \Xi_n^2 }{n}} \right)~,
	\end{align*}
	and thus, $P\{\mathcal E_{n,4}(d)\} \to 1$. 
	
	Next, we show $P \{\mathcal E_{n,1}(d)\} \to 1$. We note that, for $d \in \{0, 1\}$, conditional on $(D^{(n)},X^{(n)})$, $\{\psi_{n,i}\epsilon_{n, i}(d)\}_{1 \leq i \leq 2n}$ are independent. In what follows, we couple
	\[ \mathbb U_n = \frac{1}{n}\sum_{1 \leq i \leq 2n}I \{D_i = d\} (\psi_{n,i}\epsilon_{n, i}(d) - E[\psi_{n,i}\epsilon_{n, i}(d)|X_i]) \]
	with a centered Gaussian random vector as in Theorem 2.1 in \cite{chernozhukov2017central}. Let $Z = (Z_1, \ldots, Z_{p_n})$ be a Gaussian random vector with $E[Z_l] = 0$ for $1 \leq l \leq p_n$ and $\var[Z] = \var[\mathbb U_n | X^{(n)}, D^{(n)}]$ that additionally satisfies the conditions of that theorem. Specifically, $Z = (Z_1,\cdots,Z_{p_n})$ is a centered Gaussian random vector in $R^{p_n}$ such that on $\mathcal E_{n,0}(d) \cap \mathcal E_{n,3}(d) \cap \mathcal E_{n,4}(d)$, 
	\begin{align*}
	E[ZZ'] & = \frac{1}{n^2}\sum_{1 \leq i \leq 2n} I\{D_i = d\} E[\epsilon_{n, i}^2(d)\psi_{n, i}\psi_{n, i}'|X_i] \\
	& - \frac{1}{n}\left(\frac{1}{n}\sum_{1 \leq i \leq 2n} I\{D_i = d\} E[\epsilon_{n, i}(d)\psi_{n, i}|X_i]\right)\left(\frac{1}{n}\sum_{1 \leq i \leq 2n} I\{D_i = d\} E[\epsilon_{n, i}(d)\psi_{n, i}|X_i]\right)'
	\end{align*}
	and 
	\begin{align*}
	\max_{1 \leq l \leq p_n}E[Z_l^2] & \leq \frac{\max_{1 \leq l \leq p_n}\sum_{1 \leq i \leq 2n} I\{D_i = d\} E[\epsilon_{n, i}^2(d)\psi_{n, i,l}^2|X_i]}{n^2} \\
	& \leq \frac{\overline{\sigma}^2}{n} + \frac{\max_{1 \leq l \leq p_n}\sum_{1 \leq i \leq 2n} I\{D_i = d\} (E[\epsilon_{n, i}^2(d)\psi_{n, i,l}^2|X_i] - E[\epsilon_{n, i}^2(d)\psi_{n, i,l}^2] )}{n^2} \\
	& \leq \frac{\overline{\sigma}^2}{n} + \frac{\max_{1 \leq l \leq p_n}\sum_{1 \leq i \leq 2n} (2I\{D_i = d\}-1) (E[\epsilon_{n, i}^2(d)\psi_{n, i,l}^2|X_i] - E[\epsilon_{n, i}^2(d)\psi_{n, i,l}^2] )}{2n^2}
	\\
	& + \frac{\max_{1 \leq l \leq p_n}\sum_{1 \leq i \leq 2n}  (E[\epsilon_{n, i}^2(d)\psi_{n, i,l}^2|X_i] - E[\epsilon_{n, i}^2(d)\psi_{n, i,l}^2] )}{2n^2}\\
	& \leq \frac{1.02\overline{\sigma}^2}{n}~. 
	\end{align*}
	Further define $q(1-\alpha)$ as the $(1-\alpha)$ quantile of $||Z||_\infty$. Then, we have 
	\begin{align*}
	q(1-1/n) \leq \frac{1.02\overline{\sigma}(\sqrt{2 \log (2p_n)} + \sqrt{2 \log(n)})}{\sqrt{n}} \leq 2.04\overline{\sigma}\sqrt{\log (2n p_n)/n},  
	\end{align*}
	where the first inequality is by the last display in the proof of Lemma E.2 in \cite{chetverikov2022analytic} and the second inequality is by the fact that $\sqrt{a} + \sqrt{b} \leq \sqrt{2(a+b)}$ for $a,b>0$. Therefore, we have
	\begin{align*}
	P \{\mathcal E_{n,1}^c(d)) & \leq  P \{\mathcal E_{n,1}^c(d),  \mathcal E_{n,0}(d), \mathcal E_{n,3}(d), \mathcal E_{n,4}(d) \}  + o(1) \\
	& = E  P \{\mathcal E_{n,1}^c(d)| D^{(n)}, X^{(n)}\} I\{\mathcal E_{n,0}(d), \mathcal E_{n,3}(d), \mathcal E_{n,4}(d)\} + o(1) \\
	& \leq E [ P \{||Z||_\infty \geq 2.04 \overline{\sigma} \sqrt{\log (2n p_n)/n} | D^{(n)}, X^{(n)})\} I\{\mathcal E_{n,0}(d), \mathcal E_{n,3}(d), \mathcal E_{n,4}(d)\}] + o(1)\\
	& \leq E [P \{ ||Z||_\infty \geq q(1-1/n) | D^{(n)}, X^{(n)} \}]= o(1),
	\end{align*}
	where the second inequality is by Theorem 2.1 in \cite{chernozhukov2017central}. 
	
	Finally, we turn to $\mathcal E_{n,2}(d)$ with $d=1$. We have 
	\begin{align}
	& \frac{1}{n}\sum_{1 \leq i \leq 2n}I \{D_i = 1\} ( E[\psi_{n,i}\epsilon_{n, i}(1)|X_i]-E[\psi_{n,i}\epsilon_{n, i}(1)]) \notag \\
	& \label{eq:E2'} = \frac{1}{2n}\sum_{1 \leq i \leq 2n}( E[\psi_{n,i}\epsilon_{n, i}(1)|X_i]-E[\psi_{n,i}\epsilon_{n, i}(1)]) + \frac{1}{2n}\sum_{1 \leq i \leq 2n}(2D_i-1)( E[\psi_{n,i}\epsilon_{n, i}(1)|X_i]-E[\psi_{n,i}\epsilon_{n, i}(1)]).
	\end{align}
	Note $\{ E[\psi_{n,i}\epsilon_{n, i}(1)|X_i]-E[\psi_{n,i}\epsilon_{n, i}(1)]\}_{1 \leq i \leq 2n}$ is a sequence of independent centered random variables and
	\begin{align*}
	\max_{1 \leq l \leq p_n} E [(E[\psi_{n,i,l}\epsilon_{n, i}(1)|X_i]-E[\psi_{n,i,l}\epsilon_{n, i}(1)])^2] \leq \overline{\sigma}^2. 
	\end{align*}
	Following Theorem 2.1 in \cite{chernozhukov2017central}, Lemma E.2 in \cite{chetverikov2022analytic}, and similar arguments to the ones above, we have 
	\begin{align}
	P \left \{\left\Vert \frac{1}{2n}\sum_{1 \leq i \leq 2n}( E[\psi_{n,i}\epsilon_{n, i}(1)|X_i]-E[\psi_{n,i}\epsilon_{n, i}(1)])\right\Vert_\infty \leq  \overline{\sigma} \sqrt{2\log (2n p_n)/n}\right) \to 1.
	\label{eq:E21'}
	\end{align}
	For the second term on the RHS of \eqref{eq:E2'}, we define $g_{n,i,l} = E[\psi_{n,i,l}\epsilon_{n, i}(1)|X_i]-E[\psi_{n,i,l}\epsilon_{n, i}(1)]$. We have 
	\begin{align*}
	& P \left \{ \left\Vert \frac{1}{2n}\sum_{1 \leq i \leq 2n}(2D_i-1)E[\psi_{n,i}\epsilon_{n, i}(1)|X_i]-E[\psi_{n,i}\epsilon_{n, i}(1)]\right\Vert_\infty > t \bigg|X^{(n)}\right\} \\
	& \leq \sum_{1 \leq l \leq p_n} P \left \{  \left\vert \frac{1}{2n}\sum_{1 \leq j \leq n}(D_{\pi(2j-1)} - D_{\pi(2j)})(g_{n,\pi(2j-1),l} - g_{n,\pi(2j),l})\right\vert > t \bigg|X^{(n)} \right\} \\
	& \leq \sum_{1 \leq l \leq p_n} \exp \left ( - \frac{2n t^2}{\frac{1}{n} \sum_{1 \leq j \leq n} (g_{n, \pi(2j - 1),l} - g_{n, \pi(2j),l})^2} \right )~,
	\end{align*}
	where, conditional on $X^{(n)}$, $\{(D_{\pi(2j-1)} - D_{\pi(2j)})\}_{1 \leq j \leq n}$ is a sequence of i.i.d. Rademacher random variables and the last inequality is by Hoeffding's inequality. In addition, on $\mathcal E_{n,3}(1)$, we have 
	\begin{align*}
	& \left(\frac{1}{n} \sum_{1 \leq j \leq n} (g_{n, \pi(2j - 1),l} - g_{n, \pi(2j),l})^2\right)^{1/2} \\
	& \leq \left(\frac{1}{n} \sum_{1 \leq j \leq n} (E[\psi_{n,\pi(2j-1),l}\epsilon_{n, i}(1)|X_{\pi(2j-1)}])^2\right)^{1/2} + \left(\frac{1}{n} \sum_{1 \leq j \leq n} (E[\psi_{n,\pi(2j),l}\epsilon_{n, i}(1)|X_{\pi(2j)}])^2\right)^{1/2} \\
	& \leq \left(\frac{2}{n} \sum_{1 \leq i \leq 2n} (E[\psi_{n,i,l}\epsilon_{n, i}(1)|X_i])^2\right)^{1/2} \\
	& \leq \left(\frac{2}{n} \sum_{1 \leq i \leq 2n} E[\psi_{n,i,l}^2\epsilon_{n, i}^2(1)|X_i]\right)^{1/2} \\
	& \leq \left(\frac{2}{n} \sum_{1 \leq i \leq 2n} \left[ E[\psi_{n,i,l}^2\epsilon_{n, i}^2(1)|X_i] - E[\psi_{n,i,l}^2\epsilon_{n, i}^2(1)]\right]\right)^{1/2} + 2 \overline{\sigma} \\
	& \leq 2.02\overline{\sigma}.
	\end{align*}
	Therefore, we have
	\begin{align}
	& P \left \{\left\Vert \frac{1}{2n}\sum_{1 \leq i \leq 2n}(2D_i-1)E[\psi_{n,i}\epsilon_{n, i}(1)|X_i]-E[\psi_{n,i}\epsilon_{n, i}(1)]\right\Vert_\infty > 2.02\sqrt{\frac{\log(n p_n) \overline{\sigma}^2}{n}}\right\} \notag  \\
	&\leq  P \left \{ \left\Vert \frac{1}{2n}\sum_{1 \leq i \leq 2n}(2D_i-1)E[\psi_{n,i}\epsilon_{n, i}(1)|X_i]-E[\psi_{n,i}\epsilon_{n, i}(1)]\right\Vert_\infty >  2.02\sqrt{\frac{\log(n p_n) \overline{\sigma}^2}{n}}, \mathcal E_{n,3}(1) \right\} + o(1) \notag  \\
	& \leq E \left[P \left \{\left\Vert \frac{1}{2n}\sum_{1 \leq i \leq 2n}(2D_i-1)E[\psi_{n,i}\epsilon_{n, i}(1)|X_i]-E[\psi_{n,i}\epsilon_{n, i}(1)]\right\Vert_\infty >  2.02\sqrt{\frac{\log(n p_n) \overline{\sigma}^2}{n}} \bigg|X^{(n)}\right\}I\{\mathcal E_{n,3}(1) \}\right] + o(1) \notag\\
	& = o(1)~.
	\label{eq:E22'}
	\end{align}
	Combining \eqref{eq:E2'}, \eqref{eq:E21'}, \eqref{eq:E22'}, and the fact that $\sqrt{2}+2.02 \leq 3.98$, we have $P\{\mathcal E_{n,2}(1)\} \to 1$. The same result holds for $\mathcal E_{n,2}(0)$. 
\end{proof}

\section{Details for Simulations} \label{sec:sims-details}
The regressors in the LASSO-based adjustment are as follows.
		\begin{enumerate}[(i)]
			\item For Models 1-6, we use $\{1,X_{i},W_{i},X_{i}^2,W_{i}^2,X_{i}W_{i},(X_{i}-\tilde{X})I\{X_{i}>\tilde{X}\}, (W_{i}-\tilde{W})I\{W_{i}>\tilde{W}\},(X_{i}-\tilde{X})^2 I\{X_{i}>\tilde{X}\}, (W_{i}-\tilde{W})^{2}I\{W_{i}>\tilde{W}\}\}$ where $\tilde{X}$ and $\tilde{W}$ are the sample medians of $\{X_{i}\}_{i \in [2n]}$ and $\{W_{i}\}_{i \in [2n]}$, respectively.
			\item For Models 7-9, we use $\{1,X_{i},W_{i},X_{i}^2,W_{i}^2,X_{i1}W_{i1},X_{i2}W_{i1},X_{i1}W_{i2},X_{i2}W_{i2},(X_{ij}-\tilde{X_{j}})I\{X_{ij}>\tilde{X_{j}}\},(X_{ij}-\tilde{X_{j}})^{2}I\{X_{ij}>\tilde{X_{j}}\},(W_{ij}-\tilde{W_{1}})I\{W_{ij}>\tilde{W_{j}}\},(W_{ij}-\tilde{W_{j}})^{2}I\{W_{ij}>\tilde{W_{j}}\}\}$ where $\tilde{X_{j}}$ and $\tilde{W_{j}}$, for $j=1,2$, are the sample medians of $\{X_{ij}\}_{i \in [2n]}$ and $\{W_{ij}\}_{i \in [2n]}$, respectively.
			\item For Models 10-11, we use $\{1,X_{i},W_{i},X_{i}^2,W_{i}^2,X_{i1}W_{i1},X_{i2}W_{i2},X_{i3}W_{i1},X_{i4}W_{i2}, (X_{ij}-\tilde{X_{j}})I\{X_{ij}>\tilde{X_{j}}\},(X_{ij}-\tilde{X_{j}})^{2}I\{X_{ij}>\tilde{X_{j}}\},(W_{ij}-\tilde{W_{j}})I\{W_{ij}>\tilde{W_{j}}\},(W_{ij}-\tilde{W_{j}})^{2}I\{W_{ij}>\tilde{W_{j}}\}\}$ where $\tilde{X_{j}}$,for $j=1,2,3,4$, and $\tilde{W_{j}}$, for $j=1,2$, are the sample medians of $\{X_{ij}\}_{i \in [2n]}$ and $\{W_{ij}\}_{i \in [2n]}$, respectively.
			\item Models 12-15 already contain high-dimensional covariates. We just use $X_i$ and $W_i$ as the LASSO regressors.
	\end{enumerate}

\bibliography{wolf}

\end{document}